\documentclass[a4wide,12pt]{article}

\usepackage{a4wide}

\usepackage{amsmath, amsfonts}
\usepackage{algorithm}
\usepackage{tikz}
\usepackage{url}
\newcounter{counter:k-attribute}
\newcounter{counter:1-attribute}
\newcounter{counter:1-attr-rot}
\newcounter{counter:k-Euclidean}
\newcounter{counter:save}

\newtheorem{theorem}{Theorem}
\newtheorem{lemma}[theorem]{Lemma}
\newtheorem{observation}[theorem]{Observation}

\newtheorem{definition}{Definition}[section]

\newenvironment{proof}{\noindent {\bf Proof:}\quad}
{\hfill$\Box$ \par \medskip}

\newcommand{\From}{From\ }

\def\RHPi{\mathrm{\#RH}\Pi_1}
\def\Sat{\textsc{Sat}}
\def\poly{\mathop{\mathrm{poly}}}
\def\APred{\leq_\mathrm{AP}}
\def\BIS{$\#BIS$}

\begin{document}
\title{{\Large \bf The Complexity of Approximately Counting
                   Stable Matchings\footnotetext{A preliminary version of this
                   paper appeared in {\em APPROX 2010,
                   Lecture Notes in Computer Science} {\bf 6302}, Springer, pp.\ 81--94.}}}
\author{Prasad Chebolu\footnote{Department of Computer Science, Ashton Bldg, Ashton St,
     University of Liverpool, Liverpool L69 3BX, United
     Kingdom.
}\
     \footnote{Research supported in part by EPSRC Grant EP/F020651/1.}
\and Leslie Ann Goldberg\footnotemark[1]
\and Russell Martin\footnotemark[1]\ \footnotemark[2]}

\date{}

\maketitle

\begin{abstract}
We investigate the complexity of {\em approximately counting} stable
matchings in the   $k$-attribute model,
where the preference lists are determined by dot products of
``preference vectors'' with ``attribute vectors'', or by
Euclidean distances between ``preference points`` and ``attribute points''.
Irving and Leather~\cite{IL} proved that counting the number of
stable matchings in the general case is $\#P$-complete.  Counting
the number of stable matchings is reducible to counting the number of
downsets in a (related) partial order~\cite{IL} and is interreducible,
in an approximation-preserving sense, to a class of problems that includes
counting the number of
independent sets in a bipartite graph ($\#BIS$)~\cite{DGGJ}.
 It is
conjectured that no FPRAS exists for this class of problems.
We show this approximation-preserving interreducibilty
remains even in the restricted $k$-attribute setting when $k \geq 3$ (dot
products) or $k \geq 2$ (Euclidean distances).  Finally, we show it is
easy to count the number of stable matchings in the $1$-attribute dot-product
setting.
\end{abstract}

\section{Introduction}\label{sect:introduction}

\subsection{Stable Matchings}\label{sect:intro-SM}
The {\em stable matching problem} (or {\em stable marriage problem}) is
a classical
combinatorics problem.  An instance of this problem consists
of $n$ men and $n$ women, where each man has his own
preference list (a total ordering) of the women, and, similarly,
each woman has her own preference list of the men.
A one-to-one pairing of the men with the
women is called a {\em matching} (or {\em marriage}).
Given a matching, if there exists a man $M$ and a woman $w$
in the matching who prefer each other over their partners in the
matching, then the matching is considered {\em unstable} and the
man-woman pair $(M,w)$ is called a {\em blocking pair}. ($M$ and
$w$ would prefer to drop their current partners and pair up with
each other.)  If a matching has no blocking pairs,
then we call it a {\em stable matching}. In 1962, Gale and Shapley
proved that every stable
matching instance has a stable matching, and described
an $O(n^2)$ algorithm for finding one~\cite{GS}.

The stable matching problem has many variants, where ties
in the preference lists could be allowed, where people
might have partial preference lists (i.e.\ someone might
prefer to remain single rather than be paired with certain
members of the opposite sex), generalizations to men/women/pets,
universities and applicants, students and projects, etc.
Some of these generalizations have
also been well-studied and, indeed, algorithms for finding
stable matchings are used for assigning residents to hospitals in
Scotland, Canada, and the USA~\cite{CRMS, NRMP, SFAS}.

In this paper, we concentrate solely
on the classical problem, so the term ``matching instance'' will
refer to one where the number of men is equal to the number of women,
and each man or women has their own full
totally-ordered (i.e.\ no ties allowed) preference list for the opposite
sex.

Irving and Leather~\cite{IL} demonstrated that counting the {\em number}
of stable matchings for a given instance is $\#P$-complete.  This
completeness result relies on the connection between stable
marriages and {\em downsets} in a related partial order (explained
in more detail in Section~\ref{sect:combinatorics}), as counting the
number of downsets in a partial order is another classical $\#P$-complete
problem~\cite{PB}.

Knowing that exactly counting stable matchings is difficult (under
standard complexity-theoretic assumptions), one might
turn to methods for {\em approximately} counting this number.  In
particular, we would like to find a {\em fully-polynomial randomized
approximation scheme} (an {\em FPRAS})
for this task, i.e.\ an algorithm
that provides an arbitrarily close approximation in time polynomial
in the input size and the desired error --- see Section~\ref{sect:AP} for a formal definition.
One method that has proven
successful for other counting problems is the Markov Chain Monte
Carlo (MCMC) method.  This technique exploits a
relationship between counting and sampling described by Jerrum, Valiant,
and Vazirani~\cite{JVV}, namely, for {\em self-reducible} combinatorial
structures, the existence of an FPRAS is computationally equivalent to
a polynomial-time algorithm for approximate sampling from the set of
structures.  Although the set of stable matchings for an instance
does not obviously fit into the class of self-reducible problems,
an efficient algorithm for (approximately) sampling a random
stable matching can be transformed into a method for (approximately)
counting this number.

Bhatnagar, Greenberg, and Randall~\cite{Randall} considered this problem
of sampling a random stable matching using the MCMC method.
They examined a natural Markov chain that uses ``male-improving''
and ``female-improving'' {\em rotations} (see Section~\ref{sect:rotations})
to define a random walk on the state space of stable matchings for a given
instance.  In the most general setting, matching instances can be exhibited for
which the {\em mixing time} of the random walk has an exponential lower
bound, meaning that it will take an exponential amount of time to
(approximately) sample a random stable matching.  This exponential mixing time is
due to the existence of a ``bad cut'' in the state space.
Bhatnagar, et al.\ considered several restricted settings for matching
instances and were still able to show instances for which such a bad
cut exists in the state space, implying an exponential mixing time
in these restricted settings.

Of particular interest to us in this paper, Bhatnagar et al.\ examined
the so-called {\em $k$-attribute model}.  In this setting each man
and woman has two $k$-dimensional vectors associated with them, a
``preference'' vector and a ``position'' (or ``attribute'') vector.
A man $M_i$ has a preference vector denoted by $\hat{M}_i$,
and a position vector denoted by $\bar{M}_i$ (similarly denoted
for the woman $w_j$).  Then, $M_i$ prefers $w_j$ over
$w_\ell$
(i.e.\
$w_j$ appears higher on his preference list than
$w_\ell$) if and
only if
$\hat{M}_i \cdot \bar{w}_j > \hat{M}_i\cdot \bar{w}_\ell$,
where
$\hat{M}_i \cdot \bar{w}_j$ denotes the usual $k$-dimensional dot
product of vectors.  Since we assume that each man has a total order
over the women (and vice-versa), we note that
$\hat{M}_i \cdot \bar{w}_j \not= \hat{M}_i \cdot \bar{w}_\ell$
whenever
$j \not= \ell$
(and analogously for the women's preference vectors/men's position vectors).

Even in this restricted $k$-attribute setting (not every matching
instance can be represented in this manner if $k$ is small~\cite{BL}),
Bhatnagar, Greenberg, and Randall were still able to demonstrate
examples of matching instances having a ``bad cut'' where the
Markov chain has an exponential mixing time.  Bhatnagar et al.\
also considered two other restricted settings, the so-called
{\em $k$-range} and {\em $k$-list} models, but we will not be
considering those cases here.  (Again, they gave instances
having an exponential mixing time for the Markov chain.)

It must be noted that even though the male-improving/female-improving
Markov chain might have an exponential mixing time, this does not
necessarily imply the non-existence of an FPRAS for the corresponding
counting problems.  However, Dyer et al.~\cite{DGGJ}
give evidence suggesting that even {\em approximately} counting the
number of stable matchings is itself difficult, i.e.\
suggesting that an FPRAS may not exist.  They do this
by demonstrating {\em approximation-preserving reductions} amongst several
counting problems, one being that of counting downsets in a partial
order (once again, the connection to stable matchings is outlined in
Section~\ref{sect:combinatorics}).  Relevant background
about approximation-preserving reductions is discussed
in Section~\ref{sect:AP}.  The main point is that the existence of
an FPRAS for one problem would imply the existence of an FPRAS for this
entire class of counting problems.  Currently, the existence of such an FPRAS
remains an open question.

It is precisely the goal of this
paper to consider the complexity of the approximate counting problem
for the $k$-attribute model.

Before we continue, let us formally define some counting problems.
Two counting problems relevant to us are:

{
\noindent
{\em Name:} $\#SM$.

\parskip 0.5ex

\noindent
{\em Instance:} A stable matching instance with $n$ men and $n$ women.

\noindent
{\em Output:} The number of stable matchings.

\bigskip

\noindent
{\em Name:} $\#SM(k\textrm{-attribute})$.

\noindent
{\em Instance:} A stable matching instance with $n$ men and $n$ women
in the $k$-attribute setting, i.e.\ preference lists are determined
using dot products between $k$-dimensional preference and
position vectors as mentioned above.

\noindent
{\em Output:} The number of stable matchings.
}

\medskip

As we stated previously, if $k$ is small (relative to $n$), there exist
preference lists that are not realizable in the $k$-attribute setting~\cite{BL}.
On the other hand, if $k = n$ then we can clearly represent any set
of $n$ preference lists by simply using a separate coordinate for each
person to rank the members of the opposite sex.

Another counting problem we consider in this paper is the following one:

{

\noindent
{\em Name:} $\#SM(k\textrm{-Euclidean})$.

\parskip 0.5ex
\noindent
{\em Instance:} A stable matching instance with $n$ men and $n$ women
in the $k$-dimensional Euclidean setting.  In this setting, men and
women each have a ``preference point'' and ``position point''.  Preference
lists are determined using Euclidean distances between preference points
and position points.

\noindent
{\em Output:} The number of stable matchings.
}

\medskip

In other words, for a $k$-Euclidean stable matching instance man
$M_i$ prefers woman $w_j$ to woman
$w_\ell$
if and only if
$d(\hat{M}_i, \bar{w}_j) < d(\hat{M}_i, \bar{w}_\ell)$,
where
$d(x,y)$ is the Euclidean
distance between points $x$ and $y$.  Once again, ties
are not allowed in the preference lists.

Before we describe our results, let us give a brief introduction to
approximation-preserving (AP) reductions and AP-reducibility.  Further
details can be found in Section~\ref{sect:AP}.

\subsection{AP-reducibility (a brief introduction)}\label{sect:AP-short}

Approximate counting problems have been of increasing interest
in recent years.  Some success has been demonstrated by finding
fully-polynomial randomized approximation schemes for some
$\#P$-complete problems.  Likewise, there are some (but fewer)
problems known to
not admit an FPRAS under usual complexity-theoretic assumptions.

AP-reducibility (for {\em approximation-preserving reducibility}) is
similar in nature to reductions used in showing problems are
$NP$-complete.  Broadly speaking, if $g$ is an integer-valued
function (that counts some type of combinatorial structure) and
there is an\\ approximation-preserving reduction from another
integer-valued function $f$ (counting something else) to $g$, then
an FPRAS for $g$ gives us an FPRAS for $f$.  (Similar to Turing
reductions, the problem sizes are polynomially related and the error
terms of the approximations are also polynomially related.) In this
case we would write $f \leq_{AP} g$ to mean that $f$ has an
AP-reduction to $g$.  Similarly we write $f \equiv_{AP} g$ to mean
that $f \leq_{AP} g$ {\em and} $g \leq_{AP} f$, or that $f$ and $g$
are AP-interreducible. Definitions are provided in
Section~\ref{sect:AP}.

This kind of AP-reduction allows us to study the relative complexity
of approximate counting problems, just as
polynomial many-one
reductions allow us
to compare complexity of decision problems such as graph coloring and
satisfiability.

The complexity class $\RHPi$ of counting problems was
introduced by Dyer, Goldberg, Greenhill and Jerrum~\cite{DGGJ} as a
means to classify a wide class of approximate counting problems
that were previously of indeterminate computational complexity.
The problems in $\RHPi$ are those that can be expressed in terms of
counting the number of models of a logical formula from a
certain syntactically restricted class.  Although the authors
were not aware of it at the time, this syntactically
restricted class had already been studied under the title
``restricted Krom SNP''~\cite{Dalmau05}. The complexity class
$\RHPi$ has a completeness class (with respect to AP-reductions)
which includes a wide and ever-increasing range of natural
counting problems, including: independent sets in a bipartite
graph, downsets in a partial order, configurations in
the Widom-Rowlinson model (all~\cite{DGGJ}) and the partition
function of the ferromagnetic Ising model with mixed external
field~\cite{ising}. Either all of these problems have an
FPRAS, or none do.  No FPRAS is currently known for any of them,
despite much effort having been expended on finding one.

All the problems in the completeness class mentioned above are inter-reducible
via AP-reductions, so  any
of them could be said to exemplify the completeness class.
However, mainly for historical reasons,
the following   problem tends to be taken as a key example in the class,
much in the same way that $\Sat$ has
a privileged status in the theory on NP-completeness.

{

\noindent
{\em Name:} $\#BIS$.

\parskip 0.5ex

\noindent
{\em Instance:} A bipartite graph $B$.

\noindent
{\em Output:} The number of independent sets in $B$.
}

Ge and \v Stefankovi\v c~\cite{BISpoly}
recently proposed an interesting new MCMC algorithm for sampling independent sets
in bipartite graphs. Unfortunately, however, the relevant Markov chain mixes slowly~\cite{slow}
so even this interesting new idea does not give an FPRAS for $\#BIS$.
In fact, Goldberg and Jerrum~\cite{FerroPotts} conjecture that no
FPRAS exists for $\#BIS$ (or for the other problems in the
completeness class). We make this conjecture on empirical grounds,
namely that the problem has survived its first decade
despite considerable efforts to find an FPRAS and
the collection of known \BIS-equivalent problems is growing.

Since Dyer et al.\ show that $\#BIS$ and
counting downsets are both complete in this class,
and it is known that counting downsets is equivalent to counting stable matchings,
the result of Dyer et al.\ implies  $\#BIS \equiv_{AP} \#SM$.

The goal of this paper is to demonstrate AP-interreducibility of $\#BIS$
with the two restricted stable matching problems defined in
Section~\ref{sect:intro-SM}.

\subsection{Our results}\label{sect:results}

In this paper we prove the following results:

\setcounter{counter:k-attribute}{\value{theorem}}
\begin{theorem}\label{thm:3-attr-dot}
$\#BIS \equiv_{AP} \#SM(k\rm{-attribute})$ when $k \geq 3$.
\end{theorem}

In other words, $\#BIS$ is AP-interreducible with counting
stable matchings in the $k$-attribute setting when $k\geq 3$,
so this problem is equivalent in terms of approximability to the complete problems in  the
complexity class $\RHPi$.

\setcounter{counter:1-attribute}{\value{theorem}}
\begin{theorem}\label{thm:1-attribute}
$\#SM(1\rm{-attribute})$ is solvable in polynomial time.
\end{theorem}

We can also prove AP-interreducibility
with $\#BIS$
in the $k$-Euclidean
setting (when $k \geq 2$) in a similar manner.
Recall that in the $k$-Euclidean
setting, preference lists are determined by (closest)
Euclidean distances between the ``preference points''
and ``position points''.

\setcounter{counter:k-Euclidean}{\value{theorem}}
\begin{theorem}\label{thm:k-Euclidean}
$\#BIS \equiv_{AP} \#SM(k\rm{-Euclidean})$ when $k \geq 2$.
\end{theorem}

The rest of the paper is laid out as follows:

We review further background on approximation schemes and AP-reductions in
Section~\ref{sect:AP}.

Section~\ref{sect:combinatorics} reviews some combinatorics
of the stable matching problem that is relevant for our purposes in
this paper.

Section~\ref{sect:k-attribute} demonstrates Theorem~\ref{thm:3-attr-dot}
and Section~\ref{sect:1-attribute} is devoted to proving
Theorem~\ref{thm:1-attribute}.

We give the construction required to demonstrate
Theorem~\ref{thm:k-Euclidean} in Section~\ref{sect:k-Euclidean}.
This construction ends up giving us identical preference lists as
those for the $k$-attribute ($k \geq 3$) model. Thus, the remainder
of the proof to show AP-interreducibility between\\
$\#SM(k\textrm{-Euclidean})$ for $k \geq 2$ and $\#BIS$ is identical
to that for the $3$-attribute setting and is not repeated.

\section{Randomized Approximation Schemes and\\ Approximation-preserving reductions}\label{sect:AP}

A \emph{randomized approximation scheme\/} is an algorithm for
approximately computing the value of a function~$f:\Sigma^*\rightarrow
\mathbb{R}$.
The
approximation scheme has a parameter~$\varepsilon>0$ which specifies
the error tolerance.
A \emph{randomized approximation scheme\/} for~$f$ is a
randomized algorithm that takes as input an instance $ x\in
\Sigma^{* }$ (e.g., for the problem $\#SM$, the
input would be an encoding of a stable matching instance)
and a rational error
tolerance $\varepsilon >0$, and outputs a rational number $z$
(a random variable of the ``coin tosses'' made by the algorithm)
such that, for every instance~$x$,
\begin{equation}
\label{eq:3:FPRASerrorprob}
\Pr \big[e^{-\epsilon} f(x)\leq z \leq e^\epsilon f(x)\big]\geq \frac{3}{4}\, .
\end{equation}
The randomized approximation scheme is said to be a
\emph{fully polynomial randomized approximation scheme},
or \emph{FPRAS},
if it runs in time bounded by a polynomial
in $ |x| $ and $ \epsilon^{-1} $.
Note that the quantity $3/4$ in
Equation~(\ref{eq:3:FPRASerrorprob})
could be changed to any value in the open
interval $(\frac12,1)$ without changing the set of problems
that have randomized approximation schemes~\cite[Lemma~6.1]{JVV}.

We now define the notion of an approximation-preserving (AP) reduction.
Suppose that $f$ and $g$ are functions from
$\Sigma^{* }$ to~$\mathbb{R}$. As mentioned before,
an AP-reduction from~$f$ to~$g$ gives a way to turn an FPRAS for~$g$
into an FPRAS for~$f$. Here is the formal definition. An {\it approximation-preserving reduction\/}
from $f$ to~$g$ is a randomized algorithm~$\mathcal{A}$ for
computing~$f$ using an oracle for~$g$. The algorithm~$\mathcal{A}$ takes
as input a pair $(x,\varepsilon)\in\Sigma^*\times(0,1)$, and
satisfies the following three conditions: (i)~every oracle call made
by~$\mathcal{A}$ is of the form $(w,\delta)$, where
$w\in\Sigma^*$ is an instance of~$g$, and $0<\delta<1$ is an
error bound satisfying $\delta^{-1}\leq\poly(|x|,
\varepsilon^{-1})$; (ii) the algorithm~$\mathcal{A}$ meets the
specification for being a randomized approximation scheme for~$f$
(as described above) whenever the oracle meets the specification for
being a randomized approximation scheme for~$g$; and (iii)~the
run-time of~$\mathcal{A}$ is polynomial in $|x|$ and
$\varepsilon^{-1}$.

According to the definition, approximation-preserving reductions may use randomization
and may make multiple oracle calls. Nevertheless,
the reductions that we present in this paper are deterministic.
Each reduction makes a single oracle call (with $\delta=\epsilon$) and returns the
result of that oracle call.
A word of warning about terminology:
Subsequent to~\cite{DGGJ}, the notation $\APred$ has been
used
to denote a different type of approximation-preserving reduction which applies to
optimization problems.
We will not study optimization problems in this paper, so hopefully this will
not cause confusion.

\section{Combinatorics of the stable matching problem}\label{sect:combinatorics}

The (classical) stable matching problem has a rich combinatorial
structure which has been widely studied.  We relate some aspects
of this structure that we will need in this paper.  Many of the
definitions and results that follow can be found, for example,
in~\cite{Knuth, IL, GI, Gusfield}.

\subsection{The Gale-Shapley algorithm}\label{sect:GS}

In their seminal paper on the stable matching problem, Gale and
Shapley~\cite{GS} gave a polynomial-time algorithm for constructing
a stable matching.  This is generally referred to as
the ``proposal algorithm'' and bears the names of Gale and Shapley
in all of the literature on stable matchings.  One sex (typically the men) make
proposals to members of the other, forming ``engagements''.
Once all the ``proposers'' are engaged, the algorithm terminates
with a stable matching.

A description of this algorithm follows.

\begin{algorithm}
\caption{Gale-Shapley Algorithm}\label{GS}
\begin{itemize}
\item {Initially every man and every woman is free.}

\item {Repeat until all men are engaged:}

\item {A free man $M$ proposes to $w$, the highest woman on his list
    who has not already rejected him.
       \begin{itemize}
       \item {If $w$ is free, then she accepts the proposal and $M$ and $w$
       become engaged.}
       \item{ If $w$ is engaged to $M'$, then
              \begin{itemize}
              \item Let $M^+$ be the favorite of $w$ between men
              $M$ and $M'$.
              \item Let $M^-$ be the least favorite of $w$ between men
              $M$ and $M'$.
              \item $M^+$ and $w$ become engaged to each other.
              \item $w$ rejects $M^-$ and $M^-$ is set free.
              \end{itemize}}
       \end{itemize}}
\end{itemize}
\end{algorithm}

As noted by Gale and Shapley (and others), the above algorithm
computes the male-optimal stable
matching, which is optimal in the very
strong sense that every man likes his partner in
this
matching at least as much as his partner in any other stable matching.
Given an instance with $n$~men and $n$~women, the algorithm computes
the male-optimal stable matching in time $O(n^2)$.

During the algorithm, after a woman becomes ``engaged'' she never
becomes free,
though she might be engaged to different men at different times
during the execution of the algorithm. On the other hand, a man
could oscillate
between being free and being engaged.

It is well-known (see, e.g.~\cite{GS, Knuth})
that the male-optimal matching may be obtained by
taking {\em any} ordering of the men and have them make
proposals in that order, i.e.\ when ``a free man $M$ proposes...''
we can take the highest free man in our ordering of the men to perform
the next proposal.

By reversing
the roles of men and women (i.e.\ the women are the ``proposers''), we
can obtain the female-optimal stable matching.

\subsection{Stable matching lattice}\label{sect:lattice}

Given a matching instance and two stable matchings ${\cal M}$
and ${\cal M'}$ where
\begin{eqnarray*}
{\cal M} & = & \{(M_1,w_1), (M_2,w_2),\cdots,(M_n,w_n)\},\; \\
{\cal M'} & = & \{(M'_1,w_1), (M'_2,w_2),\cdots,(M'_n,w_{n})\},
\end{eqnarray*}
we define $\max\{M_i,M'_i\}$,
$\min\{M_i,M'_i\}$, $\max\{{\cal M},{\cal M'}\}$ and $\min\{{\cal
M},{\cal M'}\}$ as follows:
\begin{align*}
\max\{M_i,M'_i\} &= \text{ favorite choice of woman } w_i \text{
between men } M_i \text{ and } M'_i\\
\min\{M_i,M'_i\} &= \text{ least preferred choice of woman } w_i
\text{ between men } M_i \text{ and } M'_i\\
\max\{{\cal M},{\cal M'}\} &= \{(\max\{M_1,M'_1\},w_1),
(\max\{M_2,M'_2\},w_2),\cdots,(\max\{M_n,M'_n\},w_n)\}\\
\min\{{\cal M},{\cal M'}\} &= \{(\min\{M_1,M'_1\},w_1),
(\min\{M_2,M'_2\},w_2),\cdots,(\min\{M_n,M'_n\},w_n)\}
\end{align*}

Note
that in the expression $\max\{M_i,M'_i\}$, the
woman $w_i$ can deduced from the arguments since she is
the only woman married to $M_i$ in $\cal M$ and in $M'_i$ in
$\cal M'$.  From~\cite{Knuth}, we have that $\max\{{\cal M},{\cal M'}\}$ and
$\min\{{\cal M},{\cal M'}\}$ are themselves stable matchings.
Further, we define the relation
${\cal M} \leq {\cal M'}$ if and only if
${\cal M'} = \max\{{\cal M},{\cal M'}\}$. It is clear that the
relation $\leq$ is reflexive,
antisymmetric, and transitive. Hence, the stable matchings of a
stable matching instance form a lattice under the $\leq$ relation.

In fact, this lattice is a {\em distributive lattice} under the
``max'' and ``min'' operations defined above~\cite{Knuth}.
The male-optimal matching is the minimum element in this lattice
(under the $\leq$ relation), while the female-optimal matching is
the maximum element.

It is well-known (see, for instance,~\cite{Davey}) that a
finite distributive lattice is isomorphic to the lattice of
{\em downsets} of another partial order (ordered by subset inclusion).
We shall shortly see
how this other downset lattice arises in the context of stable matchings,
and its connection to the stable matching lattice.

\subsection{Stable pairs and rotations}\label{sect:rotations}

\begin{definition}
A pair $(M,w)$ is called {\em stable} if and only if $(M,w)$ is a pair in
some stable matching ${\cal M}$.
A pair $(M,w)$ that is not stable is called an {\em unstable pair}.
\end{definition}

\begin{definition} Let ${\cal M}$ be a stable matching. For any
man $M$ (woman $w$), let $sp_{{\cal M}}(M)$ ($sp_{{\cal M}}(w)$)
denote the spouse of man $M$ (woman $w$) in the matching ${\cal M}$.
\end{definition}

\begin{definition}\cite{Randall}
Let ${\cal M}$ be a stable matching.
The {\em suitor} of a man $M$ is defined to be the first woman $w$ on
$M$'s preference list such that (i) $M$ prefers
his spouse over $w$ and (ii) $w$ prefers $M$ over her spouse.
The suitor of man $M$ is denoted by $S_{{\cal M}}(M)$.\end{definition}

We note that $S_{{\cal M}}(M)$ may not exist for every man. For
instance,
if $\cal M$ is the female-optimal stable matching,
then $S_{{\cal M}}(M)$ would
not exist.

\begin{definition}\cite{IL} Let ${\cal M}$ be a stable matching.
Let $R=\{(M_0,w_0), (M_1,w_1), \cdots,$ $(M_{k-1},w_{k-1})\}$ be an
ordered list of pairs from ${\cal M}$ such that for every $i$, $0
\leq i \leq k-1$, $S_{{\cal M}}(M_i)$ is $w_{i + 1 (\!\!\!\mod k)}$.
Then $R$ is a {\em rotation} (exposed in the matching ${\cal M}$).
\end{definition}

A stable matching may have many or no exposed rotations.
Applying an exposed rotation to a stable matching
(i.e.\ breaking the pairs $(M_i,w_i)$
and forming the new pairs $(M_i, w_{i+1})$) gives a new stable
matching in which the women are ``happier" and the men are less
happy. In other words, after a rotation, every woman (respectively, man)
involved in the rotation is married to someone higher (resp.\ lower)
on her (resp.\ his) preference list than her (resp.\ his) partner
in the rotation.

We can similarly define suitors for the
women, given some stable matching ${\cal M}$.  We do not
need to do so for the purposes of this paper, but the Markov
chain that Bhatnagar, et al.\ examine in~\cite{Randall} consists
of moves that are ``male-improving'' and ``female-improving''
rotations.  Starting from any stable matching, it is possible
to obtain any other stable matching using some (appropriately
chosen) sequence of male-improving and/or female-improving
rotations~\cite{IL}.

\begin{definition}\cite{Gusfield}\label{eliminated}
A pair $(M,w)$, not necessarily stable, is said to be
{\em eliminated by the rotation $R$} if $R$ moves $w$ from $M$ or
below on her preference list to a man strictly above $M$.
\end{definition}

Note that if a stable pair $(M,w)$ in a rotation $R$ is eliminated by
$R$, and if $(M,w')$ is any other pair eliminated by $R$, then man
$M$ prefers $w$ over $w'$, for otherwise no matching that has $R$
exposed in it could be stable.

\begin{lemma}\cite{IL}
No pair is eliminated by more than one rotation, and for any pair
$(M,w)$, at most one rotation moves $M$ to $w$.
\end{lemma}

We can now define a relation on rotations.

\begin{definition}\cite{IL}\label{def:explicitly}
Let $R$ and $R'$ be two distinct rotations. Rotation $R$ is said to
{\em explicitly precede} $R'$ if and only if $R$ eliminates a pair
$(M,w)$, and $R'$ moves $M$ to a woman $w'$ such that $M$ (strictly)
prefers $w$ to $w'$. The relation ``precedes'' is defined as the
transitive closure of the ``explicitly precedes'' relation.
\end{definition}

If a rotation $R$ explicitly precedes $R'$ then there is no stable matching
with $R'$ exposed such that applying $R'$ results in a stable matching
with $R$ exposed --- the intermediate matching would have a blocking pair
(hence would not be stable).
The relation {\it
precedes} ($\leq$) defines a partial order on the set of rotations of
the stable matching instance. We call the partial order on the set
of rotations the {\em rotation poset} of the instance and denote it
$(P,\leq)$.

The following theorem relates the rotations in the
rotation poset to the stable matchings of the instance via the
downsets of $P$.

\begin{theorem}\cite[Theorem 4.1]{IL}\label{thm:sm-downsets}
For any stable matching instance, there is a one-to-one
correspondence between the stable matchings of that instance and
the downsets of its rotation poset.
\end{theorem}

Every stable matching of the instance can be obtained by starting
with the male-optimal stable matching and performing the rotations
in the corresponding downset (ensuring that a rotation is performed before
any rotation that succeeds it is performed).
Note that the downsets corresponding
to the male-optimal stable matching and the female-optimal stable
matching are $\emptyset$ and $P$, respectively.

To construct the rotation poset, we need (i) the rotations and
(ii) the precedence relations between them. We note that once we have
all the rotations in the poset, we can establish the precedence
relations using the ``explicitly precedes'' relation, i.e.\ by determining
which (stable or unstable) pairs are eliminated by each rotation.

\subsection{Gusfield's algorithm for finding all rotations}\label{sect:poset}

Given a stable matching instance,
let $H$ be the Hasse diagram of the stable matching lattice defined in
Section~\ref{sect:lattice}. That is, $H$ is the transitive reduction of
the relation $\leq$ on the set of stable matchings.
Gusfield~\cite{Gusfield} gave a fast algorithm for finding all
rotations of a stable matching instance. His algorithm is a refinement
of successive applications of the ``{\em breakmarriage}'' procedure of
McVitie and Wilson~\cite{MW}.
The key ingredient in Gusfield's proof that his algorithm is
correct is the following.

\begin{theorem} \cite[Theorem 6]{Gusfield}
\label{thm:gusfield}
Let $\Phi$ be any path in~$H$ from the male-optimal stable matching to
the female-optimal stable matching. Then any two consecutive matchings
on~$\Phi$ differ by a single rotation, and the set of all rotations
between matchings along~$\Phi$ contains every rotation exactly once.
\end{theorem}

Gusfield presented a well-tuned version of his algorithm that runs in
$O(n^2)$ time. For the sake of presentation, we use the following,
slower,
variant of his algorithm (the variant still runs in polynomial time,
which suffices for our purposes).

\begin{algorithm}
\caption{Find-All-Rotations Algorithm}\label{FR}
\begin{itemize}
\item {Initially we start with the male-optimal
stable matching ${\cal M}_0$, and some ordering of the
men.}

\item {In the current matching ${\cal M}_i$, among the ordered men, pick
the first man who
has a suitor. (If there are no men that have a suitor, then ${\cal M}_i$
is the female-optimal
matching and the algorithm stops.)  Let man $M_1$ be the first man who has a suitor,
namely $S_{{\cal M}_i}(M_1)$.
}

\item {Start constructing the sequence $(M_1,w_1), (M_2,w_2), \cdots $ where
$w_1 = sp_{{\cal M}_{i}}(M_1)$, for $l = 2,3,\cdots $, $w_l =
S_{{\cal M}_i}(M_{l-1})$ and $M_l = sp_{{\cal M}_{i}}(w_l)$.}

\item {If there exists a $t \in \{1,\cdots,l-1 \}$ such that
$w_t = w_l$, then return the rotation $(M_t,w_t),\cdots,(M_l,w_l)$, and apply
the rotation to ${\cal M}_i$ to get a new stable matching ${\cal M}_{i+1}$.
Otherwise, increment $l$ and continue constructing the sequence.}
\end{itemize}
\end{algorithm}

In Algorithm~\ref{FR},
the existence of $t$ is guaranteed by the
fact that the current stable matching ${\cal M}_i$ is different from
the female-optimal stable matching and, hence, has a rotation exposed
in it.
The correctness of the algorithm follows from
Theorem~\ref{thm:gusfield}:
Starting at ${\cal M}_i$, the algorithm applies a rotation to
obtain~${\cal M}_{i+1}$. Since rotations improve the utility of
the women involved,
${\cal M}_{i+1}\geq {\cal M}_i$.
To apply Theorem~\ref{thm:gusfield} we need only
argue that the step from ${\cal M}_i$ to ${\cal M}_{i+1}$ is a single
step in~$H$ rather than multiple steps.
Equivalently, we need to argue that the rotation between ${\cal M}_i$
and ${\cal M}_{i+1}$ cannot be decomposed as a sequence of smaller
rotations (where these smaller rotations correspond to individual steps in~$H$).
This follows by the definition of ``suitor'' --- if ${\cal M}_i$ has
rotation~$R$ exposed and applying rotation~$R$ yields a stable
matching ${\cal M}'$ with $R'$ exposed and applying $R'$ yields~${\cal
  M}_{i+1}$
then either $R$ and $R'$ are disjoint (contradicting the fact that the
transformation from
${\cal M}_i$ to ${\cal M}_{i+1}$ can be accomplished with a single
rotation)
or  $R$ and $R'$ share a man (in which case the transformation
from~${\cal M}_i$ to~${\cal M}_{i+1}$ does not move this man to his
suitor in~${\cal M}_i$, so is not a rotation).

Once we find all of the rotations using Algorithm~\ref{FR},
we can order them to find the rotation poset $P$ using
the relation given in Definition~\ref{def:explicitly}.

\subsection{$\#BIS$, independent sets, and stable matchings}\label{sect:idea}
The rotation poset for a matching instance plays a key role in what
follows.  To prove Theorem~\ref{thm:3-attr-dot}, we take a
$\#BIS$ instance $G=(V_1\cup V_2, E)$ and view this as the rotation
poset of a matching instance.  In particular, $G$ is the Hasse diagram
of the poset when we draw $G$ with the set $V_2$ ``above'' $V_1$.

Each independent set in the bipartite graph
naturally corresponds to a downset in the partial order,
and vice-versa.
See Figure~\ref{fig:downsets}
for an example.  An independent set, namely $\{d, f, g\}$,
is shown in the left of that figure.
The corresponding downset is shown on the right.  This downset
is obtained by taking the set $\{d, f, g\}$ and adding the
two elements $a$ and $b$,
as $a < f$ and $b < f$ (and $b < g$) in the Hasse diagram.
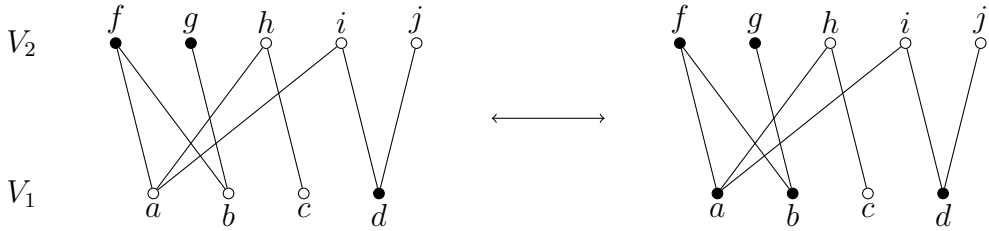
\begin{figure}[ht]
\begin{center}
\begin{tikzpicture}
     [inner sep=0.5mm]
   \node (A) at (0,0) [circle,draw,label=below:$a$] {};
  \node (B) at (1,0) [circle,draw,label=below:$b$] {};
  \node (C) at (2,0) [circle,draw,label=below:$c$] {};
  \node (D) at (3,0) [circle,draw,fill=black,label=below:$d$] {};
   \node (F) at (-.5,2) [circle,draw,fill=black,label=above:$f$] {};
  \node (G) at (.5,2) [circle,draw,fill=black,label=above:$g$] {};
  \node (H) at (1.5,2) [circle,draw,label=above:$h$] {};
  \node (I) at (2.5,2) [circle,draw,label=above:$i$] {};
  \node (J) at (3.5,2) [circle,draw,label=above:$j$] {};

  \node (A') at (7.5,0) [circle,draw,fill=black,label=below:$a$] {};
  \node (B') at (8.5,0) [circle,draw,fill=black,label=below:$b$] {};
  \node (C') at (9.5,0) [circle,draw,label=below:$c$] {};
  \node (D') at (10.5,0) [circle,draw,fill=black,label=below:$d$] {};

  \node (F') at (7,2) [circle,draw,fill=black,label=above:$f$] {};
  \node (G') at (8,2) [circle,draw,fill=black,label=above:$g$] {};
  \node (H') at (9,2) [circle,draw,label=above:$h$] {};
  \node (I') at (10,2) [circle,draw,label=above:$i$] {};
  \node (J') at (11,2) [circle,draw,label=above:$j$] {};

  \draw[-] (A) to (F);
  \draw[-] (A) to (H);
  \draw[-] (A) to (I);

  \draw[-] (B) to (F);
  \draw[-] (B) to (G);

  \draw[-] (C) to (H);

  \draw[-] (D) to (I);
  \draw[-] (D) to (J);

  \draw[-] (A') to (F');
  \draw[-] (A') to (H');
  \draw[-] (A') to (I');

  \draw[-] (B') to (F');
  \draw[-] (B') to (G');

  \draw[-] (C') to (H');

  \draw[-] (D') to (I');
  \draw[-] (D') to (J');

  \draw (-1.75,2) node {$V_2$};
  \draw (-1.75,0) node {$V_1$};

  \draw[<->] (4.5,1) -- (6,1);
\end{tikzpicture}\caption{The correspondence between independent
                  sets and downsets}\label{fig:downsets}
\end{center}
\end{figure}
Conversely, given a downset, such as the one on the right of the diagram,
we can find the corresponding independent set in $G$ by taking the set
of maximal elements of the downset.

So given $G$, we then construct a matching instance (using
$3$-dimensional preference and attribute vectors) whose rotation
poset is (isomorphic to) $G$, giving a $1$-$1$ correspondence for our
AP-reduction from $\#BIS$ to $\#SM(k\textrm{-attribute})$, showing
that $\#BIS \leq_{AP} \#SM(k\textrm{-attribute})$.  This construction,
and the proof of the correspondence, is in Section~\ref{sect:k-attribute}.

The reverse implication $\#SM(k\textrm{-attribute}) \leq_{AP} \#BIS$
follows from the two results that
$\#SM \leq_{AP} \#Downsets$ (Theorem~\ref{thm:sm-downsets}, quoted
here from~\cite{IL})
and $\#Downsets \leq \#BIS$~\cite[Lemma 9]{DGGJ},
where $\#Downsets$ is the problem of counting the number of downsets in
a partial order.

\section{Stable matchings in the k-attribute model ($k \geq 3$)}\label{sect:k-attribute}

In this section we give our construction to show AP-reducibility from
$\#BIS$ to the $k$-attribute stable matching model when $k \geq 3$.

Given our previous remarks about the relation between $\#BIS$, independent
sets, and stable matchings, our procedure is as follows:
\begin{enumerate}
\item Let $G=(V_1 \cup V_2, E)$ denote a bipartite graph where $|E| = n$.
Our goal will be to construct a $k$-attribute stable matching instance for which we can
show that the Hasse diagram of its rotation poset is~$G$.
This will give a bijection between stable matchings and downsets of~$G$,
hence a bijection between stable matchings and independent sets of~$G$.
\item Using $G$, in the manner to be specified
in Section~\ref{sect:construction},
we construct preference lists for a $3$-attribute
stable matching instance with $3n$ men and $3n$ women.
\item Given this matching instance, we find the male-optimal
and female-optimal matchings.
\item Using the {\bf Find-All-Rotations} algorithm, we extract the
rotations from our stable matching instance.
\item Having these rotations, we construct the partial order, $P$, on
these rotations (specified by the transitive closure of the
``explicitly precedes'' relation).
\item  We finally show that $P$ is isomorphic to
$G$ (when $G$ is viewed as
a partial order), thereby showing our construction is an
approximation-preserving reduction from $\#BIS$ to
$\#SM(3\rm{-attribute})$.
\end{enumerate}

\subsection{Construction of the stable matching instance}\label{sect:construction}

\subsubsection{$BIS$ and permutations}\label{sect:permutations}

Let $G= (V_1\cup V_2,E)$ denote our $BIS$ instance, where
$E \subseteq V_1 \times V_2$ and $|E| = n$.

Using $G$ we will construct a $3$-attribute stable matching instance
with $3n$ men and $3n$ women. The men and women of the instance are
denoted $\{A_1,\ldots, A_n, B_1, \ldots, B_n, C_1, \ldots,$ $C_n\}$
and $\{a_1,\ldots, a_n, b_1, \ldots, b_n, c_1, \ldots, c_n\}$,
respectively. To describe our construction, we label the edges of
$G$ $B_1$ through $B_n$ from ``left-to-right'' with respect to the
vertices ($V_1$) on the bottom. This becomes more clear from the
example in Figure~\ref{fig:BIS-labeling}.  We refer to edge $B_i$ as
man $B_i$, and this will be clear from the context.

\begin{figure}[ht]
\begin{center}
\begin{tikzpicture}[inner sep=0.6mm]

  \node (A) at (0,0) [circle,draw] {};
  \node (B) at (3,0) [circle,draw] {};
  \node (C) at (6,0) [circle,draw] {};

  \node (F) at (-1,3) [circle,draw] {};
  \node (G) at (2,3) [circle,draw] {};
  \node (H) at (5,3) [circle,draw] {};
  \node (I) at (8,3) [circle,draw] {};

  \draw (A) to node[midway, below, left] {$B_1$} (F);
  \draw (A) to node[near start, above left] {$B_2$} (G);
  \draw (A) to node [very near start, below right] {$B_3$} (I);

  \draw (B) to node[near end, above right] {$B_4$} (G);
  \draw (B) to node[near end, above left] {$B_5$} (H);
  \draw (B) to node[midway, below right] {$B_6$} (I);

  \draw (C) to node[very near end, above right] {$B_7$} (F);
  \draw (C) to node[midway, below right] {$B_8$} (I);

  \draw (-2,2.8) node {$V_2$};
  \draw (-2,.2) node {$V_1$};
\end{tikzpicture}\caption{A $BIS$ instance and our labeling of its edges}\label{fig:BIS-labeling}
\end{center}
\end{figure}
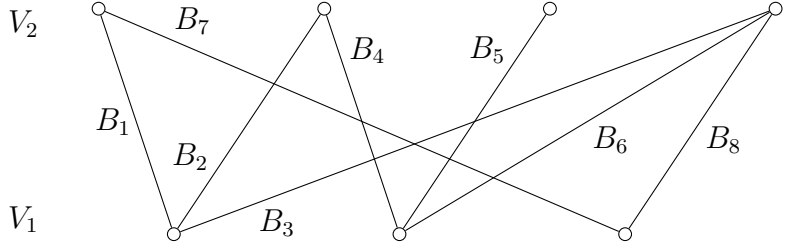

For our construction we
associate two permutations, $\rho$ and $\sigma$, of $[n]=\{1, \ldots, n\}$ with
the BIS instance. The cycles of $\rho$ correspond to vertices in $V_1$
and those of $\sigma$ correspond to vertices in $V_2$.
In other words, if the edges incident to a vertex in $V_2$
are $B_{i_1}, B_{i_2}, \ldots, B_{i_d}$, then $(i_1, i_2, \ldots, i_d)$
is a $\sigma$-cycle. We define $\rho$-cycles in a similar fashion.
If $G$ has $k =|V_1|$ vertices on the bottom and
$l=|V_2|$ vertices on the top, then the permutations $\rho$ and
$\sigma$ have $k$ and $l$ cycles, respectively.  Since the graph $G$
will turn out to be isomorphic to
a rotation poset, every vertex in $G$
will represent a rotation in
the stable matching instance. The
rotations of the stable matching instance will be governed by the
$\rho$- and $\sigma$-cycles in a manner to be specified.
The rotations corresponding to the
$\rho$-cycles will be called $\rho$-rotations and those
corresponding to the $\sigma$-cycles will be called $\sigma$-rotations.

In the example of Figure~\ref{fig:BIS-labeling}, the three
$\rho$-cycles are $\rho_1 = (1, 2, 3), \rho_2 = (4, 5, 6)$, and
$\rho_3 = (7, 8)$.  The four $\sigma$-cycles are $\sigma_1 = (1, 7),
\sigma_2 = (2,4), \sigma_3 = (5)$, and $\sigma_4 = (3, 6, 8)$.

Here is a brief overview of how we go about constructing a stable matching
instance from a given bipartite graph.

First of all, the male-optimal stable matching in our matching instance
we construct will consist of the pairs
$(A_i, a_i), (B_i, b_i), (C_i, c_i)$ for all $i\in [n]$.  (We must
show later this is indeed the case for the construction
we describe.)

A $\rho$-cycle of the form $(i_1, i_1+1, \ldots, i_2)$ will
correspond to the $\rho$-rotation, $R$, of the form
$$\{(B_{i_1},b_{i_1}), (A_{i_1},a_{i_1}), (B_{i_1 + 1},b_{i_1 +1}),
(A_{i_1+1},a_{i_1+1}), \ldots, (B_{i_2},b_{i_2}), (A_{i_2},a_{i_2})\}.$$
This rotation $R$ arises from a vertex $v \in V_1$ with edges
$B_{i_1}, B_{i_1+1}, \ldots, B_{i_2}$ incident to it.

We will later show that a $\sigma$-rotation $R'$ is of the form
$$\{(B_{i_1},a_{i_1}), (C_{i_1},c_{i_1}), (B_{i_2},a_{i_2}),
(C_{i_2},c_{i_2}),
\ldots, (B_{i_p},a_{i_p}), (C_{i_p},c_{i_p})\},$$
where $(i_1, i_2, \ldots, i_p)$ is the corresponding $\sigma$-cycle,
and that the rotation $R'$ corresponds to the vertex $v' \in V_2$ with edges
$B_{i_1},B_{i_2},\ldots,B_{i_p}$ incident to it.

In this manner,
every rotation in the rotation poset is defined in terms of the men
involved in them, the women being the (then-current) partners of
the men that are in the rotation.
Assuming that the above two claims regarding
rotations are valid (as we will show below), we make the following observation.

\begin{observation}\label{obs:sig-rho}
A $\rho$-cycle and a $\sigma$-cycle can have at most one
element in common. (This is because $G$ is a graph and not a
multi-graph.)  This means that a $\rho$-rotation and a $\sigma$-rotation
can have at most one man in common.  This similarly holds for the women.
\end{observation}

In the next section we start by assigning preference vectors and
position vectors to
the men and women in our stable matching instance.
Following that, we construct the initial
portion of their preference lists.  We then find the male- and the
female-optimal matchings using the Gale-Shapley algorithm. After
finding the male- and the female-optimal matchings, we extract all
the rotations of the rotation poset using the {\bf Find-All-Rotations}
algorithm. Finally, we obtain the rotation poset by ordering
rotations using the {\it explicitly precedes} relation.  As we
stated earlier, we will find this rotation poset is isomorphic to $G$,
showing our construction is a mapping from the set of $\#BIS$
instances to $\#SM(3\textrm{-attribute})$ instances.

\subsubsection{Assigning preference and position vectors}\label{sect:preferences}

Suppose $D_1, \ldots, D_l$ are the $l$ cycles of $\sigma$ of
lengths $p_1, \ldots, p_l$, respectively.
Let $e_i$ be a representative element of
cycle $D_i$. In other words, we can represent the $\sigma$-cycle $D_i$ as
$D_i = \left(e_i,\sigma(e_i),\ldots,\sigma^{p_i-1}(e_i)\right)$.
(We may, for example,
select $e_i$ to be the smallest number in the cycle, and we will do so here).
In what follows we will often abbreviate
$\sigma x = \sigma(x), \sigma^2 x = \sigma^2(x), \sigma^{-1} x = \sigma^{-1}(x)$, etc,
and, similarly, $\rho x = \rho(x)$, etc.

Let $Rep(\sigma) = \{e_1,e_2\cdots,e_l\}$ be the set of
representative elements we choose for the cycles of $\sigma$. Let $W_i =
\{a_x:x\in D_i\}\cup\{b_{\rho x}:x\in D_i\}\cup\{c_{\sigma^{-1} x}:x\in D_i\}$.
Let $T(x)= \{c_{\sigma^{-1} x},a_x,b_{\rho x}\}$ where $x\in D_i$. It follows
that
$W_i = \cup_{x\in D_i}T(x)$ and $T(i) \cap T(j) = \emptyset$
for $i \neq j$.

Using the definitions above, we begin to create a stable
matching instance in the $3$-attribute model whose rotation
poset is the graph $G$.  As a reminder, every man, say $A_i$, is
associated with two vectors: (i) a position vector denoted by
$\bar{A}_i$, and (ii) a preference vector denoted by $\hat{A}_i$.
Every woman similarly has her own position and preference vectors.
Each man ranks the women based on the dot
product of his preference vector with their position vectors. In
other words, if $\hat{A}_i \cdot \bar{b} > \hat{A}_i \cdot \bar{c}$,
then man $A_i$ prefers woman $b$ over $c$.
Note that we can always assign preference vectors so that $|\hat{A}_i| = 1$
(by normalizing those vectors).

Our task, therefore, is to specify the position and preference vectors
for all the men and women in our matching instance.

First we fix the position vectors of the women. The $z$-coordinate
of women $a_i$ and $c_i$ is set to 0 for $1 \leq i \leq n$. The
$z$-coordinate of woman $b_i$ is set to $4^{i}$ for $1\leq i \leq
n$. The $x$- and $y$-coordinates of $a_i$, $b_i$, and $c_i$ are such
that the projection of each women's position vector onto the $x$-$y$
plane lies on the unit
circle $x^2 + y^2 = 1$.  Furthermore, we group the projections
according to the sets $W_i$. In other words, all women in $W_i$ are
embedded in an angle of $\epsilon$ on the unit circle, where
$\epsilon = 2\pi/n^2$. These groups
are embedded around the circle in the order $W_1$ through $W_l$, and
the angle between two adjacent groups is $(2\pi-l\epsilon)/l$. Note
that $W_l$ is adjacent to $W_{l-1}$ and $W_1$.
Group $W_i$ starts at angle $2\pi(i-1)/l$ and ends at
$2\pi(i-1)/l + \epsilon$.

Within the group $W_i$, the women are
further sub-grouped into triplets
$$T(e_{i}), T(\sigma(e_{i})),\ldots,
T(\sigma^{p_i-1}(e_i)).$$
Within the angle of size $\epsilon$, the sub-groups
are embedded in the order
$T(e_{i})$ through $T(\sigma^{p_i-1}(e_{i}))$,
with each $T(\cdot)$ spanning an angle of
$6\theta_i$.
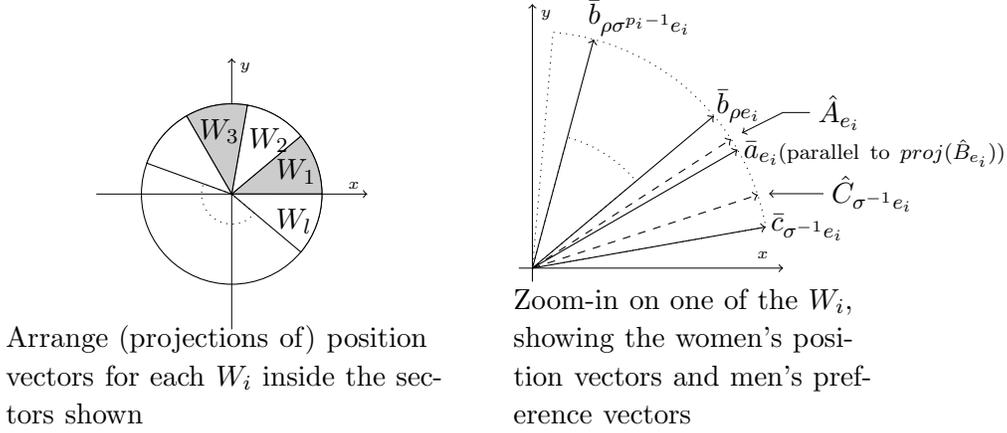
\begin{figure}[ht]
\begin{center}
\begin{tikzpicture}[scale=1.2]

  \draw[->, very thin,name=x-axis] (-1.5, 0) -- (1.5, 0);
  \draw[->, very thin,name=y-axis] (0, -1.5) -- (0, 1.5);
  \draw (1.5,.1) node[text width=.5cm] {{\tiny $x$}};
  \draw (.3,1.4) node[text width=.5cm] {{\tiny $y$}};

  \draw (0,0) circle (1cm);
  \filldraw[fill=black!20!white] (0,0) -- (1cm,0cm) arc (0:40:1cm) -- cycle;
  \filldraw[fill=white!20!white] (0,0) -- (0.766cm,0.643cm) arc (40:80:1cm) -- cycle;
  \filldraw[fill=black!20!white] (0,0) -- (0.174cm,0.985cm) arc (80:120:1cm) -- cycle;
  \filldraw[fill=white!20!white] (0,0) -- (-.5cm, 0.866cm) arc (120:160:1cm) -- cycle;
  \filldraw[fill=white!20!white] (0,0) -- (0.766cm,-0.643cm) arc (320:360:1cm) -- cycle;
  \draw[dotted] (0,0) -- (-0.313cm,0.114cm) arc (160:320:.333cm) -- cycle;

  \draw (.7cm, .25cm) node[text width=.5cm] {{\small $W_1$}};
  \draw (.4cm, .6cm) node[text width=.5cm] {{\small $W_2$}};
  \draw (-.15cm, .7cm) node[text width=.5cm] {{\small $W_3$}};
  \draw (.7cm, -.3cm) node[text width=.5cm]  {{\small $W_l$}};

  \node[text width=6cm] at (0,-2) {{\small Arrange (projections of) position
  vectors for each $W_i$ inside the sectors shown}};
\end{tikzpicture}
\hspace*{0.2cm}
\begin{tikzpicture}[scale=.9]
  \draw[->, very thin, name=x-zoomed] (.3, -.7) -- (4.2, -.7);
  \draw[->, very thin, name=y-zoomed] (.5, -.9) -- (.5, 3.2);
  \draw (4.1,-.5) node[text width=.5cm] {{\tiny $x$}};
  \draw (.9,3.05) node[text width=.5cm] {{\tiny $y$}};

  \draw[dotted] (.5,-.7) -- (3.947,-0.092) arc (10:85:3.5cm) -- cycle;

  \draw[->] (.5, -.7) -- (3.947,-0.092);
  \draw[->] (.5, -.7) -- (3.531,1.05);
  \draw[->] (.5, -.7) -- (3.181,1.550);

  \draw (4.3, -.1) node[text width=.5cm] {{\small $\bar{c}_{\sigma^{-1}e_i}$}};
  \draw (3.9, 1) node[text width=.5cm] {{\small $\bar{a}_{e_i}$}};
  \draw (3.5, 1.7) node[text width=.5cm] {{\small $\bar{b}_{\rho e_i}$}};

  \draw[dotted] (.5, -.7) -- (2.032, 0.586) arc(40:75:2cm) -- cycle;

  \draw[->] (.5, -.7) -- (1.406, 2.681);
  \draw (1.6, 3) node[text width=.5cm] {{\small $\bar{b}_{\rho\sigma^{p_i-1} e_i}$}};

  \draw[->,dashed] (.5, -.7) -- (3.829, 0.382);
  \draw[->,dashed] (.5, -.7) -- (3.435, 1.206);

  \draw (5.2,.4) node[text width=.5cm] {{\small $\hat{C}_{\sigma^{-1}e_i}$}};
  \draw[->, very thin] (4.8,.4) -- (4,.4);
  \draw (5,1.6) node[text width=.5cm] {{\small $\hat{A}_{e_i}$}};
  \draw[->, very thin] (4.6, 1.6) -- (4.2,1.6) -- (3.6, 1.3);

  \draw (5.8, 1) node[text width=3cm]
        {{$\scriptstyle \textrm{(parallel to }proj(\hat{B}_{e_i})\textrm{)}$}};

  \node[text width=5cm] at (3,-2){{\small Zoom-in on one of the $W_i$,
  showing the women's position vectors and men's preference vectors}};

\end{tikzpicture}
\end{center}\caption{Placement of the women's position vectors and
        men's preference vectors}\label{fig:women-positions}
\end{figure}
The angle between two adjacent $T(\cdot)$'s is
$\theta_i$, where $\theta_i = \epsilon/(7p_i - 1)$. Within each
$T(x)$, the women appear in the order $c_{\sigma^{-1} x}, a_x$, and
$b_{\rho x}$, and the angle between $\bar{c}_{\sigma^{-1} x}$ and
$\bar{a}_x$ is $4\theta_i$, and the angle between $\bar{a}_x$ and
$\bar{b}_{\rho x}$ is $2\theta_i$. We summarize the above
description by giving the exact coordinates for the position vector
for the women.
\begin{align*}
\textrm{Let}\ \epsilon &= \frac{2\pi}{n^2}.\\
\textrm{For}\  e_i &\in Rep(\sigma), \; \textrm{let}\ \theta_i = \epsilon/(7p_i-1).\;\; \textrm{Then for}\ 0\leq m\leq p_i-1\ \textrm{define}\\
\bar{a}_{\sigma^{m}{e_i}} &= \left(\cos(2\pi(i-1)/l + 7m\theta_i +
4\theta_i), \;\;\sin(2\pi(i-1)/l + 7m\theta_i +
4\theta_i), \;\;0\right), \\
\bar{b}_{\rho\sigma^{m}{e_i}} &= (\cos(2\pi(i-1)/l + 7m\theta_i +
6\theta_i), \;\;\sin(2\pi(i-1)/l + 7m\theta_i + 6\theta_i),
\;\;
4^{\rho\sigma^{m}{e_i}}),\; \textrm{and} \\
\bar{c}_{\sigma^{m-1}{e_i}} &= \left(\cos(2\pi(i-1)/l + 7m\theta_i),
\;\;\sin(2\pi(i-1)/l + 7m\theta_i), \;\;0\right).
\end{align*}

Next we define the preference vectors of the men. The $z$-coordinates of
all $\hat{A}_i$ and $\hat{C}_i$ are set to 0. We place $\hat{A}_i$
between $\bar{a}_i$ and the projection onto the $x$-$y$ plane of
$\bar{b}_{\rho i}$.  If the angle
between $\bar{a}_i$ and (the projection of) $\bar{b}_{\rho i}$ is $\alpha$,
then the angle between $\bar{a}_i$ and $\hat{A}_i$ is $\frac{1}{3}\alpha$,
and the angle between $\hat{A}_i$ and (the projection of) $\bar{b}_{\rho i}$
is $\frac{2}{3}\alpha$. This will ensure that $A_i$ prefers $a_i$ over
$b_{\rho i}$. We will later show that the preference list of $A_i$
starts with $a_i b_{\rho i}$. We place $\hat{C}_i$ between
$\bar{c}_i$ and $\bar{a}_{\sigma i}$ such that if the angle between
$\bar{c}_i$ and $\bar{a}_{\sigma i}$ is $\beta$, then the angle
between $\bar{c}_i$ and $\hat{C}_i$ is $\frac{2}{5}\beta$ and the
angle between $\hat{C}_i$ and $\bar{a}_{\sigma i}$ is
$\frac{3}{5}\beta$. This will ensure that $C_i$ prefers $c_i$ over
$a_{\sigma i}$. We will later show that the preference list of $C_i$
starts with $c_i a_{\sigma i}$.

Finally, we place $\hat{B}_i$, which is of unit length, such that
$\hat{B}_i$ makes an angle of $\phi = 2\pi/100$ with the vertical
axis ($z$-axis) and its projection on the $x$-$y$ plane is parallel
to $\bar{a}_i$. In other words, the projection of $\hat{B}_i$ on the
$z=0$ plane is $\sin \phi\; \bar{a}_i$. We summarize the above
discussion by providing the exact coordinates of $\hat{A}_i$, $\hat{B}_i$, and
$\hat{C}_i$.
\begin{align*}
\textrm{Let}\ \phi &= 2\pi /100\ \textrm{and}\ \epsilon = \frac{2\pi}{n^2}.\;\; \\
\textrm{For}\ e_i &\in Rep(\sigma),\ \textrm{let}\ \theta_i = \epsilon/(7p_i-1).\;\; \textrm{Then for}\ 0\leq m\leq p_i-1\ \textrm{define}\\
\hat{A}_{\sigma^{m}{e_i}} &= \left(\cos(2\pi(i-1)/l + 7m\theta_i +
(14/3)\theta_i), \;\;\sin(2\pi(i-1)/l + 7m\theta_i +
(14/3)\theta_i), \;\;0\right), \\
\hat{B}_{\sigma^{m}{e_i}} &= (\sin \phi \cos(2\pi(i-1)/l +
7m\theta_i + 4\theta_i), \;\;\sin \phi \sin(2\pi(i-1)/l + 7m\theta_i
+ 4\theta_i),
\;\;\cos \phi ),\\
\textrm{and\;\;\;}& \\
\hat{C}_{\sigma^{m-1}{e_i}} &= \left(\cos(2\pi(i-1)/l + 7m\theta_i +
(8/5)\theta_i), \;\;\sin(2\pi(i-1)/l + 7m\theta_i + (8/5)\theta_i),
\;\;0\right).
\end{align*}

In Section~\ref{sect:pref-lists} we establish the preference lists
of the men and women.  The vectors given above let us determine the
preference lists of the men, so we now specify the position vectors
of the men and the preference vectors of the women.  This proceeds
in a similar manner as above.

Suppose $E_1$ through
$E_k$ are the
$k$ cycles of $\rho$ of lengths $q_1$ through $q_k$, respectively.
As above, let $f_i$ be a
representative element of cycle $E_i$, so that we can
write the $\rho$-cycle as
$\left(f_i, \rho(f_i), \ldots, \rho^{q_i-1}(f_i)\right)$.
Let
$Rep(\rho) = \{f_1,f_2\cdots,f_k\}$ be the set of representative
elements we select for the cycles of $\rho$.
Let $U_i = \{A_{\rho^{-1}r}:r\in E_i\}\cup\{B_r:r\in E_i\}\cup\{C_r:r\in
E_i\}$.
Let $S(r) = \{A_{\rho^{-1} r},B_r,C_{r}\}$ where $r\in E_i$.
It follows that $U_i = \cup_{r\in E_i}S(r)$ and $S(i) \cap S(j) =
\emptyset$ for $i \neq j$.

We fix the position vectors of the men. The placement of the men is
similar to that of the women. The $z$-coordinate of the men $A_i$
and $B_i$ is set to 0 for $1 \leq i \leq n$. The $z$-coordinate of
man $C_{i}$ is set to $4^{i}$ for $1\leq i \leq n$. The $x$- and
$y$-coordinates of $A_i$, $B_i$ and $C_i$ are such that the
projection of the men onto the $z=0$ plane lies on the unit circle
$x^2 + y^2 = 1$.  Similar to above, the projections are grouped
according to the sets $U_i$. In other words, with
$\epsilon=2\pi/n^2$, all men in $U_i$ are embedded in an angle of
$\epsilon$ on the unit circle. The groups are embedded around the
circle in the order $U_1$ through $U_k$ and the angle between two
adjacent groups is $(2\pi-k\epsilon)/k$. Note that $U_k$ is adjacent
to $U_{k-1}$ and $U_1$. The group $U_i$ starts at angle
$2\pi(i-1)/k$ and ends at $2\pi(i-1)/k + \epsilon$. Within the group
$U_i$, the men are further sub-grouped into triplets $S(f_i),
S(\rho(f_i)),\cdots, S(\rho^{q_i-1}(f_i))$. Within the angle of
$\epsilon$, the sub-groups are embedded in the order $S(f_i)$
through $S(\rho^{q_i-1}(f_i))$ with each $S(\cdot)$ spanning an
angle of $6\omega_i$, where the angle between two adjacent
$S(\cdot)$'s is $\omega_i = \epsilon/(7q_i-1)$. Within each $S(j)$,
the men appear in the order $A_{\rho^{-1} j}, B_j$ and $C_{j}$, and
the angle between $\bar{A}_{\rho^{-1} j}$ and $\bar{B}_j$ is
$4\omega_i$ and the angle between $\bar{B}_j$ and $\bar{C}_{j}$ is
$2\omega_i$. Here are the exact co-ordinates for the position vector
of each man.
\begin{align*}
\textrm{Let}\ \epsilon &= \frac{2\pi}{n^2}. \\
\textrm{For}\ f_i &\in Rep(\rho),\; \textrm{let}\ \omega_i = \epsilon/(7q_i-1).\; \textrm{Then for}\ 0\leq m\leq q_i-1\;\textrm{we define}\\
\bar{A}_{\rho^{m-1}{f_i}} &= \left(\cos(2\pi(i-1)/k + 7m\omega_i),
\;\;\sin(2\pi(i-1)/k + 7m\omega_i), \;\;0\right), \\
\bar{B}_{\rho^{m}{f_i}} &= \left(\cos(2\pi(i-1)/k + 7m\omega_i +4\omega_i),
\;\;\sin(2\pi(i-1)/k + 7m\omega_i +
4\omega_i), \;\;0\right), \; \textrm{and} \\
\bar{C}_{\rho^{m}{f_i}} &= (\cos(2\pi(i-1)/k + 7m\omega_i +
6\omega_i), \;\;\sin(2\pi(i-1)/k + 7m\omega_i + 6\omega_i),
\;\;4^{\rho^{m}{f_i}}).
\end{align*}

Finally, we define the preference vectors of the women.
The $z$-coordinates of
$\hat{b}_i$ and $\hat{c}_i$ are set to 0. Suppose the angle
between $\bar{B}_i$ and (the projection onto the $x$-$y$ plane of)
$\bar{C}_i$ is $\alpha$.  Then we
place $\hat{c}_i$ in the $x$-$y$ plane between
$\bar{B}_i$ and (the projection of) $\bar{C}_{i}$ such that the angle between
$\bar{B}_i$ and $\hat{c}_i$ is $\frac{1}{3}\alpha$, and the angle
between $\hat{c}_i$ and (the projection of) $\bar{C}_{i}$
is $\frac{2}{3}\alpha$. We
place $\hat{b}_i$ between $\bar{A}_{\rho^{-1}i}$ and
$\bar{B}_{i}$ such that if the angle between $\bar{A}_{\rho^{-1}i}$
and $\bar{B}_{i}$ is $\beta$, then the angle between
$\bar{A}_{\rho^{-1}i}$ and $\hat{b}_i$ is $\frac{2}{5}\beta$ and the
angle between $\hat{b}_i$ and $\bar{B}_{i}$ is $\frac{3}{5}\beta$.

We place $\hat{a}_i$, which is of unit length, such that
$\hat{a}_i$ makes an angle of $\phi = 2\pi/100$ with the vertical
axis ($z$-axis) and its projection on the $z = 0$ plane is parallel
to $\bar{B}_i$. In other words, the projection of $\hat{a}_i$ on the
$z=0$ plane is $\sin \phi\; \bar{B}_i$. Therefore, the exact
coordinates of
the preference vectors $\hat{c}_i$, $\hat{a}_i$ and $\hat{b}_i$ are
as follows.
\begin{align*}
\textrm{Let}\ \phi &= 2\pi /100\ \textrm{and}\ \epsilon = \frac{2\pi}{n^2}.\\
\textrm{For}\ f_i &\in Rep(\rho), \;\textrm{let}\ \omega_i = \epsilon/(7q_i-1).\; \textrm{Then for}\ 0\leq m\leq q_i-1\; \textrm{we define}\\
\hat{a}_{\rho^{m}{f_i}} &= (\sin \phi \cos(2\pi(i-1)/k + 7m\omega_i
+ 4\omega_i), \;\sin \phi \sin(2\pi(i-1)/k + 7m\omega_i + 4\omega_i), \;\cos \phi),\;  \\
\hat{b}_{\rho^{m}{f_i}} &= \left(\cos(2\pi(i-1)/k + 7m\omega_i + (8/5)\omega_i),
\;\;\sin(2\pi(i-1)/k + 7m\omega_i + 8/5\omega_i), \;\;0\right),\; \textrm{and}\\
\hat{c}_{\rho^{m}{f_i}} &= \left(\cos(2\pi(i-1)/k + 7m\omega_i +
14/3\omega_i), \;\;\sin(2\pi(i-1)/k + 7m\omega_i + (14/3)\omega_i), \;\;0\right). \\
\end{align*}

\subsubsection{Constructing (partial) preference lists}\label{sect:pref-lists}

Using the vectors defined in the previous section, we now examine the
preference lists of the men and women of our constructed instance.

First we will establish that the preference lists of $A_i$ and $C_i$
start with $a_i b_{\rho i}$ and $c_i a_{\sigma i}$, respectively.
Since $\hat{A}_i$ and $\hat{C}_i$ have a $z$-component that is equal to
zero, it is enough to consider the projections of $\bar{a}_i$, $\bar{b}_i$ and
$\bar{c}_i$ on the $x$-$y$ plane. Furthermore, since $\hat{A}_i$,
$\hat{C}_i$, and the projections of $\bar{a}_i$, $\bar{b}_i$ and
$\bar{c}_i$ are all of unit length, the dot product is essentially a
function of the angle between the two vectors. In other words, if
the angle between $\hat{A}_i$ and $\bar{b}$ is greater than the
angle between $\hat{A}_i$ and $\bar{c}$, then $\hat{A}_i\cdot
\bar{b} < \hat{A}_i\cdot \bar{c}$. Recalling that $\hat{A}_i$ lies
between $\bar{a}_i$ and $\bar{b}_{\rho i}$, and is closer to $\bar{a}_i$,
then ${a}_i$ will appear first on the preference list of $A_i$. The
women positioned next to $a_i$ on the unit circle are $b_{\rho i}$
and $c_{\sigma^{-1} i}$. Since the angle between $\bar{a}_i$ and
$\bar{c}_{\sigma^{-1} i}$ is twice the angle between $\bar{a}_i$ and
$\bar{b}_{\rho i}$ (and $\hat{A}_i$ lies between $\bar{a}_i$ and
$\bar{b}_{\rho i}$), we see that woman $b_{\rho i}$ appears second on the
preference list of $A_i$.

Since $\hat{C}_i$ lies between $\bar{c}_i$ and $\bar{a}_{\sigma i}$
and is closer to $\bar{c}_i$, we see that ${c}_i$ will appear first on the
preference list of $C_i$.
By construction, the angle between $\bar{c}_i$ and
$\bar{a}_{\sigma i}$ is
$4\theta_i$.
We note that the angle between
$\hat{C}_i$ and $\bar{c}_i$ is
$2/5 (4\theta_i) = 8/5\ \theta_i$
and that
between $\hat{C}_i$ and $\bar{a}_{\sigma i}$ is
$3/5 (4\theta_i) = 12/5\ \theta_i$.
We consider two cases (i) $c_i$ is not the first woman
in $W_{j}$ for $1 \leq j \leq l$, i.e.\
$\sigma i \notin Rep(\sigma)$,\;
(ii) $c_i$ is the first woman in some $W_{j}$ for $1\leq j \leq l$,
i.e.\ $\sigma i = e_j$ for some $1\leq j\leq l$.

Case(i) As $c_i$ is not the first woman in $W_{j}$
for $1 \leq j \leq l$,
the women positioned next to $c_i$ are $a_{\sigma i}$ and
$b_{\rho i}$ where $c_i \in T(\sigma i)$ and $b_{\rho i} \in T(i)$.
The angle between (the projection of) $\bar{b}_{\rho i}$
and $\bar{c}_i$ is the angle
between two adjacent $T(\cdot)$'s, which is one-fourth of the angle
between $\bar{c}_i$ and $\bar{a}_{\sigma i}$, i.e.\
$1/4 (4\theta_i) = \theta_i$.
Hence, the angle between $\hat{C_i}$ and $b_{\rho i}$ is
$\theta_i + 8/5\ \theta_i = 13/5\ \theta_i > 12/5\ \theta_i$.
Hence, $a_{\sigma i}$ appears second on the preference list of $C_i$.

Case(ii) As $c_i$ is the first woman in some $W_{j}$ for $1\leq j
\leq l$, the women positioned next to $c_i$ are $a_{\sigma i}$ and
$b_{x}$
where $c_i,a_{\sigma i}\in W_j$ and
$b_{x}\in W_{j-1}$. Note
that $j-1 \stackrel{\textrm{\tiny def}}{=} l$ if $j = 1$.
The angle between $\bar{c}_i$ and
$\bar{a}_{\sigma i}$ is at most $\epsilon$ and the angle between
$\bar{c}_i$ and
$\bar{b}_{x}$
is the angle between $W_{j-1}$ and $W_j$ which is
$(2\pi-l\epsilon)/l = 2\pi/l - \epsilon > \epsilon$. Hence, $a_{\sigma i}$
appears second on the preference list of $C_i$.

Lastly we examine the preference list of $B_i$. We will show that
the relative order of the $b$- women on the preference list of $B_i$
is $b_n b_{n-1} b_{n-2} \cdots b_1$ for all $1\leq i \leq n$ and that
$B_i$ prefers $b_1$ over any woman $w \notin \{b_1,\ldots,b_n\}$.
This will imply that the preference list of $B_i$ starts with
$b_n b_{n-1}\cdots b_1$ for all $1\leq i \leq n$. The dot product of
$\hat{B}_i$ and $\bar{b}_j$ is
\begin{align*}
\hat{B}_i\cdot \bar{b}_j &= \sin \phi \;\bar{a}_i\cdot\bar{b}_j + \cos \phi \;4^{j}.\\
\textrm{Hence,}\  \cos \phi \;4^{j}  -\sin \phi &\leq  \hat{B}_i\cdot
\bar{b}_j \leq \cos \phi \;4^{j}  + \sin \phi,\; \textrm{and}\\
\cos \phi \;4^{j}  -\phi &\leq  \hat{B}_i\cdot
\bar{b}_j \leq \cos \phi\; 4^{j}  + \phi.
\end{align*}

Comparing $\hat{B}_i\cdot \bar{b}_j$ with $\hat{B}_i\cdot
\bar{b}_{j+1}$, we observe that
\begin{align*}
\hat{B}_i\cdot\bar{b}_j \leq \cos \phi \;4^{j}  + \phi < \cos \phi\;
4^{j+1}  -\phi \leq  \hat{B}_i\cdot \bar{b}_{j+1} \hspace{0.1in}
(\mbox{since } \phi = 2\pi/100, \;\cos \phi > 3/4 ).
\end{align*}
This implies that the relative order of the $b$-women on the
preference list of $B_i$ is $b_n b_{n-1} b_{n-2}$ $\cdots b_1$ for
all $1\leq i \leq n$.

Next we show that $B_i$ prefers $b_1$ over any
woman $w \notin \{b_1,\cdots,b_n\}$. Every woman $w \notin
\{b_1,\cdots,b_n\}$ lies in the $x$-$y$ plane. Hence, it is enough to
consider the projection of $\hat{B}_i$ in the plane which is
$\sin\phi\; \bar{a}_i$. Comparing $\hat{B}_i\cdot \bar{b}_1$ with
$\hat{B}_i\cdot \bar{a}_x$ and $\hat{B}_i\cdot \bar{c}_x$ for
$1\leq x \leq n$, we observe that
\begin{align*}
\hat{B}_i\cdot\bar{b}_1 \geq \cos \phi \;4  - \phi > (3/4)\cdot4
-\phi > 2, \\
\hat{B}_i\cdot \bar{a}_x = \sin \phi \; \bar{a}_i\cdot \bar{a}_x
\leq
\sin \phi \leq \phi < 1 < \hat{B}_i\cdot\bar{b}_1,\; \textrm{and} \\
\hat{B}_i\cdot \bar{c}_x = \sin \phi \; \bar{a}_i\cdot \bar{c}_x
\leq \sin \phi \leq \phi < 1 < \hat{B}_i\cdot\bar{b}_1.
\end{align*}
Hence, we have that the preference list of $B_i$ starts with $b_n
b_{n-1} b_{n-2} \cdots b_1$ for $1\leq i \leq n$. Next we show that
$B_i$ prefers $a_i$ over any woman $w \notin \{a_i,b_1,\cdots,b_n\}$.
Comparing $\hat{B}_i\cdot \bar{a}_i$ with
$\hat{B}_i\cdot \bar{a}_x$, where $x \neq i$ and $\hat{B}_i\cdot
\bar{c}_j$ for $1\leq j \leq n$, we find that
\begin{align*}
\hat{B}_i\cdot\bar{a}_i &= \sin \phi \;\bar{a}_i\cdot\bar{a}_i =
\sin \phi, \\
\hat{B}_i\cdot \bar{a}_x &= \sin \phi \; \bar{a}_i\cdot \bar{a}_x <
\sin \phi = \hat{B}_i\cdot\bar{a}_i \;\;(\mbox{since } \bar{a}_i\cdot \bar{a}_x < 1 \text{ for } i \neq x),\; \textrm{and}\\
\hat{B}_i\cdot \bar{c}_x &= \sin \phi \; \bar{a}_i\cdot \bar{c}_x <
\sin \phi = \hat{B}_i\cdot\bar{a}_i.
\end{align*}

Now the preference list of $B_i$ reads $b_n b_{n-1} b_{n-2} \cdots
b_1 a_i$ for $1\leq i \leq n$. Finally, we consider two cases -
(i) $a_i$ is not the last $a$-woman in any $W_j$ where $1\leq j\leq l$
(ii) $a_i$ is the last $a$-woman in some $W_j$ where $1\leq j\leq l$.

Case (i) Suppose $a_i \in W_j$ for some $j\in\{1,\cdots,l\}$. As
$a_i$ is not the last $a$-woman in $W_j$, the next $a$-woman in
$W_j$ is $a_{\sigma i}$. In other words, $a_i \in T(i)$ and
$a_{\sigma i} \in T(\sigma i)$ where $T(i) = \{c_{\sigma^{-1} i},
a_i, b_{\rho i}\}$ and
$T(\sigma i) = \{c_i, a_{\sigma i}, b_{\rho \sigma i}\}$. The angle
between $\bar{c}_{\sigma^{-1} i}$ and $\bar{a}_i$ is $4\theta_j$, and
that between $\bar{a}_i$ and (the projection of) $\bar{b}_{\rho i}$
is $2\theta_j$. The
angle between $\bar{b}_{\rho i}$ and $\bar{c}_i$ is the angle
between $T(i)$ and $T(\sigma i)$ which is $\theta_j$. Hence, the
angle between $\bar{a}_i$ and $\bar{c}_{i}$ is $2\theta_j + \theta_j
= 3\theta_j$. Note that the projection of the $b$-women onto the unit
circle is irrelevant as they have already been ranked by $B_i$. Hence, we
need only consider the $a$-women and the $c$-women. Given the
placement of the preference vector $\hat{B}_i$, after $a_i$, $B_i$
will prefer either $c_{\sigma^{-1} i}$ or $c_{i}$. Comparing the
dot product of $\hat{B_i}$ with $\bar{c}_{\sigma^{-1}i}$ and with
$\bar{c}_i$, we get
\begin{align*}
\hat{B}_i\cdot\bar{c}_i &= \sin \phi \;\bar{a}_i\cdot\bar{c}_i =
\sin \phi \cos 3\theta_j,\; \textrm{and}\\
\hat{B}_i\cdot \bar{c}_{\sigma^{-1}i} &= \sin \phi \; \bar{a}_i\cdot
\bar{c}_{\sigma^{-1}i} = \sin \phi \cos 4\theta_j < \sin \phi \cos
3\theta_j = \hat{B}_i\cdot\bar{c}_i.
\end{align*}
Hence, the preference list of $B_i$ reads $b_n b_{n-1} b_{n-2}
\cdots b_1 a_i c_i$.

Case (ii) Suppose $a_i \in W_j$ for some $j\in\{1,\cdots,l\}$ {\em and}
$a_i$ is the last $a$-woman in the group. This implies that $a_i \in
T(i) = \{c_{\sigma^{-1} i},a_i,b_{\rho i}\}$ and $T(i)$ is the last
sub-group of $W_j$. Since $T(i)$ is the last sub-group of $W_j$,
$W_j$ starts with the sub-group $T(\sigma i) = \{c_i,a_{\sigma i},
b_{\rho \sigma i}\}$ followed by $T(\sigma^2 i),\cdots,
T(\sigma^{p_j-1} i), T(\sigma^{p_j} i) = T(i)$. The angle subtended
by the group $W_j$ at the origin is $\epsilon$ and
$(2\pi - l\epsilon)/l$
is the angle between two adjacent $W$ groups. Hence, comparing the
dot product of $\hat{B}_i$ with any $a$-woman or $c$-woman in $W_j$
with any $a$-woman or $c$-woman from $W_x$ where $x\neq j$, we
obtain
\begin{align*}
w_1&\in W_j\;\;,\;w_2\in W_x, x\neq j,
\;\;w_1,w_2\notin\{b_1,\cdots,b_n\}\\
\hat{B}_i\cdot\bar{w}_1 &= \sin \phi \;\bar{a}_i\cdot\bar{w}_1 \geq
\sin \phi\; \cos \epsilon,\; \textrm{and} \\
\hat{B}_i\cdot \bar{w}_2 &= \sin \phi \; \bar{a}_i\cdot \bar{w}_2
\leq \sin \phi \;\cos ((2\pi - l\epsilon)/l) < \sin \phi\; \cos
\epsilon \leq \hat{B}_i\cdot\bar{w}_1.
\end{align*}

We conclude that $B_i$ prefers the $a$-women and $c$-women in $W_j$
over any $a$ or $c$-woman in any other group. Within $W_j$, the
$T(\cdot)$ sub-groups occur in the order $T(i), T(\sigma i),
\ldots,$ $T(\sigma^{p_j-1} i)$ and the projection of $\hat{B}_i$
lies inside $T(i)$. We remind the reader that within
$T(\sigma^{m}i)$, where $1\leq m \leq p_j$, the angle between
$\bar{c}_{\sigma^{m-1}i}$ and $\bar{a}_{\sigma^{m}i}$ is
$4\theta_j$, the angle between $\bar{a}_{\sigma^{m}i}$ and
$\bar{b}_{\rho \sigma^{m}i} $ is $2\theta_j$, and that between two
adjacent $T(\cdot)$'s is $\theta_j$, where $(7p_j -1)\theta_j =
\epsilon$. This implies that for $1\leq m \leq p_j$ the angle
between $\bar{a}_{\sigma^{m} i}$ and the projection of $\hat{B}_i$
(which is the angle between $\bar{a}_{\sigma^{m} i}$ and
$\bar{a}_i$) is $(p_j - m)7\theta_j$. For $1\leq m \leq p_j$, the
angle between $\bar{c}_{\sigma^{m-1} i}$ and the projection of
$\hat{B}_i$ (which is the angle between $\bar{c}_{\sigma^{m-1} i}$
and $\bar{a}_i$) is $(p_j - m)7\theta_j + 4\theta_j = (7(p_j - m)+
4)\theta_j $. Now computing the dot product of $\hat{B}_i$ with
$\bar{a}_{\sigma^{m}i}$ and $\bar{c}_{\sigma^{m-1}i}$, we see that
\begin{align*}
\textrm{for}\ 1\leq m &\leq p_j\;\;,\;\;\;\theta_j = \epsilon/(7p_j -1),\; \textrm{we have}\\
\hat{B}_i\cdot\bar{a}_{\sigma^{m}i} &= \sin \phi
     \;\bar{a}_i\cdot\bar{a}_{\sigma^{m}i} = \sin \phi \;\cos
     (7(p_j-m)\theta_j), \\
\hat{B}_i\cdot \bar{c}_{\sigma^{m-1}i} &= \sin \phi \;
     \bar{a}_i\cdot \bar{c}_{\sigma^{m-1}i}
     = \sin \phi \;\cos ((7(p_j - m)+ 4)\theta_j), \\
\hat{B}_i\cdot \bar{c}_{\sigma^{m-1}i} &=
     \sin \phi \;\cos ((7(p_j -m)+ 4)\theta_j)
     < \sin \phi \;\cos (7(p_j-m)\theta_j) \\
  & =\hat{B}_i\cdot\bar{a}_{\sigma^{m}i} \;\;,\;\;1\leq m \leq p_j,\; \textrm{and} \\
\hat{B}_i\cdot\bar{a}_{\sigma^{m}i} &=  \sin \phi \;\cos
     (7(p_j-m)\theta_j) < \sin \phi \;\cos ((7(p_j - m-1)+ 4)\theta_j) \\
  & = \hat{B}_i\cdot \bar{c}_{\sigma^{m}i}\;\;,\;\;1\leq m \leq p_j-1.
\end{align*}

Combining the above inequalities
and using the fact that $i=\sigma^{-1} e_j = \sigma^{p_j-1} e_j$, we
get
\begin{align*}
\hat{B}_i\cdot \bar{c}_{i} & < \hat{B}_i\cdot\bar{a}_{\sigma i}<
\hat{B}_i\cdot\bar{c}_{\sigma i} <\cdots < \hat{B}_i\cdot
\bar{c}_{\sigma^{(p_j-1)}i} <\hat{B}_i\cdot\bar{a}_{\sigma^{p_j}i}
= \hat{B}_i\cdot\bar{a}_{i}.
\end{align*}
Hence, the preference list of $B_i$ is $b_n b_{n-1} b_{n-2} \cdots
b_1 a_i c_{\sigma^{(p_j-1)} i} a_{\sigma^{(p_j-1)} i}\cdots
a_{\sigma^{2} i}c_{\sigma i} a_{\sigma i}c_{i}$.

We remind the reader that $e_i \in Rep(\sigma)$ is a representative element of
cycle $D_i$ and $W_i$ is partitioned into $T(e_i), T(\sigma
e_i),\ldots,T(\sigma^{(p_i-1)} e_i)$ where the sub-groups are
embedded on the unit circle in the order $T(e_i)$ through
$T(\sigma^{(p_i-1)} e_i)$ with $T(\sigma^{(p_i-1)} e_i)$ being the
last sub-group in $W_i$. We now have the initial part of the
preference lists of $A_i,\;C_i$ and $B_i$. They are as follows:
\begin{align}
\label{mennopi}
\textrm{for}\ e_i&\in Rep(\sigma),\\
\nonumber
A_{\sigma^m e_i} \;\;&:\;\;\; a_{\sigma^m e_i} b_{\rho \sigma^m
  e_i}\;\;,\;\;\;0\leq m \leq p_i-1, \\
\nonumber
C_{\sigma^{(m-1)} e_i} \;\;&:\;\;\; c_{\sigma^{(m-1)} e_i}
a_{\sigma^{m} e_i}\;\;,\;\;\;0\leq m \leq p_i-1, \\
\nonumber
B_{\sigma^m e_i} \;\;&:\;\; b_n b_{n-1} b_{n-2} \cdots b_1 a_{\sigma^m
  e_i} c_{\sigma^m e_i}\;\;,\;\;\;0\leq m \leq p_i-2,\; \textrm{and}
\\
\nonumber
B_{\sigma^{(p_i-1)} e_i} \;\;&:\;\; b_n b_{n-1} b_{n-2} \cdots b_1
a_{\sigma^{(p_i-1)} e_i} c_{\sigma^{(p_i-2)} e_i}
a_{\sigma^{(p_i-2)} e_i}\cdots a_{\sigma e_i}c_{e_i}
a_{e_i}c_{\sigma^{(p_i-1)} e_i}.
\end{align}

The partial preference lists of the women can be obtained by
arguments similar to those used for obtaining the men's preference
lists. The partial preference lists for the women are as follows:
\begin{align}
\label{women}
\textrm{for}\ f_i&\in Rep(\rho),\\
\nonumber
b_{\rho^m f_i} \;\;&:\;\;\; A_{\rho^{(m-1)} f_i} B_{\rho^m
  f_i}\;\;,\;\;\;0\leq m \leq q_i-1, \\
\nonumber
c_{\rho^{m} f_i} \;\;&:\;\;\; B_{\rho^{m} f_i} C_{\rho^{m}
  f_i}\;\;,\;\;\;0\leq m \leq q_i-1, \\
\nonumber
a_{\rho^m f_i} \;\;&:\;\; C_{n} C_{n-1} \cdots C_{1}
B_{\rho^m f_i} A_{\rho^m f_i}
\;\;,\;\;\;0\leq m \leq q_i-2,\; \textrm{and}\\
\nonumber
a_{\rho^{(q_i-1)} f_i} \;\;&:\;\;
C_{n} C_{n-1} \cdots C_{1}
\\
\nonumber
&\;\;
B_{\rho^{(q_i-1)} f_i} A_{\rho^{(q_i-2)} f_i} B_{\rho^{(q_i-2)}
f_i}\cdots B_{\rho^{2} f_i}A_{\rho f_i} B_{\rho f_i}
A_{f_i} B_{f_i}
A_{\rho^{(q_i-1)}f_i}.
\end{align}

We note that we have not specified the entire preference lists for the men
and women.  The remaining portion of each preference list appears {\em after} the part
that we have given above, and there will never be any stable pairs
involving a man/woman pair that is not shown on the partial preference lists given.
The partial lists we have given are sufficient to find
the male- and female-optimal matchings, and they contain the necessary
information to generate {\em all} of the stable matchings for our
constructed instance, or equivalently, to find all of the rotations for this
instance.

\subsection{Male- and female-optimal matchings}
\label{sect:mfo}

The rest of our analysis will use the partial preference lists
in~(\ref{mennopi}) and~(\ref{women})
and will not otherwise depend upon the position and preference
vectors.
In order to re-use our analysis in Section~\ref{sect:k-Euclidean}, we
will be
less specific about the men's partial preference
lists~(\ref{mennopi}).
Let $\tau$ be a permutation of $\{1,\ldots,n\}$.
Note that the men's partial preference lists from~(\ref{mennopi})
correspond to the following partial preference lists by taking the
permutation~$\tau$
to be the identity permutation.
\begin{align}
\label{men}
\textrm{For}\ e_i&\in Rep(\sigma),
\\
\nonumber
A_{\sigma^m e_i} \;\;&:\;\;\; a_{\sigma^m e_i} b_{\rho \sigma^m
  e_i}\;\;,\;\;\;0\leq m \leq p_i-1, \\
\nonumber
C_{\sigma^{(m-1)} e_i} \;\;&:\;\;\; c_{\sigma^{(m-1)} e_i}
a_{\sigma^{m} e_i}\;\;,\;\;\;0\leq m \leq p_i-1, \\
\nonumber
B_{\sigma^m e_i} \;\;&:\;\; b_{\tau(n)} b_{\tau(n-1)} \cdots b_{\tau(1)}
a_{\sigma^m
  e_i} c_{\sigma^m e_i}\;\;,\;\;\;0\leq m \leq p_i-2,\; \textrm{and}
  \\
\nonumber
B_{\sigma^{(p_i-1)} e_i} \;\;&:\;\; b_{\tau(n)} b_{\tau(n-1)} \cdots
b_{\tau(1)}
a_{\sigma^{(p_i-1)} e_i} c_{\sigma^{(p_i-2)} e_i}
a_{\sigma^{(p_i-2)} e_i}\cdots a_{\sigma e_i}c_{e_i}
a_{e_i}c_{\sigma^{(p_i-1)} e_i}.
\end{align}

The rest of our analysis will use the partial preference
lists~(\ref{women}) and~(\ref{men}), We will not make any assumptions
about the permutation~$\tau$ even though, for the purposes of this
section, we could assume that it is the identity permutation.

We will find the male-optimal and female-optimal stable matchings
using the Gale-Shapley algorithm. Recall that the order in which the
men proposes does not matter and any order always leads to the
male-optimal matching (provided a man proposes to the highest-ranked
woman (on his preference list) who hasn't yet rejected him).
Therefore, we may suppose the men propose in the order $\{A_1,
\ldots, A_n, C_1, \ldots, C_n, B_{\tau(n)}, \ldots, B_{\tau(1)}\}$.

For $1\leq i \leq n$, men $A_i$ and $C_i$ are paired up with their
first choices, women $a_i$ and $c_i$ respectively, as each of these
women will receive exactly one proposal during the algorithm. Man
$B_{\tau(n)}$ is paired with his first choice, woman $b_{\tau(n)}$.
Man $B_{\tau(n-1)}$ proposes to woman $b_{\tau(n)}$ and gets rejected
as woman $b_{\tau(n)}$ prefers man $B_{\tau(n)}$ over $B_{\tau(n-1)}$.
Man $B_{\tau(n-1)}$ then proposes to woman $b_{\tau(n-1)}$ and gets
accepted. In this manner, man $B_{\tau(i)}$'s proposals to women
$b_{\tau(n)}, b_{\tau(n-1)}, \cdots, b_{\tau(i+1)}$ are all rejected as
woman $b_{\tau(j)}$, $i+1 \leq j \leq n$, prefers man $B_{\tau(j)}$
over man $B_{\tau(i)}$. Hence, $B_{\tau(i)}$ is paired up with woman
$b_{\tau(i)}$ for $1 \leq i \leq n$. Therefore, the male-optimal
matching matches men $A_i$, $C_i$ and $B_i$ with women $a_i$, $c_i$
and $b_i$ for $1 \leq i \leq n$.

We find the female-optimal matching by reversing the roles of men
and women. In other words, women make proposals and men accept or
reject them. Suppose the women propose in the order $\{b_1, \ldots,
b_n, c_1, \ldots, c_n, a_{\sigma n},\ldots, a_{\sigma 1}\}$. Women
$b_i$ and $c_i$ are paired up with their first choices, namely, men
$A_{\rho^{-1} i}$ and $B_i$. Woman $a_{\sigma n}$ is paired with her
first choice, namely, man $C_{n}$. Woman $a_{\sigma (n-1)}$ proposes
to man $C_{n}$ and gets rejected as man $C_{n}$ prefers woman
$a_{\sigma n}$ over $a_{\sigma (n-1)}$. Woman $a_{\sigma (n-1)}$
then proposes to man
$C_{n-1}$
and gets accepted. In
this manner, woman $a_{\sigma i}$'s proposals to men $C_{n},
C_{n-1}, \cdots, C_{i+1}$ are all rejected as man $C_{j}$, $i+1 \leq
j \leq n$, prefers woman $a_{\sigma j}$ over woman $a_{\sigma i}$.
Hence, $a_{\sigma i}$ is paired up with man $C_{i}$ for $1 \leq i
\leq n$. Therefore, the female-optimal matching matches women
$b_i$,$c_i$ and $a_i$ with men $A_{\rho^{-1}i}$, $B_i$ and
$C_{\sigma^{-1}i}$, respectively, for $1 \leq i \leq n$.

As is always the case, as we move from the male-optimal to the
female-optimal matching (by performing a sequence of rotations), the
men go down their preference lists starting from their male-optimal
matching partner (their best possible partner) and ending at their
female-optimal matching partner (their worst possible partner) while
the women go up their preference lists starting from their
male-optimal matching partner (their worst) and ending at their
female-optimal matching partner (their best). Hence, a man will
never be paired with a woman who appears either ahead of his
male-optimal matching partner or after his female-optimal matching
partner on his preference list. Similarly, in any stable matching a
woman will never be paired with a man who appears either ahead of
her female-optimal matching partner or after her male-optimal
matching partner on her preference list. Hence, the only part of a
man's preference list that we need to consider is the sub-list that
starts at the male-optimal matching partner and ends at the
female-optimal matching partner. Similarly, for the women we need to
consider the sub-list starting at the female-optimal matching
partner and ending at the male-optimal matching partner. These
sub-lists are typically referred to as their {\em truncated
preference lists}, and these are as follows:

\begin{align*}
\textrm{For}\ e_i&\in Rep(\sigma)\;\;, \;\;0\leq m \leq p_i-1 \;\;, \;\;1\leq j\leq n, \\
A_{\sigma^m e_i} \;\;&:\;\;\; a_{\sigma^m e_i}
     b_{\rho \sigma^m e_i}\;\;,\\
C_{\sigma^{(m-1)} e_i} \;\;&:\;\;\; c_{\sigma^{(m-1)} e_i}
     a_{\sigma^{m} e_i}\;\;,\\
B_{\tau(j)} \;\;&:\;\; b_{\tau(j)} b_{\tau(j-1)}\cdots b_{\tau(1)}
     a_{\tau(j)} c_{\tau(j)}\;\;,\;\;\;\tau(j) \neq \sigma^{p_i-1} e_i\\
B_{\tau(j)} \;\;&:\;\; b_{\tau(j)} b_{\tau(j-1)}\cdots b_{\tau(1)}
     a_{\tau(j)} c_{\sigma^{-1} \tau(j)} a_{\sigma^{-1} \tau(j)}\cdots\\
&\;\;\;\;\;c_{\sigma \tau(j)} a_{\sigma \tau(j)}
     c_{\tau(j)}\;\;,\;\;\;\tau(j) = \sigma^{p_i-1} e_i
\end{align*}

\begin{align*}
\textrm{For}\ f_i&\in Rep(\rho)\;\;,\\
b_{\rho^m f_i} \;\;&:\;\;\; A_{\rho^{(m-1)} f_i} B_{\rho^m f_i}\;\;,\;\;\;0\leq m \leq q_i-1\\
c_{\rho^{m} f_i} \;\;&:\;\;\; B_{\rho^{m} f_i} C_{\rho^{m} f_i}\;\;,\;\;\;0\leq m \leq q_i-1\\
a_{\rho^m f_i} \;\;&:\;\; C_{\sigma^{-1}(\rho^m f_i)}
     C_{\sigma^{-1}(\rho^m f_i)-1} \cdots C_{1}
      B_{\rho^m f_i} A_{\rho^m f_i}\;\;,\;\;\;0\leq m \leq q_i-2\\
a_{\rho^{(q_i-1)} f_i} \;\;&:\;\; C_{\sigma^{-1}( \rho^{(q_i-1)} f_i)}
     C_{\sigma^{-1}( \rho^{(q_i-1)} f_i)-1} \cdots C_{1}
\\ &\;\;
B_{\rho^{(q_i-1)} f_i} A_{\rho^{(q_i-2)} f_i} B_{\rho^{(q_i-2)}
     f_i}\cdots B_{\rho^{2} f_i}A_{\rho f_i} B_{\rho f_i}
      A_{f_i} B_{f_i}
      A_{\rho^{(q_i-1)}f_i}.
\end{align*}

\subsection{Extracting rotations}\label{sect:extracting-rotations}

We first observe that the male-optimal matching and female-optimal
matching partners are different for every person. This implies that
each man is involved in at least one rotation and, hence, every man
has a well-defined suitor with respect to the male-optimal stable
matching. Also, every man has at least two stable partners. The
truncated preference lists of men $A_i$ and $C_i$, $1 \leq i \leq
n$, are of length two each. Hence, men $A_i$ and $C_i$ are each
involved in exactly one rotation.  Their suitors in the male-optimal
matching ${\cal M}_0$ are $S_{{\cal M}_0}(A_i) = b_{\rho i}$ and
$S_{{\cal M}_0}(C_i) = a_{\sigma i}$, respectively.

\medskip

{\bf Note:}  Throughout this section we use ${\cal M}_0$ to
denote the male-optimal stable matching.

\begin{lemma}\label{suitor-A}
In a stable matching ${\cal M}$, if $A_i$ is paired with $a_i$, then
$S_{{\cal M}}(A_i) = b_{\rho i} = S_{{\cal M}_0}(A_i)$.
\end{lemma}
\begin{proof} The truncated preference lists of man $A_i$ and woman
$b_{\rho i}$ are $a_i b_{\rho i}$ and $A_{i}B_{\rho i}$,
respectively. Since $b_{\rho i}$ is paired with $B_{\rho i}$ in the
male-optimal stable matching ${\cal M}_0$, the spouse of $b_{\rho
i}$ in ${\cal M}$ is a man $M^* \in \{A_i, B_{\rho i}\}$. Since $A_i$
is paired with $a_i$ in ${\cal M}$, $b_{\rho i}$ is paired with
$B_{\rho i}$. Hence, $b_{\rho i}$ prefers $A_i$ over her partner in
${\cal M}$. This, in turn, implies that the suitor of $A_i$ in
${\cal M}$, $S_{{\cal M}}(A_i)$, is $ b_{\rho i}$.
\end{proof}

\begin{lemma}\label{suitor-C}
In a stable matching ${\cal M}$, if $C_i$ is paired with $c_i$, then
$S_{{\cal M}}(C_i) = a_{\sigma i} = S_{{\cal M}_0}(C_i)$.
\end{lemma}
\begin{proof} The truncated preference list of man $C_i$ is $c_i
a_{\sigma i}$. The truncated preference list of woman $a_{\sigma i}$
is either $$C_{i} C_{i-1} \cdots C_{1} B_{\sigma i} A_{\sigma i}$$
$$\textrm{or}\ \ C_{i} C_{i-1}
\cdots C_{1} B_{\sigma i} A_{\rho^{-1} \sigma i} B_{\rho^{-1} \sigma
i} \cdots A_{\rho \sigma i} B_{\rho \sigma i} A_{\sigma i}.$$ We
note that the truncated list of $a_{\sigma i}$ starts with $C_i$. In
${\cal M}_0$, $a_{\sigma i}$ is paired up with $A_{\sigma i}$. This
implies that in the current stable matching ${\cal M}$, $a_{\sigma
i}$ is paired with a man $M^*$ who is as high as $A_{\sigma i}$ on
her preference list. As $C_i$ is paired up with $c_i$ in ${\cal M}$,
$M^* \neq C_i$. Hence, in the current stable matching ${\cal M}$,
$a_{\sigma i}$ prefers $C_i$ over $M^*$. This, in turn, implies that
$S_{{\cal M}}(C_i) = a_{\sigma i}$.
\end{proof}

Next we prove that $S_{{\cal M}_0}(B_i) = a_i$.
\begin{lemma}\label{suitor-B0}
The suitor of man $B_{\tau(i)}$ in  ${\cal M}_0$ is $S_{{\cal
M}_0}(B_{\tau(i)}) = a_{\tau(i)}$.
\end{lemma}
\begin{proof} The truncated preference list of man $B_{\tau(i)}$ depends on
the position of his subscript in the $\rho$-cycle. But we note that
the initial part of the truncated preference list of $B_{\tau(i)}$ is
$b_{\tau(i)} b_{\tau(i-1)} \cdots b_{\tau(1)} a_{\tau(i)}$ for all $i$.
Since our arguments only require the initial part of the truncated
preference list, we do not have to consider separate cases. The
spouse of $B_{\tau(i)}$ in ${\cal M}_0$ is $sp_{{\cal
M}_0}(B_{\tau(i)}) = b_{\tau(i)}$. Suppose the suitor of $B_{\tau(i)}$
is $S_{{\cal M}_0}(B_{\tau(i)}) = b_{\tau(j)}$ for some $j$,  $1 \leq
j <i$. This would imply that $b_{\tau(j)}$ prefers $B_{\tau(i)}$ over
$B_{\tau(j)}$. But the initial part of $b_{\tau(j)}$'s preference list
is $A_{\rho^{-1} {\tau(j)}} B_{\tau(j)} \cdots$. Hence, $b_{\tau(j)}$
prefers $B_{\tau(j)}$ over $B_{\tau(i)}$. This contradicts the
assumption that $S_{{\cal M}_0}(B_{\tau(i)}) = b_{\tau(j)}$. This
would imply that $S_{{\cal M}_0}(B_{\tau(i)}) \neq b_{\tau(j)}$ for
every $j <i$. Since $a_{\tau(i)}$ is paired up with $A_{\tau(i)}$ in
the male-optimal matching ${\cal M}_0$ and $a_{\tau(i)}$ prefers
$B_{\tau(i)}$ over $A_{\tau(i)}$, $S_{{\cal M}_0}(B_{\tau(i)}) =
a_{\tau(i)}$.
\end{proof}

The next two lemmas give the suitor of $B_i$ in stable matchings
which satisfy certain conditions.

\begin{lemma}\label{suitor-B1}
In a stable matching ${\cal M}$, if $C_k$ is paired with $c_k$ for
$1 \leq k \leq n$ and $B_{\tau(i)}$ is paired with $b_{\tau(i)}$, then
$S_{{\cal M}}(B_{\tau(i)}) = a_{{\tau(i)}} = S_{{\cal
M}_0}(B_{\tau(i)})$.
\end{lemma}
\begin{proof} The initial part of the truncated preference list of man
$B_{\tau(i)}$ is $b_{\tau(i)} b_{\tau(i-1)} \cdots b_{\tau(1)}$
$a_{\tau(i)}$. The spouse of $B_{\tau(i)}$ in ${\cal M}$ is
$sp_{{\cal M}}(B_{\tau(i)}) = b_{\tau(i)}$. Suppose the suitor of
$B_{\tau(i)}$ is $S_{{\cal M}}(B_{\tau(i)}) = b_{\tau(j)}$ for some
$j$,  $1 \leq j <i$. This would imply that $b_{\tau(j)}$ prefers
$B_{\tau(i)}$ over $B_{\tau(j)}$. As the initial part of
$b_{\tau(j)}$'s preference list is $A_{\rho^{-1} {\tau(j)}}
B_{\tau(j)} \cdots$ and $b_{\tau(j)}$ is paired with $B_{\tau(j)}$
in the male-optimal stable matching, the partner of $b_{\tau(j)}$ in
the current matching ${\cal M}$ would be a man $M^*$ who is as high
as $B_{\tau(j)}$ on her preference list. Since $b_{\tau(j)}$ prefers
$B_{\tau(j)}$ over $B_{\tau(i)}$, $b_{\tau(j)}$ would prefer $M^*$
over $B_{\tau(i)}$. This contradicts the assumption that $S_{{\cal
M}}(B_{\tau(i)}) = b_{\tau(j)}$. This would entail that $S_{{\cal
M}}(B_{\tau(i)}) \neq b_{\tau(j)}$ for every $j <i$.

Next we will show that $S_{{\cal M}}(B_{\tau(i)}) = a_{\tau(i)}$. In
the male-optimal stable matching ${\cal M}_0$, $a_{\tau(i)}$ is
paired up with $A_{\tau(i)}$. This implies that in the current stable
matching ${\cal M}$, $a_{\tau(i)}$ is paired with a man $M^*$ who is
as high as $A_{\tau(i)}$ on her preference list, i.e.\ $sp_{{\cal
M}}(a_{\tau(i)}) = M^*$. We note that $a_{\tau(i)}$'s truncated
preference list is either
$$C_{\sigma^{-1}(\tau(i))} C_{\sigma^{-1}( \tau(i))-1}
\cdots C_{1} B_{\tau(i)} A_{\tau(i)}$$
$$\textrm{or}\ \ C_{\sigma^{-1}(\tau(i))} C_{\sigma^{-1}( \tau(i))-1}
\cdots C_{1} B_{\tau(i)} A_{\rho^{-1}( \tau(i))} B_{\rho^{-1}(\tau(i))}
\cdots A_{\rho \tau(i)} B_{\rho \tau(i)} A_{\tau(i)}.$$ The initial
part of $a_{\tau(i)}$'s truncated list is $C_{\sigma^{-1}(\tau(i))}
C_{\sigma^{-1}( \tau(i))-1} \cdots C_{1} B_{\tau(i)}$. As $C_k$ is
paired up with $c_k$ for $1 \leq k \leq n$, and $B_{\tau(i)}$ is
paired up with $b_{\tau(i)}$ in the current matching ${\cal M}$, $M^*
\notin \{C_{\sigma^{-1}(\tau(i))}, C_{\sigma^{-1}( \tau(i))-1}, \cdots
C_{1}, B_{\tau(i)}\}$. Hence, in the current matching ${\cal M}$,
$a_{\tau(i)}$ prefers $B_{\tau(i)}$ over her partner $M^*$. This, in
turn, implies that $S_{{\cal M}}(B_{\tau(i)}) = a_{\tau(i)}$.
\end{proof}

\begin{lemma}\label{suitor-B2}
In a stable matching ${\cal M}$, if, for all $k$, woman $a_k$ is
paired with man $M_k$, who is at least as high as $B_k$ on her
preference list, and if $B_{\tau(i)}$ is paired with $a_{\tau(i)}$,
then
$S_{{\cal M}}(B_{\tau(i)}) = c_{\tau(i)} = sp_{{\cal
M}_t}(B_{\tau(i)})$,
where ${\cal M}_t$ is the female-optimal stable matching.
\end{lemma}
\begin{proof} We will consider two cases depending on the preference
list of $B_{\tau(i)}$.

Case (i) The truncated preference list of man $B_{\tau(i)}$ is
$$b_{\tau(i)} b_{\tau(i-1)} \cdots b_{\tau(1)} a_{\tau(i)} c_{\tau(i)}.$$
As the truncated preference list of woman $c_{\tau(i)}$ is
$B_{\tau(i)} C_{\tau(i)}$ and the spouse of $c_{\tau(i)}$ in the
male-optimal stable matching is $C_{\tau(i)}$, $sp_{{\cal
M}}(c_{\tau(i)}) \in \{B_{\tau(i)},C_{\tau(i)}\}$. As the spouse of
$B_{\tau(i)}$ in ${\cal M}$ is $a_{\tau(i)}$ (by assumption),
$sp_{{\cal M}}(c_{\tau(i)}) = C_{\tau(i)}$. Hence, $c_{\tau(i)}$
prefers $B_{\tau(i)}$ over her partner in ${\cal M}$. Therefore, the
suitor of $B_{\tau(i)}$ in ${\cal M}$, $S_{{\cal M}}(B_{\tau(i)}) =
c_{\tau(i)}$.

Case (ii) The truncated preference list of man $B_{\tau(i)}$ is
$$b_{\tau(i)} b_{\tau(i-1)} \cdots b_{\tau(1)} a_{\tau(i)}
c_{\sigma^{-1}(\tau(i))} a_{\sigma^{-1}(\tau(i))}\cdots c_{\sigma
\tau(i)} a_{\sigma \tau(i)} c_{\tau(i)}.$$
As the spouse of $B_{\tau(i)}$.
in ${\cal M}$ is $a_{\tau(i)}$ (by assumption), and $c_{\tau(i)}$ is
the partner of $B_{\tau(i)}$ in the female-optimal stable matching,
$S_{{\cal M}}(B_{\tau(i)}) \in \{c_{\sigma^{-1} (\tau(i))},
a_{\sigma^{-1}(\tau(i))}, \ldots c_{\sigma \tau(i)}, a_{\sigma
\tau(i)}, c_{\tau(i)}\}$.

Suppose $S_{{\cal M}}(B_{\tau(i)}) = c_{\sigma^{k} (\tau(i))} \neq
c_{\tau(i)}$. As the initial part of the preference list of woman
$c_{\sigma^{k} (\tau(i))}$ is $B_{\sigma^{k} (\tau(i))} C_{\sigma^{k}
(\tau(i))}$, we see that $c_{\sigma^{k} (\tau(i))}$ prefers
$C_{\sigma^{k} (\tau(i))}$ over $B_{\tau(i)}$. As the spouse of
$c_{\sigma^{k}(\tau(i))}$ in the male-optimal stable matching is
$C_{\sigma^{k} (\tau(i))}$, then $sp_{{\cal M}}(c_{\sigma^{k}
(\tau(i))}) \stackrel{\textrm{\tiny def}}{=} M^*$ is at least as high
as $C_{\sigma^{k} (\tau(i))}$ on her preference list. Hence,
$c_{\sigma^{k} (\tau(i))}$ prefers her partner in ${\cal M}$ over
$B_{\tau(i)}$. Therefore, $S_{{\cal M}}(B_{\tau(i)}) = c_{\sigma^{k}
(\tau(i))} (\neq c_{\tau(i)})$ is not possible.

Suppose $S_{{\cal M}}(B_{\tau(i)}) = a_{\sigma^{k}(\tau(i))} \neq
a_{\tau(i)}$. The initial part of the truncated preference list of
woman $a_{\sigma^{k} (\tau(i))}$ is $C_{\sigma^{k-1} (\tau(i))}
C_{\sigma^{k-1} (\tau(i)) -1} \cdots C_{1} B_{\sigma^{k} (\tau(i))}$.
As the partner of $a_{\sigma^{k} (\tau(i))}$ in the stable matching
${\cal M}$ is at least as high as $B_{\sigma^{k} (\tau(i))}$, we have
$$M^* \in \{C_{\sigma^{k-1} (\tau(i))}, C_{\sigma^{k} (\tau(i)) -1},
\cdots, C_{1}, B_{\sigma^{k} (\tau(i))}\}.$$
Hence, $a_{\sigma^{k}
(\tau(i))}$ prefers $M^*$ over $B_{\tau(i)}$. Therefore, $S_{{\cal
M}}(B_{\tau(i)}) = a_{\sigma^{k} (\tau(i))} (\neq a_{\tau(i)})$ is not
possible.

Now we will show that $S_{{\cal M}}(B_{\tau(i)}) = c_{\tau(i)}$. As
the truncated preference list of woman $c_{\tau(i)}$ is $B_{\tau(i)}
C_{\tau(i)}$ and the spouse of $c_{\tau(i)}$ in the male-optimal
stable matching is $C_{\tau(i)}$, $sp_{{\cal M}}(c_{\tau(i)}) \in
\{B_{\tau(i)},C_{\tau(i)}\}$. As the spouse of $B_{\tau(i)}$ in ${\cal
M}$ is $a_{\tau(i)}$, $sp_{{\cal M}}(c_{\tau(i)}) = C_{\tau(i)}$.
Hence, $c_{\tau(i)}$ prefers $B_{\tau(i)}$ over her partner in ${\cal
M}$. Therefore, the suitor of $B_{\tau(i)}$ in ${\cal M}$, $S_{{\cal
M}}(B_{\tau(i)}) = c_{\tau(i)}$.
\end{proof}

We will later observe that a stable matching obtained after
performing a set of $\rho$-rotations satisfies conditions laid out
in Lemmas~\ref{suitor-A} and~\ref{suitor-B1}. Hence, the above
lemmas help in establishing the suitors of $A$-men and $B$-men in
the stable matchings obtained after performing $\rho$-rotations.

We note that the {\bf Find-All-Rotations} algorithm obtains all
the rotations of the instance irrespective of whatever ordering of
the men we use in that procedure (to initialize the first proposal
to his suitor). We order the men as follows: $\{A_1, \cdots, A_n,
C_1, \cdots, C_n, B_1,$ $\cdots, B_n\}$. In the male-optimal
matching ${\cal M}_0$, $A_1$ is paired with $a_1$. {\bf
Find-All-Rotations} starts with man $A_1$ whose suitor is $b_{\rho
1}$. The sequence in that algorithm starts with the pair
$(A_1,a_1)$. As $b_{\rho 1}$ is the suitor of $A_1$, the next pair
in the sequence is $(sp_{{\cal M}_0}(b_{\rho 1}),b_{\rho 1}) =
(B_{\rho 1}, b_{\rho 1})$. With $a_{\rho 1}$ being the suitor of
$B_{\rho 1}$, the next pair is $(sp_{{\cal M}_0}(a_{\rho 1}),a_{\rho
1})= (A_{\rho 1},a_{\rho 1})$. The pair $(A_{\rho 1},a_{\rho 1})$
results in $(B_{\rho^2 1},b_{\rho^2 1})$. In this manner, we grow
the sequence $(A_1,a_1)$, $(B_{\rho 1}$, $b_{\rho 1})$, $(A_{\rho
1}$, $a_{\rho 1})$, $(B_{\rho 2}$, $b_{\rho 2}),\cdots $. As the
suitor of $A_i$ is $b_{\rho i}$ and that of $B_i$ is $a_i$, we
observe that the sequence alternates between $A$-men and $B$-men. We
also note that the pair $(A_{1},a_{1})$ results in the pair
$(A_{\rho 1},a_{\rho 1})$ and the pair $(B_{\rho 1},b_{\rho 1})$
results in $(B_{\rho^2 1},a_{\rho^2 1})$. In other words, the
subscripts of the $A$-men and the $B$-men involved in the above
sequence are from a $\rho$-cycle, in particular, the $\rho$-cycle
containing $1$. Suppose the $\rho$-cycle containing $1$ is of size
$p_1$, that is, the $\rho$-cycle containing $1$ is $(1,2,\ldots,
p_1)$. Then the sequence we end up with is
\begin{eqnarray*}
\{(A_1,a_1), (B_{\rho 1}, b_{\rho 1}),
(A_{\rho 1}, a_{\rho 1}), \cdots, (B_{\rho^{p_1-1} 1},
b_{\rho^{p_1-1} 1}), (A_{\rho^{p_1-1} 1}, a_{\rho^{p_1-1} 1}),
     (B_{\rho^{p_1} 1}, b_{\rho^{p_1} 1}) \}  \\
 = \{(A_1,a_1), (B_{2}, b_{2}), (A_{2}, a_{2}), \cdots, (B_{p_1},
b_{p_1}), (A_{p_1}, a_{p_1}), (B_{1}, b_{1})\},
\end{eqnarray*}
using that $(B_{\rho^{p_1} 1}, b_{\rho^{p_1} 1}) = (B_{1}, b_{1})$.
Lemma~\ref{suitor-B1} tells us this ends the sequence in the {\bf
Find-All-Rotations} algorithm, and we have therefore found a
rotation.

(Note:  If we start with any $A_i$ or $B_i$,
with $i\in (1, 2, \ldots, p_1)$, we will discover the
same rotation, as the resulting sequence we find is a cyclic shift
of the one given above.)

After applying the above $\rho$-rotation, $A_i$ is paired with
$b_{\rho i}$ (for $1\leq i \leq p_1$), his partner in the
female-optimal matching. Hence, in this new stable matching
the men $A_1, \ldots, A_{p_1}$ do not have
suitors and will therefore not participate in future rotations. We also note
that the only men who changed their partners in the above rotation
were $A$-men and $B$-men with subscripts in the $\rho$-cycle
containing $1$. Hence, $C_k$ is still paired up with $c_k$ for $1 \leq k
\leq n$, and $A_i$ and $B_i$ are paired up with $a_i$ and $b_i$,
respectively, when the subscript $i$ does {\em not} belong to the $\rho$-cycle
containing $1$.

Let ${\cal M}_1$ denote this new matching after applying
this first $\rho$-cycle we have discovered. We
note that ${\cal M}_1$ satisfies the conditions laid out in
Lemmas~\ref{suitor-A} and~\ref{suitor-B1} which, in turn, tells us that
$S_{{\cal M}_1}(A_i) = b_{\rho i}$ and $S_{{\cal M}_1}(B_i) = a_{i}$
for $i \notin\{1, 2, \cdots, p_1\}$. {\bf Find-All-Rotations} then
picks the next
man who has a well-defined suitor, namely $A_{p_1+1}$, whose suitor
is $S_{{\cal M}_1}(A_{p_1+1}) = b_{\rho (p_1+1)}$, and constructs
the next rotation. From the above exercise of constructing a
rotation corresponding to a $\rho$-cycle, it is clear that the
rotation containing man $A_{p_1+1}$ will be a $\rho$-rotation
involving $A$-men and $B$-men whose subscripts belong to the
$\rho$-cycle that contains $p_1 + 1$.
Proceeding in this manner, we obtain all
rotations ($\rho$-rotations) involving men $A_i$, $1 \leq i \leq n$.

Every $B_i$ will participate in exactly one $\rho$-rotation as every
$i\in [n]$, belongs to exactly one cycle of the permutation $\rho$.
After applying all the $\rho$-rotations, we will obtain a stable
matching, say ${\cal M}'$, in which the spouses of men $A_i$, $B_i$
and $C_i$ are $b_{\rho i}$, $a_i$, and $c_i$, respectively.
As was observed before, all the
$A$-men are paired up with their partners from the female-optimal
stable matching. Thus, none of the $A$-men will participate in any
of the future rotations. As was also noted, the $C$-men each
participate in exactly one rotation and the suitor of $C_i$,
$1 \leq i \leq n$, is $a_{\sigma i}$ as long as $C_i$ is
paired with $c_i$ in the stable matching.
From Lemma~\ref{suitor-B2}, it follows that
the suitor of $B_i$ in ${\cal M}'$ is $c_i$. All of the
$B$-men and $C$-men have well-defined suitors.

The next man picked by {\bf Find-All-Rotations} is $C_1$ whose
suitor is $a_{\sigma i}$. The sequence starts with the pair
$(C_i,c_i)$. As $a_{\sigma i}$ is the suitor of $C_i$, the next pair
in the sequence is $(sp_{{\cal M}'}(a_{\sigma i}), a_{\sigma i}) =
(B_{\sigma i}, a_{\sigma i})$. Similarly, as the suitor of
$B_{\sigma i}$ is $c_{\sigma i}$, the next pair in the sequence is
$(sp_{{\cal M}'}(c_{\sigma i}), c_{\sigma i}) = (C_{\sigma
i},c_{\sigma i})$. Continuing from $(C_{\sigma i},c_{\sigma i})$, we
get the pair $(B_{\sigma^2 i}, a_{\sigma^2 i})$. Proceeding in this
manner, we generate the rotation. We note that the suitor of $B_{i}$
is $c_i$ and that of $C_i$ is $a_{\sigma i}$ thereby forcing us to
alternate between $C$-men and $B$-men. We also note that the pair
$(C_i,c_i)$ eventually results in the pair $(C_{\sigma i}, c_{\sigma
i}$ and $(B_{\sigma i}, a_{\sigma i})$ resulted in the pair
$(B_{\sigma^2 i}, a_{\sigma^2 i})$. In other words, the subscripts
of the $C$-men and the $B$-men in the rotation are governed by a
$\sigma$-cycle, in particular, the $\sigma$ cycle containing $1$.
Suppose the $\sigma$-cycle containing $1$ is of size $q_1$, that is,
the $\sigma$-cycle containing $1$ is $(1, \sigma 1, \ldots,
\sigma^{q_1-1} 1)$. Then the rotation we end up with is
$\{(C_1,c_1), (B_{\sigma 1}, a_{\sigma 1}), (C_{\sigma 1}, c_{\sigma
1}), \cdots, (B_{\sigma^{q_1-1} 1}, a_{\sigma^{q_1-1} 1}),
(C_{\sigma^{q_1-1} 1}, c_{\sigma^{q_1-1} 1}), (B_{\sigma^{q_1} 1},
a_{\sigma^{q_1} 1})\}$,\\ where $(B_{\sigma^{q_1} 1},
a_{\sigma^{q_1} 1})= (B_{1}, a_{1})$.

After performing the above $\sigma$-rotation, for $0 \leq k \leq q_1-1$,
$C_{\sigma^k 1}$, is paired with $a_{\sigma^{k+1} 1}$, his
partner in the female-optimal matching. Hence, men $C_1, \cdots,
C_{\sigma^{q_1-1} 1}$ no longer have suitors and will not participate
in any future rotations. We note that the only men who changed
their partners in the above rotation were $C$-men and $B$-men with
subscripts in the $\sigma$-cycle containing $1$. Hence, $A_k$ is still
paired with $b_{\rho k}$ for $k \in [n]$, and $C_i$ and
$B_i$ are paired up with $c_i$ and $a_i$, respectively, when the
subscript $i$ does not belong to the $\sigma$-cycle containing $1$.

Let ${\cal M}'_1$ denote the new matching after applying this
$\sigma$-rotation.  We note that ${\cal M}'_1$ satisfies the
conditions laid out in Lemmas~\ref{suitor-C}
and~\ref{suitor-B2} which, in turn, entails that
$S_{{\cal M}'_1}(C_i) = a_{\sigma i}$ and
$S_{{\cal M}'_1}(B_i) = c_{i}$ for
$i \notin\{1, \sigma 1, \cdots, \sigma^{q_1-1} 1\}$. {\bf Find-All-Rotations}
picks the next man who has a well-defined suitor, say $C_{i_1}$,
where
$$i_1 = \min\{i: C_i \text{ is paired with } c_i \text{ in }
{\cal M}'_1 \text{ and } 1 \leq i \leq n\},$$ and constructs a new
rotation. The suitor of $C_{i_1}$ is $S_{{\cal M}'_1}(C_{i_1}) =
a_{\sigma i_1}$. From the above exercise of constructing a rotation
corresponding to a $\sigma$-cycle, it is clear that the rotation
containing man $C_{i_1}$ will be a $\sigma$-rotation involving
$C$-men and $B$-men whose subscripts belong to the $\sigma$-cycle
containing $i_1$. Proceeding in this manner, we obtain all
$\sigma$-rotations involving men $C_i$, $1 \leq i \leq n$. Each
$B_i$ will participate in exactly one $\sigma$-rotation, as every
$i\in [n]$ belongs to exactly one cycle of the
permutation $\sigma$.

After applying all the $\sigma$-rotations, we have a stable
matching, say ${\cal M}''$, in which the spouses of men $A_i$, $B_i$
and $C_i$ are $b_{\rho i}$, $c_i$, and $a_{\sigma i}$, respectively.
All the men are paired up with their partners from the
female-optimal stable matching. Hence, ${\cal M}'' = {\cal M}_t$ where
${\cal M}_t$ is the female-optimal stable matching.
Therefore we do not have any further rotations to
extract. Hence, the only rotations in the rotation poset of the
stable matching instance are the rotations governed by the $\rho$-
and $\sigma$-cycles, namely $\rho$-rotations and $\sigma$-rotations.
Therefore, the only stable pairs of the instance are $(A_i,a_i)$,
$(A_i,b_{\rho i})$, $(B_i,b_i)$, $(B_i,a_i)$, $(B_i,c_i)$,
$(C_i,c_i),$ and $(C_i,a_{\sigma i})$ for $i \in [n]$.

We still have to prove that the $\rho$-rotations correspond to
vertices in $V_1$ and the $\sigma$-rotations correspond to vertices
in $V_2$. We prove this fact in the next section.

\subsection{Ordering rotations}
In this section, we compare rotations using the explicitly precedes
relation as in Definition~\ref{def:explicitly}.

Recalling Definition~\ref{eliminated}, it follows that a man-woman pair
$(M,w)$ can be eliminated if and only if there exist stable pairs
$(M_1,w)$ and $(M_2,w)$ such that $w$ prefers $M_1$ over $M$ and $M$
appears as high as $M_2$ on $w$'s preference list. In other words, a
man-woman pair $(M,w)$ can be eliminated if and only if $M$ appears
on the truncated preference list of $w$ and is not the partner of
$w$ in the female-optimal stable matching. This identifies all the
man-woman pairs eliminated by rotations of the instance.

Recall from the previous section that the only stable pairs of the matching
instance
are $(A_i,a_i), (A_i,b_{\rho i}), (B_i,b_i), (B_i,a_i), (B_i,c_i),
(C_i,c_i), (C_i,a_{\sigma i})$ for $i \in [n]$. Of these
stable pairs, $(A_i,b_{\rho i}), (B_i,c_i)$, and $(C_i,a_{\sigma
i})$ are pairs in the female-optimal stable
matching. Hence, the only stable pairs that are eliminated by
rotations are $(A_i,a_i), (B_i,b_i), (B_i,a_i)$ and $(C_i,c_i)$ for
$i \in [n]$. We list all the eliminated pairs below and
highlight those that are stable.

\begin{align*}
\textrm{For } f_i&\in Rep(\rho)\;\;,\\
b_{\rho^m f_i} \;\;&:\;\;\;
{\bf B_{\rho^m f_i}}\;\;,\;\;\;0\leq m \leq q_i-1\\
c_{\rho^{m} f_i} \;\;&:\;\;\;
{\bf C_{\rho^{m} f_i}}\;\;,\;\;\;0\leq m \leq q_i-1\\
a_{\rho^m f_i} \;\;&:\;\; C_{\sigma^{-1} (\rho^m f_i -1)} \cdots
C_{1}
{\bf B_{\rho^m f_i} A_{\rho^m f_i}}\;\;,\;\;\;0\leq m \leq q_i-2\\
a_{\rho^{(q_i-1)} f_i} \;\;&:\;\; C_{\sigma^{-1} (\rho^{(q_i-1)}
f_i-1)} \cdots C_{\sigma^{-1}1} {\bf B_{\rho^{(q_i-1)} f_i}}
A_{\rho^{(q_i-2)} f_i} B_{\rho^{(q_i-2)} f_i}\cdots B_{\rho^{2}
f_i}A_{\rho f_i} B_{\rho f_i} {\bf A_{\rho^{(q_i-1)}f_i}}
\end{align*}

In Definition~\ref{def:explicitly}, rotation $R$ eliminates pair
$(M,w)$ and rotation $R'$ moves man $M$ to woman $w'$ such that $M$
prefers $w$ over $w'$. Hence, woman $w$ and man $M$ belong to
rotations $R$ and $R'$, respectively.

\begin{lemma}\label{rho-minimal}
Suppose $R'$ is a $\rho$-rotation. Then there does not exist a
rotation $R$ which explicitly precedes $R'$. Therefore, every
$\rho$-rotation is a minimal element of the rotation poset.
\end{lemma}
\begin{proof} Suppose there exists a rotation $R$ which explicitly
precedes $$R' = \{(B_{j},b_j), (A_{j},a_{j}), (B_{j+1},b_{j+1}),
(A_{j+1},a_{j+1}), \cdots, (B_{j+q-1},b_{j+q-1}),
(A_{j+q-1},a_{j+q-1})\}.$$ We consider two cases - (I) $R$ is a
$\rho$-rotation, (II) $R$ is a $\sigma$-rotation.

Case (I) Suppose $R$ is a $\rho$-rotation where
$$R = \{ (B_{i},b_i),
(A_{i},a_i), (B_{i+1},b_{i+1}), (A_{i+1},a_{i+1}), \ldots,
(B_{i+p-1},b_{i+p-1}), (A_{i+p-1},a_{i+p-1})\}.$$ The $\rho$-cycles
corresponding to rotations $R$ and $R'$ are $(i,i+1,\ldots,i+p-1)$
and $(j,j+1,\ldots,j+q-1)$. Since any two $\rho$-cycles are
disjoint, the corresponding $\rho$-rotations are disjoint, i.e.\
$\rho$-rotations $R$ and $R'$ do not share either a man or a woman.
Since $R$ explicitly precedes $R'$, there exists a man-woman pair
$(M,w)$ with $M$ belonging to rotation $R'$ and $w$ belonging to
rotation $R$ such that $R$ eliminates the pair $(M,w)$ and $R'$
moves $M$ to a woman $w'$ below $w$ on his list. In other words,
there exist
\begin{eqnarray*}
\textrm{a man } M & \in & \{B_j, B_{j+1}, \ldots, B_{j+q-1}, A_j,
A_{j+1}, \ldots, A_{j+q-1}\}, \textrm { and} \\
\textrm{a woman } w & \in & \{b_i, b_{i+1},
\cdots, b_{i+p-1}, a_i, a_{i+1}, \cdots, a_{i+p-1}\}
\end{eqnarray*}
satisfying
the above property.  We consider a set of sub-cases, depending
upon possible values of $M$ and $w$.

Subcase (I-a) $(M, w) \in \{(B_x, b_y), (A_x, b_y)\}$.
We note that $x \neq y$,
and $x$ and $y$ are from different $\rho$-cycles. From the table of
eliminated pairs, we note that the set of eliminated pairs involving
woman $b_y$ is $\{(B_y,b_y)\}$. Hence, $(M,w) \neq (B_x,b_y)$ and
$(M, w) \neq (A_x,b_y)$.

Subcase (I-b) $(M,w) = (A_x,a_y)$. We again note that $x \neq y$,
and $x$ and $y$ are from different $\rho$-cycles. We also note that
woman $a_y$ could have one of two possible preference lists which is
reflected in the table of eliminated pairs. After performing
rotation $R$, woman $a_y$ is paired up with $B_y$. Hence, the set of
pairs eliminated by rotation $R$ involving woman $a_y$ could be
either
\begin{eqnarray*}
S = \{(A_{y},a_y)\} \ \ \textrm{ or} \\
T = \{(A_{\rho^{-1} y},a_y),
(B_{\rho^{-1} y},a_y), \ldots, (A_{\rho y},a_y), (B_{\rho y},a_y),
(A_{y},a_y)\}.
\end{eqnarray*}
Since we are only interested in eliminated pairs that involve an
$A$-man, we consider subsets of $S$ and $T$ containing $A$-men which
are $\{(A_y,a_y)\}$ and $\{(A_{\rho^{-1} y},a_y), (A_{\rho^{-2}
y},a_y), \cdots,$ $(A_{\rho y},a_y), (A_{y},a_y)\}$, respectively.
We note that every element of $\{\rho^{-1} y, \rho^{-2} y, \cdots,
\rho y, y\}$ belongs to the $\rho$-cycle containing $y$. Since $x
\neq y$ and $x$ and $y$ are from different $\rho$-cycles, $(A_x,a_y)
\notin \{(A_y,a_y)\}$ and $(A_x,a_y) \notin \{(A_{\rho^{-1} y},a_y),
(A_{\rho^{-2} y},a_y), \ldots, (A_{\rho y},a_y), (A_{y},a_y)\}$.
Hence, $(M,w) \neq (A_x,a_y)$.

Subcase (I-c) $(M,w) = (B_x,a_y)$. As before, $x \neq y$ and $x$,
and $y$ are from different $\rho$-cycles. From the table of
eliminated pairs, we note that the set of eliminated pairs involving
woman $a_y$ could be either
\begin{eqnarray*}
S & = & \{(A_{y},a_y)\} \ \textrm{ or} \\
T & = & \{(A_{\rho^{-1} y},a_y), (B_{\rho^{-1} y},a_y), \ldots, (A_{\rho
y},a_y), (B_{\rho y},a_y), (A_{y},a_y)\}.
\end{eqnarray*}
Since we are only
interested in eliminated pairs that involve a $B$-man, we consider
subsets of $S$ and $T$ containing $B$-men which are $\emptyset$ and
$\{(B_y,a_y), (B_{\rho^{-1} y},a_y), (B_{\rho^{-2} y},a_y), \ldots,
(B_{\rho y},a_y)\}$, respectively. We note that every element of
$\{y, \rho^{-1} y, \rho^{-2} y, \cdots, \rho y\}$ belongs to the
$\rho$-cycle containing $y$. Since $x \neq y$ and $x$ and $y$ are
from different $\rho$-cycles,
$$(B_x,a_y) \notin \{(B_y,a_y),
(B_{\rho^{-1} y},a_y), (B_{\rho^{-2} y},a_y), \cdots, (B_{\rho
y},a_y)\}$$
and $(B_x,a_y) \notin \emptyset$ (vacuously). Hence, $(M,w)
\neq (B_x,a_y)$.

Therefore, if $R$ explicitly precedes $R'$, then $R$ is not a
$\rho$-rotation.

Case (II) Suppose $R$ is a $\sigma$-rotation where
$$R = \{ (B_{i},a_i),
(C_{i},c_i), (B_{\sigma i},a_{\sigma i}), (C_{\sigma i},c_{\sigma
i}), \ldots, (B_{\sigma^{p-1} i},a_{\sigma^{p-1} i}),
(C_{\sigma^{p-1} i},c_{\sigma^{p-1} i})\}.$$ The $\sigma$- and
$\rho$-cycles corresponding to rotations $R$ and $R'$ are
$\sigma_1 = (i,\sigma i,\ldots,\sigma^{p-1}i)$ and
$\rho_1 = (j,j+1,\ldots,j+q-1)$, respectively.
Since $R$ explicitly precedes $R'$, there exists a man-woman pair
$(M,w)$ with $M$ belonging to rotation $R'$ and $w$ belonging to
rotation $R$ such that $R$ eliminates the pair $(M,w)$ and $R'$
moves $M$ to a woman $w'$ below $w$ on his list. In other words,
there exist a man $M \in \{B_j, B_{j+1}, \ldots, B_{j+q-1}, A_j,
A_{j+1}, \ldots, A_{j+q-1}\}$ and a woman $w \in \{a_i, a_{\sigma
i}, \ldots, a_{\sigma^{p-1}i}, c_i, c_{\sigma i}, \ldots,
c_{\sigma^{p-1}i}\}$ satisfying the above property. Once again,
there are a set of sub-cases to consider.

Subcase (II-a) $(M,w) \in \{(B_x,c_y), (A_x,c_y)\}$.
From the table of eliminated pairs, we note that the set of
eliminated pairs involving woman $c_y$ is $\{(C_y,c_y)\}$. Hence,
$(M, w) \neq (B_x, c_y)$ and $(M, w) \neq (A_x, c_y)$.

Subcase (II-b) $(M,w) \in \{(B_x,a_y), (A_x,a_y)\}$. We note that
after performing rotation $R'$, woman $a_y$ is paired up with $B_y$.
Hence, the pair $(B_y,a_y)$ cannot be an eliminated pair. This
implies that when the eliminated pair $(M,w)$ is of the form
$(B_x,a_y)$, the subscript $x$  cannot assume the value $y$. We
observe that even though woman $a_y$ could have one of two possible
preference lists, the initial part of her preference list stays the
same. We note that before performing the rotation $R$ woman $a_y$ is
paired up with man $B_y$. After performing the rotation $R$, woman
$a_y$ is paired up with $C_{\sigma^{-1} y}$. Hence, the pairs
eliminated by rotation $R$ that involve woman $a_y$ are
$\{(C_{\sigma^{-1}(y)-1}, a_y), (C_{\sigma^{-1}(y)-2}, a_y), \ldots,
(C_{1}, a_y), (B_y,a_y)\}$. Since
$$\{(B_x,a_y), (A_x,a_y)\} \cap
\{(C_{\sigma^{-1}(y)-1}, a_y), (C_{\sigma^{-1}(y)-2}, a_y), \ldots,
(C_{1}, a_y), (B_y,a_y)\} = \emptyset,$$
we see $(M,w) \neq (B_x,a_y)$ and $(M, w) \neq (A_x,a_y)$.

Therefore, if $R$ explicitly precedes $R'$, then $R$ cannot be a
$\sigma$-rotation.

From cases (I) and (II), we conclude that a $\rho$-rotation cannot
be explicitly preceded either by another $\rho$-rotation or a
$\sigma$-rotation. Therefore, every $\rho$-rotation is a minimal
element in the rotation poset.
\end{proof}

\begin{lemma}\label{sigma-maximal}
Suppose $R$ is a $\sigma$-rotation. Then there does not exist a
rotation $R'$ such that $R$ explicitly precedes $R'$. Therefore,
every $\sigma$-rotation is a maximal element of the rotation poset.
\end{lemma}
\begin{proof}

Suppose there exists a rotation $R$ which explicitly
precedes $R'$.
We consider two cases.

Case (I) $R'$ is a $\rho$-rotation. This case has been dealt with in
case (II) of Lemma~\ref{rho-minimal}.

Case(II) $R'$ is a $\sigma$-rotation. Let
$$R = \{ (B_{i},a_i), (C_{i},c_i), (B_{\sigma i},a_{\sigma i}),
(C_{\sigma i},c_{\sigma i}), \ldots, (B_{\sigma^{p-1} i},a_{\sigma^{p-1} i})
(C_{\sigma^{p-1} i},c_{\sigma^{p-1} i})\},$$
and
$$R' = \{(B_{j},a_j), (C_{j},c_{j}),(B_{\sigma j},a_{\sigma j}), (C_{\sigma j},c_{\sigma j}), \ldots,
(B_{\sigma^{q-1}j},a_{\sigma^{q-1}j}),
(C_{\sigma^{q-1}j},c_{\sigma^{q-1}j})\}.$$
The $\sigma$-cycles
corresponding to rotations $R$ and $R'$ are
$(i,\sigma i,\ldots,\sigma^{p-1} i)$ and
$(j,\sigma j,\ldots,\sigma^{q-1}j)$.
Since any two $\sigma$-cycles are disjoint, the corresponding
$\sigma$-rotations are disjoint, i.e.\ $\sigma$-rotations $R$ and
$R'$ do not share either a man or a woman. As has been observed
before, the implication of $R$ explicitly preceding $R'$ is that
there exists a man-woman pair $(M,w)$ with $M$ belonging to rotation
$R'$ and $w$ belonging to rotation $R$ such that $R$ eliminates the
pair $(M,w)$ and $R'$ moves $M$ to a woman $w'$ below $w$ on his
list. In other words, there exist
\begin{eqnarray*}
\textrm{a man } M & \in & \{B_j, B_{\sigma j},
\ldots, B_{\sigma^{q-1} j}, C_j, C_{\sigma j}, \ldots,
C_{\sigma^{q-1} j}\} \textrm{ and} \\
\textrm{ a woman } w & \in &
   \{a_i, a_{\sigma i}, \ldots, a_{\sigma^{p-1}i}, c_i, c_{\sigma i},
   \ldots, c_{\sigma^{p-1} i}\}
\end{eqnarray*}
satisfying the above property. As the
pair $(M,w)$ has a set of possibilities, we consider a set of
sub-cases.

Subcase (II-a) $(M,w) \in \{(B_x,c_y), (C_x,c_y)\}$.
We note that $x \neq y$,
and $x$ and $y$ are from different $\sigma$-cycles. From the table
of eliminated pairs, we note that the set of eliminated pairs
involving woman $c_y$ is $\{(C_y,c_y)\}$. Hence, $(M,w) \neq
(B_x,c_y)$ and $(M, w) \neq (C_x,c_y)$.

Subcase (II-b) $(M,w) = (B_x,a_y)$. As before, we note that $x \neq
y$ and $x$ and $y$ are from different $\sigma$-cycles. We also note
that the spouses of woman $a_y$ before and after performing the
rotation $R$ are $B_y$ and $C_{\sigma^{-1} y}$, respectively. Hence,
the pairs eliminated by rotation $R$ are
$$\{(C_{\sigma^{-1}(y)-1},a_y), (C_{\sigma^{-1}(y)-2},a_y), \ldots,
(C_{1},a_y), (B_{y},a_y)\}.$$ As $x \neq y$, $(B_x,a_y)\notin
\{(C_{\sigma^{-1}(y-1)},a_y), (C_{\sigma^{-1}(y-2)},a_y), \ldots,
(C_{\sigma^{-1}1},a_y), (B_{y},a_y)\}$. Therefore, $(M,w) \neq
(B_x,a_y)$.

Subcase (II-c) $(M,w) = (C_x,a_y)$. We note that $x$ and $y$ are
from different $\sigma$-cycles. We also note that the spouses of
woman $a_y$ before and after performing the rotation $R$ are $B_y$
and $C_{\sigma^{-1} y}$, respectively. Hence, the pairs eliminated
by rotation $R$ are
$$\{(C_{\sigma^{-1}(y)-1},a_y), (C_{\sigma^{-1}(y)-2},a_y), \ldots,
(C_{1},a_y), (B_{y},a_y)\}.$$

Suppose $(C_x,a_y)\in \{(C_{\sigma^{-1}(y)-1},a_y),
(C_{\sigma^{-1}(y)-2},a_y), \ldots, (C_{1},a_y), (B_{y},a_y)\}$.
This implies $x \in \{\sigma^{-1}(y)-1, \sigma^{-1}(y)-2, \ldots,
1\}$. Recalling that the only
stable pairs are $(A_i,a_i)$, $(A_i,b_{\rho i})$, $(B_i,b_i)$,
$(B_i,a_i)$, $(B_i,c_i)$, $(C_i,c_i)$, $(C_i,a_{\sigma i})$
for $i \in [n]$, we see that $(C_x,a_y)$ is an unstable pair
as $x \neq \sigma^{-1} y$.
Therefore, every pair of the form $(C_x,a_y)$ that
rotation $R$ eliminates is an unstable pair. Since $R$ explicitly
precedes $R'$ and the pairs of the form $(C_x,a_y)$ eliminated by
$R$ are unstable, rotation $R'$ has to move some $C_x$ below $a_y$
on his list. The initial part of the preference list of $C_x$ is
$c_x a_{\sigma x}$ for all $x \in\{1,2,\cdots,n\}$. In other words,
$C_x$ has $a_{\sigma x}$ above $a_y$ for all $y \neq \sigma x$.
After performing rotation $R'$, $C_x$ moves from $c_x$ to $a_{\sigma
x}$ for every $x \in \{j, \sigma j, \cdots, \sigma^{q-1} j\}$.
Therefore, rotation $R'$ does not take $C_x$ below $a_y$ on his
preference list for $x \in \{j, \sigma j, \cdots, \sigma^{q-1} j\}$.
Therefore, $(M,w) \neq (C_x,a_y)$.

Putting together cases (I) and (II), we conclude that if $R$ is a
$\sigma$-rotation, then it cannot explicitly precede any rotation
$R'$. This, in turn, implies that every $\sigma$-rotation is a
maximal element in the rotation poset.
\end{proof}

From Lemmas~\ref{rho-minimal} and~\ref{sigma-maximal}, it follows
that the rotation poset of our constructed matching instance is of
height at most $1$.

\begin{lemma}\label{rho-sigma}
Suppose $R$ is a $\rho$-rotation and $R'$ is a $\sigma$-rotation.
Then $R$ explicitly precedes $R'$ if and only if $R$ and $R'$ have a
common man. In other words, the $\rho$ and $\sigma$-cycles
corresponding to $R$ and $R'$, respectively, have an element in
common.
\end{lemma}
\begin{proof} Let
$$R = \{ (B_{i},b_i), (A_{i},a_i), (B_{i+1},b_{i+1}),
(A_{i+1},a_{i+1}), \ldots, (B_{i+p-1},b_{i+p-1}), (A_{i+p-1},a_{i+p-1})\}$$
and
$$R' = \{(B_{j},a_j), (C_{j},c_{j}), (B_{\sigma j},a_{\sigma j}),
(C_{\sigma j},c_{\sigma j}), \ldots, (B_{\sigma^{q-1}j},a_{\sigma^{q-1}j}),
(C_{\sigma^{q-1}j},c_{\sigma^{q-1}j})\}.$$

Suppose $R$ and $R'$ do not have a common man, i.e.\
$$\{B_i, B_{i+1}, \ldots, B_{i+p-1}\} \cap \{B_j, B_{\sigma
j}, \ldots, B_{\sigma^{q-1} j}\} = \emptyset.$$
This, in turn,
entails that $\rho_1 \cap \sigma_1 = \emptyset$ where
$\rho_1 = \{i, i+1, \ldots, i+p-1\}$ and
$\sigma_1 = \{j, \sigma j, \ldots, \sigma^{q-1} j\}$.
As has been observed before, the implication of
$R$ explicitly preceding $R'$ is that there exists a man-woman pair
$(M,w)$ with $M$ belonging to rotation $R'$ and $w$ belonging to
rotation $R$ such that $R$ eliminates the pair $(M,w)$ and $R'$
moves $M$ to a woman $w'$ below $w$ on his list. In other words,
there exist
\begin{eqnarray*}
\textrm{a man } M & \in & \{B_j, B_{\sigma j}, \ldots,
B_{\sigma^{q-1} j}, C_j, C_{\sigma j}, \ldots, C_{\sigma^{q-1} j}\}
\textrm{ and} \\
\textrm{a woman } w & \in &
\{b_i, b_{i+1}, \ldots, b_{i+p-1}, a_i, a_{i+1},
\ldots, a_{i+p-1}\}
\end{eqnarray*}
satisfying the above property.
We consider a set of sub-cases.

Case (I) $(M,w) \in \{ (B_x,b_y), (C_x,b_y)\}$. We note that $x \neq
y$ as $\rho_1 \cap \sigma_1 = \emptyset$ and $x\in \sigma_1$ and
$y\in \rho_1$. From the table of eliminated pairs, we note that the
set of eliminated pairs involving woman $b_y$ is $\{(B_y,b_y)\}$.
Hence, $(M,w) \not\in \{(B_x,b_y), (C_x,b_y)\}$.

Case (II) $(M,w) \in  \{(B_x,a_y), (C_x,a_y)\}$.
We note that $a_y$ could have one of two possible preference
lists which is
reflected in the table of eliminated pairs.

The spouses of $a_y$
before and after performing the rotation $R$ are $A_y$ and $B_y$,
respectively. Hence, the set of pairs eliminated by rotation $R$
involving woman $a_y$ could be either $S = \{(A_{y},a_y)\}$ or
$T = \{(A_{\rho^{-1} y},a_y), (B_{\rho^{-1} y},a_y), \ldots, (A_{\rho
y},a_y), (B_{\rho y},a_y), (A_{y},a_y)\}$. Since none of the
eliminated pairs involve a $C$-man, $(M,w) \neq (C_x,a_y)$.

With $(C_x,a_y)$ eliminated from being a possible candidate, we are only
interested in eliminated pairs that involve a $B$-man. We consider
subsets of $S$ and $T$ containing $B$-men which are $\emptyset$ and
$\{(B_{\rho^{-1} y},a_y), (B_{\rho^{-2} y},a_y), \ldots, (B_{\rho
y},a_y)\}$, respectively. We note that every element of
$\rho_1 = \{\rho^{-1} y, \rho^{-2} y, \cdots, \rho y, y\} =
\{i, i+1, \cdots, i+p-1\}$ belongs to the $\rho$-cycle
containing $y$. As $\sigma_1$
and $\rho_1$ do not have an element in common, $x \neq y$ and $x$
does not belong to the $\rho$-cycle containing $y$. Therefore,
$(B_x,a_y) \notin \{(B_{\rho^{-1} y},a_y), (B_{\rho^{-2} y},a_y),
\ldots, (B_{\rho y},a_y)\}$. Hence, $(M,w) \neq (B_x,a_y)$.

From cases (I) and (II), it follows that if $R$ and $R'$ do not share
a common man, then there does not exist a man-woman pair $(M,w)$
such that $R$ eliminates $(M,w)$ and $R'$ moves $M$ to $w'$ below
$w$. Hence, $R$ does not explicitly precede $R'$.

Suppose rotations $R$ and $R'$ have a common man. In other words,
$$\{B_i, B_{i+1}, \ldots, B_{i+p-1}\} \cap \{B_j, B_{\sigma j},
\ldots, B_{\sigma^{q-1} j}\} \neq \emptyset.$$
This, in turn, entails
that $\rho_1 \cap \sigma_1 \neq \emptyset$ where
$\rho_1 = \{i, i+1, \ldots, i+p-1\}$ and
$\sigma_1 = \{j, \sigma j, \ldots, \sigma^{q-1}
j\}$. We also note that a $\rho$-cycle and a $\sigma$-cycle can have
at most one element in common. Therefore, $\{B_i, B_{i+1}, \ldots,
B_{i+p-1}\} \cap \{B_j, B_{\sigma j}, \cdots, B_{\sigma^{q-1} j}\} =
\{B_l\}$, say.

As has been observed before, in order to establish that $R$
explicitly precedes $R'$, it is enough to produce a man-woman pair
$(M,w)$ with $M$ belonging to rotation $R'$ and $w$ belonging to
rotation $R$ such that $R$ eliminates the pair $(M,w)$ and $R'$
moves $M$ to a woman $w'$ below $w$ on his list. In other words, it
is enough to show that there exist a man $M \in \{B_j, B_{\sigma j},
\ldots, B_{\sigma^{q-1} j}, C_j, C_{\sigma j}, \ldots,
C_{\sigma^{q-1} j}\}$ and a woman $w \in \{b_i, b_{i+1}, \ldots,
b_{i+p-1}, a_i, a_{i+1} \ldots,$ $a_{i+p-1}\}$ satisfying the above
property. We show that the pair $(B_l,b_l)$ is the required pair.

Since $B_l$ participates in the rotation $R$, the spouses of $B_l$
before and after the rotation $R$ are $b_l$ and $a_l$, respectively.
Hence, the pair $(B_l,b_l)$ is eliminated by $R$, and $R$ moves $B_l$
to $a_l$, which is below $b_l$. Since $B_l$ belongs to $R'$, the
spouses of $B_l$ before and after performing the rotation $R'$ are
$a_l$ and $c_l$, respectively. Hence, $R'$ moves $B_l$ to $c_l$
which is below $a_l$. This entails that rotation $R'$ moves $B_l$ to
$c_l$ which is below $b_l$ on his preference list. Therefore, $R$
eliminates the pair $(B_l,b_l)$ and $R'$ moves $B_l$ to the woman
$c_l$ who is below $a_l$ on $B_l$'s preference list. This implies
that $R$ explicitly precedes $R'$.

Hence, it follows that if $R$ and $R'$ share a man, then $R$
explicitly precedes $R'$. This proves the lemma.
\end{proof}

From Lemma~\ref{rho-sigma}, it follows that the rotation poset has
an edge from a $\rho$-rotation to a $\sigma$-rotation (i.e.\ $\rho
\leq \sigma$ in the ordering of the rotations) if and only if the
rotations share a common man. Hence, the rotation poset has height
1. In other words, the rotation poset has an edge between two
vertices if and only if the cycles corresponding to the vertices
have an element in common. The edges of the bipartite graph, which
was introduced early on, were defined in a similar fashion. Hence,
the rotation poset when considered as a graph is isomorphic to a
bipartite graph.

 \section{The 1-attribute case}\label{sect:1-attribute}

In this section we concentrate our attention to the $1$-attribute
model.  This case is very special and we establish the following
result.

\setcounter{counter:save}{\value{theorem}}
\setcounter{theorem}{\value{counter:1-attribute}}
\begin{theorem}
$\#SM(1\rm{-attribute})$ is solvable in polynomial time.
\end{theorem}
\setcounter{theorem}{\value{counter:save}}

Theorem~\ref{thm:1-attribute} is obtained as a corollary of the
following theorem.

\setcounter{counter:1-attr-rot}{\value{theorem}}
\begin{theorem}\label{thm:1-attr-rot}
 In the $1$-attribute model, the rotation poset of a stable matching instance
is (isomorphic to) a path.
\end{theorem}

Theorem~\ref{thm:1-attribute} follows as it is straightforward to
count the downsets of a path.
We establish Theorem~\ref{thm:1-attr-rot} through a series of lemmas.

First, we observe that in the $1$-attribute model the men have only two
possible preference lists for the women. The two
preference lists are such that one is reverse of the other.  Similarly,
the women have only two possible preference lists, one being the reverse
of the other.

We start by establishing that every rotation in the 1-attribute
model is of even size. In other words, every rotation involves an
even number of men.

\begin{lemma}
In the $1$-attribute model, every rotation is of even size,
and the preference lists of the men (and, similarly, the women)
involved in the rotation alternate.
\end{lemma}

\begin{proof} We establish the statement by showing that two
consecutive men in a rotation cannot both have the same preference
list. Then, since there are only two possible preference lists and the
preference lists of any two consecutive men have to be different, we
conclude that the number of men
involved in the rotation has to be even.

Suppose we have a rotation $R$ of size $k$.  Without loss of generality
(by relabeling), we assume this rotation is
$(B_0, b_0), (B_1, b_1), \ldots, (B_{k-1}, b_{k-1})$. Each man $B_i$ is
married to woman $b_i$ (in $M_1$) before the rotation, and to woman
$b_{i+1\pmod k}$ (in $M_2$) after the rotation as shown in the table below.
\begin{center}
\begin{tabular}{c c c}
\hline \hline

Men &  (before $R$) &  (after $R$) \\
    & $M_1$  &  $M_2$  \\ [0.5ex]

\hline
$B_0$ &$b_0$&$b_1$\\
$B_1$ &$b_1$&$b_2$\\
\vdots \\
$B_i$ &$b_i$&$b_{i+1}$\\
$B_{i+1}$ &$b_{i+1}$&$b_{i+2}$\\
\vdots\\
$B_{k-2}$&$b_{k-2}$&$b_{k-1}$\\
$B_{k-1}$&$b_{k-1}$&$b_0$\\
\hline
\end{tabular}
\end{center}

Before we proceed, we note that all subscripts that follow
are computed $\!\!\!\mod k$.

To establish our result it is enough to show that two consecutive men
in a rotation cannot both have the same preference lists.

So, suppose to the contrary that men $B_i$ and $B_{i+1}$ have the same
preference list. Recalling that the rotation is female-improving,
the preference list of $B_i$ and $B_{i+1}$
are shown below.

\begin{center}
\begin{tabular}{r c c c c}

\hline
$B_i \;\mid$ &$\cdots$& $b_i$&$\cdots$&$b_{i+1}$\\
$B_{i+1}\; \mid$ &$\cdots$& $b_{i+1}$&$\cdots$&$b_{i+2}$\\

\hline
\end{tabular}
\end{center}

Since $B_i$ and $B_{i+1}$ have the same preference lists, $b_i$
comes ahead of $b_{i+1}$ on $B_{i+1}$'s preference list as shown
below.

\begin{center}
\begin{tabular}{r c c c c c c}

\hline
$B_i\; \mid$ &&&$\cdots$& $b_i$&$\cdots$&$b_{i+1}$\\
$B_{i+1}\; \mid$ &$\cdots$& $b_{i}$ &$\cdots$& $b_{i+1}$&$\cdots$&$b_{i+2}$\\

\hline
\end{tabular}
\end{center}

As
$M_1$ is a stable matching, the pair $(B_{i+1},b_i)$ must not form
a blocking pair to the stable pairs $(B_i,b_i)$ and
$(B_{i+1},b_{i+1})$ in $M_1$. Since $b_i$ comes ahead of $b_{i+1}$ on
$B_{i+1}$'s list, $B_{i+1}$ must appear {\em after} $B_i$ on $b_i$'s
preference list to ensure that $(B_{i+1},b_i)$ does not form such a
blocking pair. Therefore, the preference lists for man $B_{i+1}$ and woman
$b_{i}$ are as follows.

\begin{center}
\begin{tabular}{r c c c c c c c c c c c c c c}

\hline
$B_{i+1}\; \mid$&$\cdots$&$b_{i}$&$\cdots$& $b_{i+1}$&$\cdots$&$b_{i+2}$&\hspace{0.2in}$\mid$&$b_i\; \mid$&$\cdots$&$B_{i-1}$&$\cdots$&$B_i$&$\cdots$&$B_{i+1}$\\

\hline
\end{tabular}
\end{center}

Comparing the preference lists of women $b_{i}$ and $b_{i+1}$,

\begin{center}
\begin{tabular}{r c c c c c c}

\hline
$b_i \; \mid$ &$\cdots$&$B_{i-1}$&$\cdots$& $B_i$&$\cdots$&$B_{i+1}$\\
$b_{i+1}\; \mid$ &$\cdots$& $B_{i}$ &$\cdots$& $B_{i+1}$&$\cdots$&\\

\hline
\end{tabular}
\end{center}

\noindent we note that they are the same. Hence, woman $b_{i+1}$
also has $B_{i-1}$ ahead of $B_i$ on her list.

\begin{center}
\begin{tabular}{r c c c c c c c c}
\hline
$b_i\; \mid$ &&&$\cdots$&$B_{i-1}$&$\cdots$& $B_i$&$\cdots$&$B_{i+1}$\\
$b_{i+1}\; \mid$&$\cdots$&$B_{i-1}$ &$\cdots$& $B_{i}$ &$\cdots$& $B_{i+1}$&$\cdots$&\\
\hline
\end{tabular}
\end{center}

Because $M_2$ is also a stable matching and $B_{i-1}$ is ahead of
$B_{i}$ on $b_{i+1}$'s
preference list, $b_{i+1}$ must appear {\em after} $b_i$ on $B_{i-1}$'s
preference list to prevent $(B_{i-1},b_{i+1})$ from being a
blocking pair in $M_2$.
Comparing the preference lists of $B_{i-1}, B_i$, and
$B_{i+1}$, we note that

\begin{center}
\begin{tabular}{r c c c c c c c c}

\hline
$B_{i-1}\; \mid$ &&&$\cdots$& $b_{i-1}$ &$\cdots$& $b_{i}$&$\cdots$&$b_{i+1}$\\
$B_i\; \mid$ &&&$\cdots$& $b_i$&$\cdots$&$b_{i+1}$&&\\
$B_{i+1}\; \mid$ &$\cdots$& $b_{i}$ &$\cdots$& $b_{i+1}$&$\cdots$&$b_{i+2}$&&\\

\hline
\end{tabular}
\end{center}

\noindent they are all the same.

We have shown that man $B_{i-1}$ has the same preference list as men
$B_i$ and $B_{i+1}$, and that women $b_{i}$ and $b_{i+1}$ both have
the same preference list. We can now repeat the argument with men
$B_{i-1}$ and $B_i$ to conclude that men $B_{i-2}$ and $B_{i-1}$
have the same preference list and women $b_{i-1}$ and $b_{i}$ have
the same preference list and so on. In this fashion, we can show
that all men involved in the rotation have the same preference list
and so do all the women involved.

We know that the men involved in a rotation
get less happy with their partners as a result of applying the
rotation.
Since the men all have the same
preference lists, the relative order of women $b_1$ through $b_k$
should be the same. Suppose the order is $b_{i_1}, b_{i_2},
\cdots,b_{i_k}$. After the rotation, the man married to $b_{i_1}$
would go down his list and the man married to $b_{i_k}$ would go up
his list which cannot both happen at the same time. Hence, we cannot
have a rotation if some two consecutive men on the rotation have the
same preference lists.

Therefore, the preference lists of the men involved in the rotation
have to alternate, forcing the rotation to be of even size.
\end{proof}

\begin{lemma}\label{lem:rotations-have-size-2}
In the $1$-attribute model, every rotation is of size $2$.
\end{lemma}

\begin{proof} Suppose we have a rotation of size $2k$ involving
men $B_0,B_1,...,B_{2k-1}$
and women $b_0,b_1,...,b_{2k-1}$ where $k >
1$. Every man $B_i$ is married to woman $b_i$ before the rotation
and to woman $b_{i+1\pmod k}$ after the rotation as shown in the
table below.

\begin{center}
\begin{tabular}{c c c}
\hline \hline

Men & Before & After \\ [0.5ex]

\hline
$B_0$ &$b_0$&$b_1$\\
$B_1$ &$b_1$&$b_2$\\
\vdots \\
$B_i$ &$b_i$&$b_{i+1}$\\
$B_{i+1}$ &$b_{i+1}$&$b_{i+2}$\\
\vdots\\
$B_{2k-2}$&$b_{2k-2}$&$b_{2k-1}$\\
$B_{2k-1}$&$b_{2k-1}$&$b_{0}$\\
\hline
\end{tabular}
\end{center}

Since men $B_i$ with $i$ even
have the same preference lists, the relative
order of women on their lists  is the
same. Suppose the order is $b_{i_1}, b_{i_2}, \cdots,b_{i_{2k}}$.
Since,
for $0\leq i \leq k-1$, $b_{2i}$ is ahead of $b_{2i+1}$,
it
is clear that
$i_1$ is even.

Consider men $B_{i_1 -2}$, $B_{i_1 - 1}$, and $B_{i_1}$. Note that we
are implicitly using the assumption that $k > 1$ (otherwise there are not three distinct men).
Their preference
lists appear as follows.

\begin{center}
\begin{tabular}{r c c c c c c c c}

\hline
$B_{i_1-2}\; \mid$ &$\cdots$& $b_{i_1}$ &$\cdots$& $b_{i_1-2}$&$\cdots$&$b_{i_1-1}$&&\\
$B_{i_1-1}\; \mid$ &&&$\cdots$& $b_{i_1-1}$&$\cdots$&$b_{i_1}$&&\\
$B_{i_1}\; \mid$ &&&$\cdots$& $b_{i_1}$&$\cdots$&$b_{i_1+1}$&&\\

\hline
\end{tabular}
\end{center}

As all pairs $(B_j,b_j)$
are part of a stable matching, the pair $(B_{i_1-2}, b_{i_1})$ should
not form a blocking pair to the stable pairs $(B_{i_1-2}, b_{i_1-2})$
and $(B_{i_1}, b_{i_1})$. Since $b_{i_1}$ comes ahead of $b_{i_1-2}$
on $B_{i_1-2}$'s list, $B_{i_1-2}$ should appear after $B_{i_1}$ on
$b_{i_1}$'s preference list to ensure that $(B_{i_1-2}, b_{i_1})$
does not form a blocking pair. The preference lists for women
$b_{i_1}$ and woman $b_{i_1+1}$ are as follows.

\begin{center}
\begin{tabular}{r c c c c c c c c c c}
\hline
$b_{i_1}\; \mid$&&&$\cdots$&$B_{i_1-1}$ &$\cdots$& $B_{i_1}$ &$\cdots$& $B_{i_1-2}$&$\cdots$&\\
$b_{i_1+1}\; \mid$&$\cdots$&$B_{i_1-2}$ &$\cdots$& $B_{i_1}$ &$\cdots$& $B_{i_1+1}$&$\cdots$&&&\\
\hline
\end{tabular}
\end{center}

Note that $B_{i_1-2}$ appears ahead of $B_{i_1}$ on $b_{i_1+1}$'s
list as the lists of $b_{i_1}$ and $b_{i_1+1}$ are reverses of each
other.

Again,
all pairs $(B_j,b_{j+1})$
are part of a stable
matching and $(B_{i_1-2}, b_{i_1+1})$ could form a blocking pair to
the stable pairs $(B_{i_1-2}, b_{i_1-1})$ and $(B_{i_1}, b_{i_1+1})$.
Since $B_{i_1-2}$ is ahead of $B_{i_1}$ on $b_{i_1+1}$'s preference
list, $b_{i_1+1}$ has to appear after $b_{i_1-1}$ on $B_{i_1-2}$'s
preference list to prevent $(B_{i_1-2}, b_{i_1+1})$ from becoming a
blocking pair. The preference lists of $B_{i_1-2}$, $B_{i_1-1}$ and
$B_{i_1}$ are as follows.

\begin{center}
\begin{tabular}{r c c c c c c c c c c}
\hline
$B_{i_1-2}\; \mid$&$\cdots$&$b_{i_1}$ &$\cdots$& $b_{i_1-2}$ &$\cdots$& $b_{i_1-1}$&$\cdots$&$b_{i_1+1}$&\\
$B_{i_1-1}\; \mid$&$\cdots$&$b_{i_1+1}$ &$\cdots$& $b_{i_1-1}$ &$\cdots$& $b_{i_1}$&$\cdots$&&&\\
$B_{i_1}\; \mid$&&&$\cdots$& $b_{i_1}$ &$\cdots$& $b_{i_1+1}$&$\cdots$&&&\\
\hline
\end{tabular}
\end{center}

Note that $b_{i_1+1}$ appears ahead of $b_{i_1-1}$ on $B_{i_1-1}$'s
list as the lists of $B_{i_1-2}$ and $B_{i_1-1}$ are reverses of
each other.

The pair $(B_{i_1-1}, b_{i_1+1})$
should not form a blocking pair to the stable pairs
$(B_{i_1-1}, b_{i_1-1})$ and $(B_{i_1+1}, b_{i_1+1})$. Since
$b_{i_1+1}$ comes ahead of $b_{i_1-1}$ on $B_{i_1-1}$'s list,
$B_{i_1-1}$ should appear after $B_{i_1+1}$ on $b_{i_1+1}$'s
preference list to ensure that $(B_{i_1-1}, b_{i_1+1})$ does not form
a blocking pair. The preference lists for women $b_{i_1}$ and woman
$b_{i_1+1}$ are as follows.

\begin{center}
\begin{tabular}{r c c c c c c c c c c}
\hline
$b_{i_1}\; \mid$&&&$\cdots$&$B_{i_1-1}$ &$\cdots$& $B_{i_1}$ &$\cdots$& $B_{i_1-2}$&$\cdots$&\\
$b_{i_1+1}\; \mid$&$\cdots$&$B_{i_1-2}$ &$\cdots$& $B_{i_1}$ &$\cdots$& $B_{i_1+1}$&$\cdots$&$B_{i_1-1}$&&\\
$b_{i_1+2}\; \mid$&$\cdots$&$B_{i_1-1}$ &$\cdots$&$B_{i_1+1}$ &$\cdots$& $B_{i_1+2}$ &$\cdots$&&&\\
\hline
\end{tabular}
\end{center}

Note that $B_{i_1-1}$ appears ahead of $B_{i_1+1}$ on $b_{i_1+2}$'s
list as the lists of $b_{i_1+1}$ and $b_{i_1+2}$ are reverses of
each other.

Similarly, $(B_{i_1-1}, b_{i_1+2})$ must not be a blocking pair to
the stable pairs $(B_{i_1-1}, b_{i_1})$ and $(B_{i_1+1}, b_{i_1+2})$.
Since $B_{i_1-1}$ is ahead of $B_{i_1+1}$ on $b_{i_1+2}$'s
preference list, $b_{i_1+2}$ has to appear after $b_{i_1}$ on
$B_{i_1-1}$'s preference list to prevent $(B_{i_1-1}, b_{i_1+2})$
from becoming a blocking pair. The preference lists of $B_{i_1-2}$,
$B_{i_1-1}$ and $B_{i_1}$ are as follows.

\begin{center}
\begin{tabular}{r c c c c c c c c c c}
\hline
$B_{i_1-2}\; \mid$&$\cdots$&$b_{i_1}$ &$\cdots$& $b_{i_1-2}$ &$\cdots$& $b_{i_1-1}$&$\cdots$&$b_{i_1+1}$&\\
$B_{i_1-1}\; \mid$&$\cdots$&$b_{i_1+1}$ &$\cdots$& $b_{i_1-1}$ &$\cdots$& $b_{i_1}$&$\cdots$&$b_{i_1+2}$&$\cdots$&\\
$B_{i_1}\; \mid$&$\cdots$&$b_{i_1+2}$&$\cdots$& $b_{i_1}$ &$\cdots$& $b_{i_1+1}$&$\cdots$&&&\\
\hline
\end{tabular}
\end{center}

Note that $b_{i_1+2}$ appears ahead of $b_{i_1}$ on $B_{i_1}$'s list
as the lists of $B_{i_1-1}$ and $B_{i_1}$ are reverses of each
other. This contradicts the relative order of the women on
lists of men $B_j$ with $j$ even since $b_{i_1}$ should be first.

Hence, the size of any rotation is 2.
\end{proof}

\begin{lemma}\label{lem:at-most-1-rotation}
In the $1$-attribute model, every man (and woman) participates
in at most one rotation.
\end{lemma}
\begin{proof} Suppose man $B_1$ participates in more than one
rotation. Starting with his partner in the male-optimal matching,
man $B_1$ goes down his preference list with each rotation he
participates in. Suppose $b_1$ is the partner of $B_1$ in the
male-optimal matching, and $b_2$ and $b_3$ are partners of $B_1$ after
the first and second rotations, respectively, that involve $B_1$.
Let $B_2$ and $B_3$ be the
partners of $b_2$ and $b_3$ when they participate in the respective
rotations with $B_1$. The preference lists of $B_1$, $B_2$, $B_3$,
$b_1$, $b_2$, and $b_3$ are as follows.

\begin{center}
\begin{tabular}{r c c c c c c c c c c}
\hline
$B_{1}\; \mid$&$\cdots$&$b_{1}$ &$\cdots$& $b_{2}$ &$\cdots$& $b_{3}$&$\cdots$\\
$B_{2}\; \mid$&$\cdots$&$b_{2}$ &$\cdots$& $b_{1}$ &$\cdots$&&&\\
$B_{3}\; \mid$&&&$\cdots$&$b_{3}$&$\cdots$& $b_{2}$ &$\cdots$&\\
\hline
\end{tabular}\;
\begin{tabular}{r c c c c c c c c c c}
\hline
$b_{1}\; \mid$&&&$\cdots$&$B_{2}$ &$\cdots$& $B_{1}$ &$\cdots$\\
$b_{2}\; \mid$&$\cdots$&$B_{3}$ &$\cdots$& $B_{1}$ &$\cdots$&$B_{2}$&&\\
$b_{3}\; \mid$&$\cdots$&$B_{1}$&$\cdots$& $B_{3}$ &$\cdots$&\\
\hline
\end{tabular}
\end{center}

We note that $B_2$ and $B_3$ have the same preference lists and
$B_1$ has the reverse preference list. Hence, their preference lists
appear as follows.

\begin{center}
\begin{tabular}{r c c c c c c c c c c}
\hline
$B_{1}\; \mid$&&&$\cdots$&$b_{1}$ &$\cdots$& $b_{2}$ &$\cdots$& $b_{3}$&$\cdots$\\
$B_{2}\; \mid$&$\cdots$&$b_{3}$&$\cdots$&$b_{2}$ &$\cdots$& $b_{1}$ &$\cdots$&&&\\
$B_{3}\; \mid$&&&&&$\cdots$&$b_{3}$&$\cdots$& $b_{2}$ &$\cdots$&$b_1$\\
\hline
\end{tabular}
\end{center}

Similarly, $b_1$ and $b_3$ have the same preference lists and $b_2$
has the reverse preference list. Hence, their preference lists are
as follows.

\begin{center}
\begin{tabular}{r c c c c c c c c c c}
\hline
$b_{1}\; \mid$&&&&&$\cdots$&$B_{2}$ &$\cdots$& $B_{1}$ &$\cdots$&$B_{3}$\\
$b_{2}\; \mid$&&&$\cdots$&$B_{3}$ &$\cdots$& $B_{1}$ &$\cdots$&$B_{2}$&\\
$b_{3}\; \mid$&$\cdots$&$B_{2}$&$\cdots$&$B_{1}$&$\cdots$& $B_{3}$ &$\cdots$&\\
\hline
\end{tabular}
\end{center}

When $B_1$ and $B_2$ participate in the rotation, their partners are
$b_1$ and $b_2$  respectively. This implies that $(B_2, b_2)$ is a
stable pair and is part of a stable matching. Hence, the pair
$(B_2, b_3)$ cannot be a blocking pair. For $(B_2, b_3)$ to be not a
blocking pair, $b_3$ should be married to someone higher than $B_2$
on her list, say $B_x$. In other words, $b_3$ should be married to
$B_x$ before the rotation involving $B_1$ and $B_2$ occurs and
cannot to married to anyone lower than $B_x$ after the rotation has
occurred because $b_3$ can only go up her preference list after
future rotations.

\begin{center}
\begin{tabular}{r c c c c c c c c c c}
\hline
$b_{3}\; \mid$&$\cdots$&$B_{x}$&$\cdots$&$B_{2}$&$\cdots$&$B_{1}$&$\cdots$& $B_{3}$ &$\cdots$\\
\hline
\end{tabular}
\end{center}

This implies that $b_3$ can never be married to $B_1$ or $B_3$ in
the future, and the rotation involving the pairs $(B_1, b_2)$ and
$(B_3, b_3)$ which happens after the rotation involving $B_1$ and
$B_2$ violates that. Hence, any man (and woman) can participate in at
most one rotation.
\end{proof}

We are now in a position to prove Theorem~\ref{thm:1-attr-rot}, which we
repeat here.
\setcounter{counter:save}{\value{theorem}}
\setcounter{theorem}{\value{counter:1-attr-rot}}
\begin{theorem}
In the $1$-attribute model, the rotation poset of a stable matching instance
is (isomorphic to) a path.
\end{theorem}
\setcounter{theorem}{\value{counter:save}}

\begin{proof}
In order to prove this theorem, we need to show that any two rotations
are comparable, as this gives a total ordering on the set of rotations.

We start by computing the male-optimal and female-optimal stable
matchings.  The men who have the same partner in both matchings are
removed along with their partners from the problem instance as their
presence or absence does not affect the rotation poset.  So
we may assume that every man and women in the stable matching instance
is involved in at least one rotation.

Since every rotation involves exactly two men and two women
(Lemma~\ref{lem:rotations-have-size-2}), and, by
removing the men and women that are not involved in any rotations, we
see that each man and woman that remains is involved in exactly one
rotation (Lemma~\ref{lem:at-most-1-rotation}).  Thus, the number
of men and women in the (reduced) matching instance must be even.

Let the $2k$ men be denoted $\{ B_1, \ldots, B_{2k}\}$ and the $2k$
women be denoted $\{ b_1, \ldots, b_{2k}\}$. By relabeling, we
can assume that the male-optimal matching
pairs man $B_i$ with woman $b_i$, and the female-optimal matching
pairs man $B_{2i-1}$ with woman $b_{2i}$, and man $B_{2i}$ with woman
$b_{2i-1}$. In other words, there are $k$ rotations
$R_1,R_2,\cdots,R_k$ and rotation $R_i$ is of the form
$\{ (B_{2i-1},b_{2i-1}), (B_{2i}, b_{2i}) \}$.
We want to show that any two rotations
are comparable, i.e., for every $i,j\in \{1,2,\ldots,k\}$,
where $i \neq j$, either
$R_i$ precedes $R_j$ or
$R_j$ precedes  $R_i$.

Let us compare two rotations, say $R_1$ and $R_2$. The men and women
involved in the two rotations are $\{ B_1, B_2, B_3, B_4 \}$ and
$\{ b_1, b_2, b_3, b_4 \}$.

The preference list of $B_2$ is the reverse of $B_1$'s and that of
$B_3$'s is reverse of $B_4$'s. Without loss of generality, we could
assume that $B_1$ and $B_3$ have the same preference lists and that
$b_3$ comes ahead of $b_1$ on their preference lists. Therefore, the
partial preference lists of the men appear as follows.

\begin{center}
\begin{tabular}{r c c c c c c c c c c}
\hline
$B_{1}\; \mid$&$\cdots$&$b_{3}$&$\cdots$&$b_{1}$&$\cdots$& $b_{2}$&$\cdots$&&&\\
$B_{2}\; \mid$&&&$\cdots$&$b_{2}$ &$\cdots$& $b_{1}$ &$\cdots$&$b_{3}$&\\
$B_{3}\; \mid$&$\cdots$&$b_{3}$&$\cdots$&$b_{4}$ &$\cdots$&&&&&\\
$B_{4}\; \mid$&$\cdots$&$b_{4}$&$\cdots$& $b_{3}$ &$\cdots$&\\
\hline
\end{tabular}
\end{center}

The partial preference lists of the women are given below.

\begin{center}
\begin{tabular}{r c c c c c c c c c c}
\hline
$b_{1}\; \mid$ &$\cdots$&$B_{2}$&$\cdots$& $B_{1}$&$\cdots$&&&\\
$b_{2}\; \mid$&$\cdots$&$B_{1}$ &$\cdots$& $B_{2}$ &$\cdots$&\\
$b_{3}\; \mid$&$\cdots$&$B_{4}$ &$\cdots$& $B_{3}$ &$\cdots$&&\\
$b_{4}\; \mid$&$\cdots$&$B_{3}$&$\cdots$&$B_{4}$&$\cdots$&&&\\
\hline
\end{tabular}
\end{center}

The pair $(B_1, b_3)$ must not be a blocking pair to the male-optimal
matching that pairs $B_i$ to $b_i$. Since $b_3$ appears ahead of
$b_1$ on $B_1$'s preference list, $B_1$ should appear after $B_3$ on
$b_3$'s preference list. Therefore, the women's partial preference
lists are as follows.

\begin{center}
\begin{tabular}{r c c c c c c c c c c}
\hline
$b_{1}\; \mid$&$\cdots$&$B_{2}$&$\cdots$& $B_{1}$&$\cdots$&&&\\
$b_{2}\; \mid$&$\cdots$&$B_{1}$ &$\cdots$& $B_{2}$ &$\cdots$&\\
$b_{3}\; \mid$&&&$\cdots$&$B_{4}$&$\cdots$&$B_{3}$ &$\cdots$& $B_{1}$ &\\
$b_{4}\; \mid$&$\cdots$&$B_{1}$&$\cdots$&$B_{3}$ &$\cdots$& $B_{4}$ &$\cdots$&\\
\hline
\end{tabular}
\end{center}

Since the female-optimal matching pairs $(B_{2i-1}, b_{2i})$ and
$(B_{2i}, b_{2i-1})$, the pair $(B_1, b_4)$ cannot be a blocking pair.
Since $B_1$ appears ahead of $B_3$ on $b_4$'s preference list, $b_4$
should appear after $b_2$ on $B_1$'s preference list. So the men's
partial preference lists are as follows.

\begin{center}
\begin{tabular}{r c c c c c c c c c c}
\hline
$B_{1}\; \mid$&$\cdots$&$b_{3}$&$\cdots$&$b_{1}$&$\cdots$& $b_{2}$&$\cdots$&$b_{4}$&&\\
$B_{2}\; \mid$&$\cdots$&$b_{4}$&$\cdots$&$b_{2}$ &$\cdots$& $b_{1}$ &$\cdots$&$b_{3}$&\\
$B_{3}\; \mid$&$\cdots$&$b_{3}$&$\cdots$&$b_{1}$&$\cdots$& $b_{2}$&$\cdots$&$b_{4}$&&\\
$B_{4}\; \mid$&$\cdots$&$b_{4}$&$\cdots$&$b_{2}$ &$\cdots$& $b_{1}$ &$\cdots$&$b_{3}$&\\
\hline
\end{tabular}
\end{center}

Since the male-optimal matching pairs $(B_i, b_i)$, we see that
$(B_2, b_4)$ cannot be a blocking pair. Since $b_4$ appears ahead of
$b_2$ on $B_2$'s preference list, $B_2$ should appear after $B_4$ on
$b_4$'s preference list. This gives us more information about the women's
partial preference lists.

\begin{center}
\begin{tabular}{r c c c c c c c c c c}
\hline
$b_{1}\; \mid$&$\cdots$&$B_{2}$&$\cdots$&$B_{4}$&$\cdots$&$B_{3}$ &$\cdots$& $B_{1}$&\\
$b_{2}\; \mid$&$\cdots$&$B_{1}$&$\cdots$&$B_{3}$ &$\cdots$& $B_{4}$ &$\cdots$&$B_{2}$\\
$b_{3}\; \mid$&$\cdots$&$B_{2}$&$\cdots$&$B_{4}$&$\cdots$&$B_{3}$ &$\cdots$& $B_{1}$&\\
$b_{4}\; \mid$&$\cdots$&$B_{1}$&$\cdots$&$B_{3}$ &$\cdots$& $B_{4}$ &$\cdots$&$B_{2}$\\
\hline
\end{tabular}
\end{center}

Comparing the preference lists for the men and the women, we observe
that in the men's preference lists the women involved in one
rotation are sandwiched by the women of the other rotation.  A similar
thing happens in the women's preference lists, except that the
rotations reverse their roles here, i.e.\ if the women from rotation $R$
sandwich the women from rotation $R'$ in the men's preference lists,
then the men from $R'$ sandwich the men from $R$ in the women's
preference lists.  Since this is true for the men and women in
every pair of rotations, we could assume that all odd men have $b_1$ and
$b_2$ as the innermost pair, enveloped by $b_3$ and $b_4$ and so on. In
other words, the preference lists for the men are as follows.

\begin{center}
\begin{tabular}{r c c c c c c c c c c}
\hline
$B_{1}\; \mid$&$b_{2k-1} \;b_{2k-3} \;\cdots \;b_{3}\; b_{1}\; b_{2}\; b_{4} \;\cdots\; b_{2k-2}\; b_{2k}$&&\\
$B_{2}\; \mid$&$b_{2k} \;b_{2k-2} \;\cdots \;b_{4}\; b_{2}\; b_{1}\; b_{3} \;\cdots\; b_{2k-3}\; b_{2k-1}$&&\\
$B_{3}\; \mid$&$b_{2k-1} \;b_{2k-3} \;\cdots \;b_{3}\; b_{1}\; b_{2}\; b_{4} \;\cdots\; b_{2k-2}\; b_{2k}$&&\\
\vdots &\vdots\hspace{1in}\vdots\\
$B_{2k-2}\; \mid$&$b_{2k} \;b_{2k-2} \;\cdots \;b_{4}\; b_{2}\; b_{1}\; b_{3} \;\cdots\; b_{2k-3}\; b_{2k-1}$&&\\
$B_{2k-1}\; \mid$&$b_{2k-1} \;b_{2k-3} \;\cdots \;b_{3}\; b_{1}\; b_{2}\; b_{4} \;\cdots\; b_{2k-2}\; b_{2k}$&&\\
$B_{2k}\; \mid$&$b_{2k} \;b_{2k-2} \;\cdots \;b_{4}\; b_{2}\; b_{1}\; b_{3} \;\cdots\; b_{2k-3}\; b_{2k-1}$&&\\
\hline
\end{tabular}
\end{center}

This fixes the preference lists for the women and they are as
follows.

\begin{center}
\begin{tabular}{r c c c c c c c c c c}
\hline
$b_{1}\; \mid$&$B_{2} \;B_{4} \;\cdots \;B_{2k-2}\; B_{2k}\; B_{2k-1}\; B_{2k-3} \;\cdots\; B_{3}\; B_{1}$&&\\
$b_{2}\; \mid$&$B_{1} \;B_{3} \;\cdots \;B_{2k-3}\; B_{2k-1}\; B_{2k}\; B_{2k-2} \;\cdots\; B_{4}\; B_{2}$&&\\
$b_{3}\; \mid$&$B_{2} \;B_{4} \;\cdots \;B_{2k-2}\; B_{2k}\; B_{2k-1}\; B_{2k-3} \;\cdots\; B_{3}\; B_{1}$&&\\
\vdots &\vdots\hspace{1in}\vdots\\
$b_{2k-2}\; \mid$&$B_{1} \;B_{3} \;\cdots \;B_{2k-3}\; B_{2k-1}\; B_{2k}\; B_{2k-2} \;\cdots\; B_{4}\; B_{2}$&&\\
$b_{2k-1}\; \mid$&$B_{2} \;B_{4} \;\cdots \;B_{2k-2}\; B_{2k}\; B_{2k-1}\; B_{2k-3} \;\cdots\; B_{3}\; B_{1}$&&\\
$b_{2k}\; \mid$&$B_{1} \;B_{3} \;\cdots \;B_{2k-3}\; B_{2k-1}\; B_{2k}\; B_{2k-2} \;\cdots\; B_{4}\; B_{2}$&&\\
\hline
\end{tabular}
\end{center}

Suppose $1\leq i < k$.
Recall that $R_i$ is of the form $\{ (B_{2i-1},b_{2i-1}), (B_{2i}, b_{2i}) \}$
and $R_{i+1}$ is $\{ (B_{2i+1},b_{2i+1}), (B_{2i+2}, b_{2i+2}) \}$.
Now $R_i$ moves $b_{2i-1}$ from $B_{2i-1}$, which is below $B_{2i+1}$ on its preference list to $B_{2i}$, which is above
$B_{2i+1}$ on its preference list.
Hence the rotation $R_i$ eliminates the pair $(B_{2i+1},b_{2i-1})$.
Also, $R_{i+1}$ moves $B_{2i+1}$ to $b_{2i+2}$, which is strictly worse for $B_{2i+1}$ than $b_{2i-1}$.
Thus, $R_i$ explicitly precedes $R_{i+1}$ (taking $M=B_{2i+1}$ and $w=b_{2i-1}$ in Definition~\ref{def:explicitly}).

\end{proof}

\section{Stable matchings in the $k$-Euclidean model}\label{sect:k-Euclidean}

Having given our construction for the $k$-attribute setting, we now turn
to the $k$-Euclidean model.  We remind the reader that in this model
every man, say $A_i$, is associated
with two points in $\mathbb{R}^k$. One of the points, $\bar{A}_{i}$, denotes
his {\em position} and the other, $\hat{A}_i$, denotes the position
of his ideal partner. We refer to $\bar{A}_i$ as the {\em position point}
of $A_i$ and to $\hat{A}_i$ as
the {\em preference point} of $A_i$.  Similarly, each women has her own
position and preference points.  Each man ranks the women
based on the Euclidean distance between his own preference point and the
women's position points.
In other words, if the
distance between $\hat{A}_i$ and $\bar{b}$ is {\em less than}
the distance between $\hat{A}_i$ and $\bar{c}$, then $A_i$
prefers $b$ over $c$ ($b$ appears higher in his preference list than $c$).

In this section we work in the $2$-dimensional Euclidean model.
Our goal here is to establish Theorem~\ref{thm:k-Euclidean}, which
we repeat below.

\setcounter{counter:save}{\value{theorem}}
\setcounter{theorem}{\value{counter:k-Euclidean}}
\begin{theorem}
$\#BIS \equiv_{AP} \#SM(k\rm{-Euclidean})$ when $k \geq 2$.
\end{theorem}
\setcounter{theorem}{\value{counter:save}}

Theorem~\ref{thm:k-Euclidean} asserts that $\#BIS$ and
$\#SM(k\textrm{-Euclidean})$ are AP-interreducible for $k\geq 2$.
Since the AP-reduction from $\#SM(k\textrm{-Euclidean})$ to $\#BIS$
follows easily from known results (see Section~\ref{sect:idea}),
we now give an AP-reduction from $\#BIS$ to~$\#SM(k\textrm{-Euclidean})$.

As in Section~\ref{sect:k-attribute}, we will show how to take an
instance~$G$ of~$\#BIS$ and, in polynomial time, construct
an instance~$I$ of~$\#SM(k\textrm{-Euclidean})$
so that the number of stable matchings of~$I$ is equal to the number
of independent sets of~$G$.

Let~$G=(V_1\cup V_2,E)$ be an instance of~$\#BIS$
with $|E|=n$. We will construct
a $2$-Euclidean stable matching instance having $3n$ men and $3n$ women.
Our construction will use the
$\rho$-cycles and $\sigma$-cycles   defined
in Section~\ref{sect:permutations} .
To specify the stable matching instance, we now give position and
preference
points for the $3n$ men and women.

First, we position
the $3n$ women $a_1,\cdots, a_n,$ $b_1,\cdots,b_n,c_1,\cdots,c_n$
such that the $b$-women lie on the $y$-axis and the $a$-women and
$c$-women lie on the $x$-axis. We represent woman $w_i$ by
$\bar{w}_i = (\bar{w}_{i}(x), \bar{w}_{i}(y))$ where
$\bar{w}_{i}(x)$ and $\bar{w}_{i}(y)$
are her $x$- and $y$-coordinates. The coordinates of
$\bar{a}_i$, $\bar{b}_i$ and
$\bar{c}_i$ are $(\bar{a}_{i}(x),0), (0,\bar{b}_{i}(y))$ and
$(\bar{c}_{i}(x),0)$, respectively. We impose
further restrictions on the coordinates of $\bar{a}_i$,
$\bar{b}_i$, and $\bar{c}_i$.

\begin{align*}
\textrm{Let}\ \bar{b}_{\rho i}(y) = \bar{a}_i(x),\;\;
\bar{c}_{\sigma^{-1} i}(x) = \bar{a}_i(x) - 0.7 \;\; \ \textrm{for}\
1\leq i \leq n.
\end{align*}

Fixing the $x$-coordinates of $\bar{a}_1, \ldots, \bar{a}_n$
therefore fixes the positions of all of the women. Suppose $D_1$
through $D_l$ are the $l$ cycles of $\sigma$ of lengths $p_1$
through $p_l$, respectively. As before, let $e_i$ be a
representative element of cycle $D_i$. So
$D_i = \{e_i, \sigma(e_{i}), \ldots, \sigma^{p_i-1}(e_{i})\}$.
Also as before, let $Rep(\sigma) = \{e_1,e_2\cdots,e_l\}$
be the set of representative elements of the $\sigma$-cycles.

Let $W_i = \{a_x:x\in D_i\}\cup\{c_x:x\in D_i\}$.
We set $p_0 = 0$. For woman $a_{\sigma^{h} e_j}$,
where $e_j \in Rep(\sigma), \ \textrm{and}\ 0\leq h \leq p_j-1$, we
set
$\bar{a}_{\sigma^h e_{j}}(x) = \sum_{i=0}^{j-1} 2p_i + h+1$.
The position points of the women are as follows.
\begin{align*}
\textrm{For} \ e_j &\in Rep(\sigma),\;\;\;0 \leq h \leq p_j-1 \ \textrm{let}\\
\bar{a}_{\sigma^{h} e_j} &= \left(\;\sum_{i=0}^{j-1}2p_i + h+1\;,\;0\right), \\
\bar{b}_{\rho\sigma^{h} e_j} &= \left(0\;,\;\sum_{i=0}^{j-1}2p_i + h+1\right), \ \textrm{and} \\
\bar{c}_{\sigma^{(h-1)} e_j} &= \left(\;\sum_{i=0}^{j-1}2p_i +
h+0.3\;,\;0\right).
\end{align*}

Next we fix the locations in the $x$-$y$ plane for the ideal partners
of the men as follows, i.e.\ we specify the preference points for
each man.
\begin{align*}
\textrm{Let} \ \epsilon &= 1/100^{n}. \\
\textrm{For} \ e_j & \in Rep(\sigma),\;\;\;0 \leq h \leq p_j-1 \ \textrm{let} \\
\hat{A}_{\sigma^{h} e_j} &= \left(\;\sum_{i=0}^{j-1}2p_i + h+1\;,\;\sum_{i=0}^{j-1}2p_i + h+1 - \epsilon\right), \\
\hat{B}_{\sigma^{h} e_j} &= \left(\;\sum_{i=0}^{j-1}2p_i + h+1\;,\;1000^{n}\right), \ \textrm{and} \\
\hat{C}_{\sigma^{(h-1)} e_j} &= \left(\;\sum_{i=0}^{j-1}2p_i + h+0.6\;,\;0\right).
\end{align*}

Having fixed the position of the women and the preference points for
the men, we next fix the position of the men and the preference
points of the women.

First, we position the $3n$ men $A_1,\cdots, A_n,$
$B_1,\cdots,B_n,C_1,\cdots,C_n$ such that the $C$-men lie on the
$y$-axis and the $A$-men and $B$-men lie on the $x$-axis. We
represent man $m_i$ by $\bar{M}_i = (\bar{M}_{i}(x),
\bar{M}_{i}(y))$ where $\bar{M}_{i}(x)$ and $\bar{M}_{i}(y)$ are his
$x$- and $y$-coordinates. The coordinates of $\bar{A}_i$,
$\bar{B}_i$ and $\bar{C}_i$ are $(\bar{A}_{i}(x),0),
(\bar{B}_{i}(x),0)$ and $(0,\bar{C}_{i}(y))$, respectively. We
impose further restrictions on the coordinates of $\bar{A}_i$,
$\bar{B}_i$, and $\bar{C}_i$.
\begin{align*}
\textrm{Let}\ \bar{B}_i(x) = \bar{C}_{i}(y),\;\; \bar{A}_{\rho^{-1}
i}(x) = \bar{B}_i(x) - 0.7 \;\; \ \textrm{for}\ 1\leq i \leq n.
\end{align*}

Here again, fixing the $x$-coordinates of $\bar{B}_1, \ldots,
\bar{B}_n$ therefore fixes the positions of all of the men. Suppose
$E_1$ through $E_k$ are the $k$ cycles of $\rho$ of lengths $q_1$
through $q_k$, respectively. As before, let $f_i$ be a
representative element of cycle $E_i$, which is the element of~$E_i$
with the smallest index.
So
$E_i = \{f_i, \rho(f_{i}), \ldots, \rho^{q_i-1}(f_{i})\} = \{f_i,
f_{i} + 1, \ldots, f_{i}+q_i-1\} $. Also as before, let $Rep(\rho) =
\{f_1,f_2\cdots,f_k\}$ be the set of representative elements of the
$\rho$-cycles.

Let $W_i = \{B_x:x\in E_i\}\cup\{A_x:x\in E_i\}$. We set $q_0 = 0$.
For man $B_{\rho^{h} f_j}$, where $f_j \in Rep(\rho), \
\textrm{and}\ 0\leq h \leq q_j-1$, we set
$\bar{B}_{\rho^h f_{j}}(x) = \sum_{i=0}^{j-1} 2q_i + h+1$. The
position points of the men are as follows.
\begin{align*}
\textrm{For} \ f_j &\in Rep(\rho),\;\;\;0 \leq h \leq q_j-1 \ \textrm{let}\\
\bar{A}_{\rho^{h-1} f_j} &= \left(\;\sum_{i=0}^{j-1}2q_i + h+0.3\;,\;0\right), \\
\bar{B}_{\rho^{h} f_j} &= \left(\;\sum_{i=0}^{j-1}2q_i + h+1\;,\;0\right), \ \textrm{and} \\
\bar{C}_{\rho^{h} f_j} &= \left(0\;,\;\sum_{i=0}^{j-1}2q_i +
h+1\right).
\end{align*}

Next we fix the locations in the $x$-$y$ plane for the ideal
partners of the women as follows, i.e.\ we specify the preference
points for each woman.
\begin{align*}
\textrm{Let} \ \epsilon &= 1/100^{n}. \\
\textrm{For} \ f_j & \in Rep(\rho),\;\;\;0 \leq h \leq q_j-1 \ \\
\textrm{let}\;\;
\hat{a}_{\rho^{h} f_j} &= \left(\;\sum_{i=0}^{j-1}2q_i + h+1\;,\;\;1000^{n}\right), \\
\hat{b}_{\rho^{h} f_j} &= \left(\;\sum_{i=0}^{j-1}2q_i + h+0.6\;,\;0\right), \ \textrm{and} \\
\hat{c}_{\rho^{h} f_j} &= \left(\;\sum_{i=0}^{j-1}2q_i +
h+1\;,\;\sum_{i=0}^{j-1}2q_i + h+1 - \epsilon\right).
\end{align*}

Having assigned position and preference points for both the men and
the women, we construct the initial part of the preference lists of
the men starting with man $C_{\sigma^{(h-1)} e_j}$. We compare the
distances of the women from $\hat{C}_{\sigma^{(h-1)} e_j}$ to
produce the initial part of the preference list.
\begin{eqnarray*}
e_j,e_m & \in & Rep(\sigma)\;\;,\;\;\;0 \leq f,h \leq p_j-1\;\;,\;\;\;0 \leq g \leq p_m-1\\
d^2(\hat{C}_{\sigma^{(h-1)} e_j},\bar{b}_{\rho\sigma^{g} e_m}) &= &
(\sum_{i=0}^{j-1}2p_i + h+0.6 - 0)^2  \\
 & & + \ (0-\sum_{i=0}^{m-1}2p_i - g
-1)^2 \geq 0.6^2 + 1^2 = 1.36
\end{eqnarray*}

\begin{eqnarray*}
d^2(\hat{C}_{\sigma^{(h-1)} e_j},\bar{c}_{\sigma^{(h-1)} e_j}) & = &
(\sum_{i=0}^{j-1}2p_i + h+0.6 \\
  & & - \ \sum_{i=0}^{j-1}2p_i - h-0.3)^2 + (0-0)^2 = 0.09 \ \ \textrm{ and} \\
d^2(\hat{C}_{\sigma^{(h-1)} e_j},\bar{a}_{\sigma^{h} e_j}) & = &
(\sum_{i=0}^{j-1}2p_i + h+0.6 \\
  & & - \ \sum_{i=0}^{j-1}2p_i - h-1)^2 + (0-0)^2 = 0.16.\\
\end{eqnarray*}
For $h \neq f$,
\begin{eqnarray*}
d^2(\hat{C}_{\sigma^{(h-1)} e_j},\bar{c}_{\sigma^{(f-1)} e_j}) & = &
(\sum_{i=0}^{j-1}2p_i + h+0.6 - \sum_{i=0}^{j-1}2p_i - f-0.3)^2 +
(0-0)^2\\
& \geq & (|h-f| - 0.3)^2 \geq (1-0.3)^2 = 0.49 \ \ \textrm{ and} \\
d^2(\hat{C}_{\sigma^{(h-1)} e_j},\bar{a}_{\sigma^{(f)} e_j}) & = &
(\sum_{i=0}^{j-1}2p_i + h+0.6 - \sum_{i=0}^{j-1}2p_i - f-1)^2 +
(0-0)^2\\
& \geq & (|h-f| - 0.4)^2 \geq (1-0.4)^2 = 0.36.\\
\end{eqnarray*}

\noindent For $m > j$,
\begin{eqnarray*}
d^2(\hat{C}_{\sigma^{(h-1)} e_j},\bar{c}_{\sigma^{(g-1)} e_m}) & = &
(\sum_{i=0}^{j-1}2p_i + h+0.6 - \sum_{i=0}^{m-1}2p_i - g-0.3)^2 +
(0-0)^2\\
& \geq & (|\sum_{i=j}^{m-1}2p_i + g|- |h+0.3|)^2 \\
& \geq & (2p_j - (p_j -1) - 0.3)^2 \geq (2-0.3)^2 = 2.89\ \ \textrm{ and} \\
d^2(\hat{C}_{\sigma^{(h-1)} e_j},\bar{a}_{\sigma^{g} e_m}) & = &
(\sum_{i=0}^{j-1}2p_i + h+0.6 - \sum_{i=0}^{m-1}2p_i - g-1)^2 +
(0-0)^2\\
& \geq & (|\sum_{i=j}^{m-1}2p_i + g + 0.4|- |h|)^2 \\
& \geq & (2p_j +0.4 - (p_j-1))^2 \geq 2.4^2 = 5.76.
\end{eqnarray*}
For $j > m$,
\begin{eqnarray*}
d^2(\hat{C}_{\sigma^{(h-1)} e_j},\bar{c}_{\sigma^{(g-1)} e_m}) & = &
(\sum_{i=0}^{j-1}2p_i + h+0.6 - \sum_{i=0}^{m-1}2p_i - g-0.3)^2 +
(0-0)^2\\
& \geq & (|\sum_{i=m}^{j-1}2p_i + h +0.3|- |g|)^2 \\
& \geq & (2p_m +0.3 - (p_m-1))^2 \geq 2.3^2 = 5.29 \ \ \textrm{ and}\\
d^2(\hat{C}_{\sigma^{(h-1)} e_j},\bar{a}_{\sigma^{g} e_m}) & = &
(\sum_{i=0}^{j-1}2p_i + h+0.6 - \sum_{i=0}^{m-1}2p_i - g-1)^2 +
(0-0)^2\\
& \geq & (|\sum_{i=m}^{j-1}2p_i + h|- |g+0.4|)^2 \\
& \geq & (2p_m - (p_m-1)-0.4)^2 \geq (1.6)^2 = 2.56.\\
\end{eqnarray*}

\From the above analysis, it follows that the preference list of
$C_{\sigma^{(h-1)} e_j}$ starts with $c_{\sigma^{(h-1)} e_j}
a_{\sigma^{h} e_j}$ for $1 \leq j \leq l,\;0\leq h \leq p_j-1$
as in Section~\ref{sect:pref-lists}.

Now we carry out a similar analysis to determine the initial part of
the preference list of $A_{\sigma^{h} e_j}$. We note that
$\sum_{i=1}^{l}2p_i = 2n$ and $2\epsilon \cdot 2n = \frac{4n}{100^n}
\leq 0.04$. This implies that in the following analysis we could
upper bound the term $2\epsilon \cdot (\sum_{i=1}^{j}(2p_i) + h +
1)$ by $0.04$.
\begin{eqnarray*}
\epsilon = 1/100^n\; ,& & e_j,e_m \in Rep(\sigma)\;\;,\;\;\;0 \leq f,h \leq p_j-1\;\;,\;\;\;0 \leq g \leq p_m-1\\
d^2(\hat{A}_{\sigma^{h}
e_j},\bar{a}_{\sigma^{h} e_j}) & = & (\sum_{i=0}^{j-1}2p_i + h+1-
\sum_{i=0}^{j-1}2p_i - h-1)^2 \\
& & + \ (\sum_{i=0}^{j-1}2p_i + h+1-\epsilon-0)^2\\
& = & (\sum_{i=0}^{j-1}2p_i + h+ 1-\epsilon)^2
\end{eqnarray*}
\begin{eqnarray*}
d^2(\hat{A}_{\sigma^{h} e_j},\bar{b}_{\rho\sigma^{h} e_j}) & = &
(\sum_{i=0}^{j-1}2p_i + h+1- 0)^2 \\
& & + \ (\sum_{i=0}^{j-1}2p_i + h+1-\epsilon-\sum_{i=0}^{j-1}2p_i - h-1)^2\\
& = & (\sum_{i=0}^{j-1}2p_i + h+ 1)^2 + \epsilon^2
\end{eqnarray*}
\begin{eqnarray*}
d^2(\hat{A}_{\sigma^{h} e_j},\bar{c}_{\sigma^{h-1} e_j}) & = &
(\sum_{i=0}^{j-1}2p_i + h+1- \sum_{i=0}^{j-1}2p_i - h-0.3)^2 \\
& & + \ (\sum_{i=0}^{j-1}2p_i + h+1-\epsilon-0)^2\\
& = & 0.7^2 + (\sum_{i=0}^{j-1}2p_i + h+ 1-\epsilon)^2 \\
& = & (\sum_{i=0}^{j-1}2p_i + h+ 1-\epsilon)^2 + 0.49\\
& \geq & (\sum_{i=0}^{j-1}2p_i + h+ 1)^2 + \epsilon^2 + 0.45
\end{eqnarray*}
For $h \neq f$,
\begin{eqnarray*}
d^2(\hat{A}_{\sigma^{h} e_j},\bar{a}_{\sigma^{f} e_j}) & = &
(\sum_{i=0}^{j-1}2p_i + h+1- \sum_{i=0}^{j-1}2p_i - f-1)^2  \\
   & & + \ (\sum_{i=0}^{j-1}2p_i + h+1-\epsilon-0)^2\\
& = & (h-f)^2 + (\sum_{i=0}^{j-1}2p_i + h+ 1-\epsilon)^2  \\
& \geq & (\sum_{i=0}^{j-1}2p_i + h+ 1-\epsilon)^2 + 1\\
& \geq & (\sum_{i=0}^{j-1}2p_i + h+ 1)^2 + \epsilon^2 + 0.96
\end{eqnarray*}
\begin{eqnarray*}
d^2(\hat{A}_{\sigma^{h} e_j},\bar{b}_{\rho\sigma^{f} e_j}) & = &
(\sum_{i=0}^{j-1}2p_i + h+1- 0)^2 \\
  & & + \ (\sum_{i=0}^{j-1}2p_i + h+1-\epsilon-\sum_{i=0}^{j-1}2p_i - f-1)^2\\
& = & (\sum_{i=0}^{j-1}2p_i + h+ 1)^2 + (h-f-\epsilon)^2 \\
& \geq & (\sum_{i=0}^{j-1}2p_i + h+ 1)^2 + (1-\epsilon)^2\\
& \geq & (\sum_{i=0}^{j-1}2p_i + h+ 1)^2 + \epsilon^2 + 0.98
\end{eqnarray*}
\begin{eqnarray*}
d^2(\hat{A}_{\sigma^{h} e_j},\bar{c}_{\sigma^{f-1} e_j}) & = &
(\sum_{i=0}^{j-1}2p_i + h+1- \sum_{i=0}^{j-1}2p_i - f-0.3)^2 \\
   &  & + \ (\sum_{i=0}^{j-1}2p_i + h+1-\epsilon-0)^2\\
& = & (h-f-0.7)^2 + (\sum_{i=0}^{j-1}2p_i + h+ 1-\epsilon)^2 \\
& \geq & (\sum_{i=0}^{j-1}2p_i + h+ 1-\epsilon)^2 + 0.09\\
& \geq & (\sum_{i=0}^{j-1}2p_i + h+ 1)^2 + \epsilon^2 + 0.05\\
\end{eqnarray*}

\noindent For  $m > j$,
\begin{eqnarray*}
d^2(\hat{A}_{\sigma^{h} e_j},\bar{a}_{\sigma^{g} e_m}) & = &
(\sum_{i=0}^{j-1}2p_i + h+1- \sum_{i=0}^{m-1}2p_i - g-1)^2 \\
  &  & + \ (\sum_{i=0}^{j-1}2p_i + h+1-\epsilon-0)^2\\
& \geq & (|\sum_{i=j}^{m-1}2p_i + g| - |h|)^2 + (\sum_{i=0}^{j-1}2p_i +
h+ 1-\epsilon)^2\\
& \geq & (2p_j - (p_j-1))^2 + (\sum_{i=0}^{j-1}2p_i + h+ 1-\epsilon)^2 \\
& \geq & (\sum_{i=0}^{j-1}2p_i + h+ 1-\epsilon)^2 + 4\\
& \geq & (\sum_{i=0}^{j-1}2p_i + h+ 1)^2 + \epsilon^2 + 3.96
\end{eqnarray*}
\begin{eqnarray*}
d^2(\hat{A}_{\sigma^{h} e_j},\bar{b}_{\rho\sigma^{g} e_m}) & = &
(\sum_{i=0}^{j-1}2p_i + h+1- 0)^2 \\
 & & + \ (\sum_{i=0}^{j-1}2p_i + h+1-\epsilon-\sum_{i=0}^{m-1}2p_i - g-1)^2\\
& \geq & (\sum_{i=0}^{j-1}2p_i + h+ 1)^2 + (|\sum_{i=j}^{m-1}2p_i + g
+\epsilon| - |h|)^2\\
& \geq & (\sum_{i=0}^{j-1}2p_i + h+ 1)^2 + (2p_j +\epsilon - (p_j-1))^2 \\
& \geq & (\sum_{i=0}^{j-1}2p_i + h+ 1)^2 + (2 +\epsilon)^2 \\
& \geq & (\sum_{i=0}^{j-1}2p_i + h+ 1)^2 + \epsilon^2 +4
\end{eqnarray*}
\begin{eqnarray*}
d^2(\hat{A}_{\sigma^{h} e_j},\bar{c}_{\sigma^{g-1} e_m}) & = &
(\sum_{i=0}^{j-1}2p_i + h+1- \sum_{i=0}^{m-1}2p_i - g-0.3)^2 \\
 & & + \ (\sum_{i=0}^{j-1}2p_i + h+1-\epsilon-0)^2\\
& \geq & (|\sum_{i=j}^{m-1}2p_i + g|-|h+0.7|)^2 + (\sum_{i=0}^{j-1}2p_i
+ h+ 1-\epsilon)^2 \\
& \geq &(2p_j-(p_j-1) -0.7)^2 + (\sum_{i=0}^{j-1}2p_i + h+
1-\epsilon)^2\\
& \geq & (2-0.7)^2 + (\sum_{i=0}^{j-1}2p_i + h+ 1-\epsilon)^2 \\
& = & (\sum_{i=0}^{j-1}2p_i + h+ 1-\epsilon)^2 + 1.69\\
& \geq & (\sum_{i=0}^{j-1}2p_i + h+ 1)^2 + \epsilon^2 + 1.65
\end{eqnarray*}

\noindent For $j > m$,
\begin{eqnarray*}
d^2(\hat{A}_{\sigma^{h} e_j},\bar{a}_{\sigma^{g} e_m}) & = &
(\sum_{i=0}^{j-1}2p_i + h+1- \sum_{i=0}^{m-1}2p_i - g-1)^2 \\
 & & + \ (\sum_{i=0}^{j-1}2p_i + h+1-\epsilon-0)^2\\
& \geq & (|\sum_{i=m}^{j-1}2p_i + h| - |g|)^2 + (\sum_{i=0}^{j-1}2p_i +
h+ 1-\epsilon)^2\\
& \geq & (2p_m - (p_m-1))^2 + (\sum_{i=0}^{j-1}2p_i + h+ 1-\epsilon)^2 \\
& \geq & (\sum_{i=0}^{j-1}2p_i + h+ 1-\epsilon)^2 + 4\\
&\geq & (\sum_{i=0}^{j-1}2p_i + h+ 1)^2 + \epsilon^2 + 3.96
\end{eqnarray*}
\begin{eqnarray*}
d^2(\hat{A}_{\sigma^{h} e_j},\bar{b}_{\rho\sigma^{g} e_m}) & = &
(\sum_{i=0}^{j-1}2p_i + h+1- 0)^2 \\
 & & + \ (\sum_{i=0}^{j-1}2p_i + h+1-\epsilon-\sum_{i=0}^{m-1}2p_i - g-1)^2\\
& \geq & (\sum_{i=0}^{j-1}2p_i + h+ 1)^2 + (|\sum_{i=m}^{j-1}2p_i + h| - |g +\epsilon|)^2\\
& \geq & (\sum_{i=0}^{j-1}2p_i + h+ 1)^2 + (2p_m -(p_m -1) -\epsilon)^2 \\
& \geq & (\sum_{i=0}^{j-1}2p_i + h+ 1)^2 + (2 -\epsilon)^2\\
& \geq & (\sum_{i=0}^{j-1}2p_i + h+ 1)^2 + \epsilon^2 + 3.96
\end{eqnarray*}
\begin{eqnarray*}
d^2(\hat{A}_{\sigma^{h} e_j},\bar{c}_{\sigma^{g-1} e_m}) & = &
(\sum_{i=0}^{j-1}2p_i + h+1- \sum_{i=0}^{m-1}2p_i - g-0.3)^2 \\
 & & + \ (\sum_{i=0}^{j-1}2p_i + h+1-\epsilon-0)^2\\
& \geq &(|\sum_{i=m}^{j-1}2p_i + h +0.7|-|g|)^2 +
(\sum_{i=0}^{j-1}2p_i
+ h+ 1-\epsilon)^2 \\
& \geq & (2p_m-(p_m-1) + 0.7)^2 + (\sum_{i=0}^{j-1}2p_i + h+
1-\epsilon)^2\\
& \geq & (2+ 0.7)^2 + (\sum_{i=0}^{j-1}2p_i + h+ 1-\epsilon)^2 \\
& = & (\sum_{i=0}^{j-1}2p_i + h+ 1-\epsilon)^2 + 7.29\\
& \geq & (\sum_{i=0}^{j-1}2p_i + h+ 1)^2 + \epsilon^2 + 7.25
\end{eqnarray*}
From the above analysis, it follows that the preference list of
$A_{\sigma^{h} e_j}$ starts with $a_{\sigma^{h} e_j} b_{\rho
\sigma^{h} e_j}$ for $1\leq j \leq l,\; 0 \leq h \leq p_j-1$
as in Section~\ref{sect:pref-lists}.

Last, we study the preference list of $B_{\sigma^{h} e_j}$.
First we will show that the preference list of man
$B_{\sigma^{h} e_j}$, where $1\leq j \leq l,\;0\leq h\leq
p_j-1$,  starts with
$$b_{\rho\sigma^{(p_l-1)}e_l}b_{\rho\sigma^{(p_l-2)}e_l}\cdots
b_{\rho e_l} b_{\rho\sigma^{(p_{l-1}-1)}e_{l-1}}\cdots b_{\rho
e_{l-1}}\cdots b_{\rho\sigma^{(p_1-1)}e_1}\cdots b_{\rho e_1}.$$
We
obtain the above preference list by comparing distances between
$\hat{B}_{\sigma^{h} e_j}$ and the positions of the women.
\begin{eqnarray*}
e_j,e_k \in Rep(\sigma), & &0 \leq h \leq p_j-1\;\;,\;\;\;\;0 \leq f \leq p_k-1\\
d^2(\hat{B}_{\sigma^{h} e_j},\bar{b}_{\rho\sigma^{f} e_k}) & = &
(\sum_{i=0}^{j-1}2p_i + h+1- 0)^2 +
(1000^{n}-\sum_{i=0}^{k-1}2p_i - f-1)^2\\
& \leq & (\sum_{i=0}^{j-1}2p_i + h+1)^2 +
(1000^{n} - 1)^2 < (1000^{n})^2
\end{eqnarray*}
\begin{eqnarray*}
d^2(\hat{B}_{\sigma^{h} e_j},\bar{a}_{\sigma^{f} e_k}) & = &
(\sum_{i=0}^{j-1}2p_i + h+1- \sum_{i=0}^{k-1}2p_i - f-1)^2 +
(1000^{n}-0)^2\\
& \geq & (1000^{n})^2
\end{eqnarray*}
\begin{eqnarray*}
d^2(\hat{B}_{\sigma^{h} e_j},\bar{c}_{\sigma^{f-1} e_k}) & = &
(\sum_{i=0}^{j-1}2p_i + h+1- \sum_{i=0}^{k-1}2p_i - f-0.3)^2 +
(1000^{n}-0)^2\\
& \geq & (1000^{n})^2\\
\end{eqnarray*}

It immediately follows that the $b$-women are all closer to
$\hat{B}_{\sigma^{h} e_j}$ than any of the $a$-women or $c$-women.
Hence, the preference list of $B_{\sigma^{h} e_j}$ would start with
all the $b$-women coming first. We also note that the $b$-women all
have their $x$-component set to 0. Hence, $B_{\sigma^{h} e_j}$ would
rank the $b$-women by measuring their distance from
$\hat{B}_{\sigma^{h} e_j}$ in the $y$-component. We also note that
$1000^n - 2n > 0$ for $n \geq 1$. Next we compare distances between
$\hat{B}_{\sigma^{h} e_j}$ and the $b$-women only using the
$y$-component. We will use the notation $d_y(\cdot,\cdot)$ to denote
the distance in the $y$-component.
\begin{align*}
e_j,e_{k_1},e_{k_2} \in Rep(\sigma)\;\;,\;\;\;0 \leq h\leq p_j-1\;\;,\;\;\;
&0 \leq g \leq p_{k_1}-1\;\;,\;\;\;0 \leq f \leq p_{k_2}-1
\end{align*}

\noindent For $k_1 = k_2$ and $g > f$, we have
\begin{eqnarray*}
d_y(\hat{B}_{\sigma^{h} e_j},\bar{b}_{\rho\sigma^{g} e_{k_1}}) & = &
1000^{n}-\sum_{i=0}^{k_{1}-1}2p_i - g-1 <
1000^{n}-\sum_{i=0}^{k_{1}-1}2p_i - f -1\\
& = & 1000^{n}-\sum_{i=0}^{k_{2}-1}2p_i - f -1 =
d_y(\hat{B}_{\sigma^{h}
e_j},\bar{b}_{\rho\sigma^{f} e_{k_2}}).
\end{eqnarray*}

\noindent For $k_1 > k_2$, we have
\begin{eqnarray*}
d_y(\hat{B}_{\sigma^{h} e_j},\bar{b}_{\rho\sigma^{f} e_{k_2}})
   - d_y(\hat{B}_{\sigma^{h} e_j},\bar{b}_{\rho\sigma^{g} e_{k_1}}) & = &
       (1000^{n}-\sum_{i=0}^{k_{2}-1}2p_i - f-1) \\
   & & \ - \ (1000^{n}-\sum_{i=0}^{k_{1}-1}2p_i - g -1)\\
& = & \sum_{i=0}^{k_{1}-1}2p_i + g -\sum_{i=0}^{k_{2}-1}2p_i - f = \sum_{i=k_2}^{k_{1}-1}2p_i + g - f\\
& \geq & 2p_{k_{2}} + g - (p_{k_{2}} - 1) = p_{k_2} + 1 + g > 0.
\end{eqnarray*}
From the above discussion, it follows that the preference list of
$B_{\sigma^{h} e_j}$ starts with
$$b_{\rho\sigma^{(p_l-1)}e_l}b_{\rho\sigma^{(p_l-2)}e_l}\cdots
b_{\rho e_l} b_{\rho\sigma^{(p_{l-1}-1)}e_{l-1}}\cdots b_{\rho
e_{l-1}}\cdots b_{\rho\sigma^{(p_1-1)}e_1}\cdots b_{\rho e_1}.$$
Next
we compare the distances of $a$-women and $c$-women from
$\hat{B}_{\sigma^{h} e_j}$. As $a$-women and $c$-women all have
their $y$-component set to 0, $B_{\sigma^{h} e_j}$ would rank the
$b$-women by measuring their distance from $\hat{B}_{\sigma^{h}
e_j}$ in the $x$-component. We will use the notation
$d_x(\cdot,\cdot)$ to denote the distance in the $x$-component. We
consider two cases (i) $h \neq p_i-1,\; 1\leq i \leq l$, (ii) $h =
p_i-1,\;$ for some $i \in \{1,2,\cdots,l\}$.

Case(i) $h \neq p_i-1,\; 1\leq i \leq l$: Now we compute and compare
the distances between $\hat{B}_{\sigma^{h} e_j}$ and the $a$-women
and the $c$-women.
$$
e_j,e_{k} \in Rep(\sigma)\;\;,\;\;\; 0 \leq h\leq p_j-2\;\;,\;\;\;
0 \leq g \leq p_{k}-1\;\;,\\
$$
For $k = j$ and $g = h$, we have
\begin{eqnarray*}
d_x(\hat{B}_{\sigma^{h} e_j},\bar{a}_{\sigma^{g} e_{k}}) & = &
d_x(\hat{B}_{\sigma^{h} e_j},\bar{a}_{\sigma^{h} e_{j}}) \\
& = & |\sum_{i=0}^{j-1}2p_i + h+1 -\sum_{i=0}^{j-1}2p_i - h-1| = 0 \ \ \textrm { and }\\
d_x(\hat{B}_{\sigma^{h} e_j},\bar{c}_{\sigma^{g} e_{k}}) &= &
d_x(\hat{B}_{\sigma^{h} e_j},\bar{c}_{\sigma^{h} e_{j}}) \\
& = & |\sum_{i=0}^{j-1}2p_i + h+1 -\sum_{i=0}^{j-1}2p_i - (h+1)-0.3| = 0.3.
\end{eqnarray*}
For $k = j$ and $g \notin \{h,p_j-1\}$, we have
\begin{eqnarray*}
d_x(\hat{B}_{\sigma^{h} e_j},\bar{a}_{\sigma^{g} e_{k}}) & = &
d_x(\hat{B}_{\sigma^{h} e_j},\bar{a}_{\sigma^{g} e_{j}}) \\
& = & |\sum_{i=0}^{j-1}2p_i + h+1 -\sum_{i=0}^{j-1}2p_i - g-1| = |h-g| \geq 1\ \ \textrm{ and }\\
d_x(\hat{B}_{\sigma^{h} e_j},\bar{c}_{\sigma^{g} e_{k}}) & = &
d_x(\hat{B}_{\sigma^{h} e_j},\bar{c}_{\sigma^{g} e_{j}}) \\
& = & |\sum_{i=0}^{j-1}2p_i + h+1 -\sum_{i=0}^{j-1}2p_i - g-1-0.3| \geq ||h-g|-0.3| \geq 0.7.
\end{eqnarray*}
For $k = j$ and $g = p_j-1$, we have
\begin{eqnarray*}
d_x(\hat{B}_{\sigma^{h} e_j},\bar{a}_{\sigma^{g} e_{k}}) & = &
d_x(\hat{B}_{\sigma^{h} e_j},\bar{a}_{\sigma^{p_j-1} e_{j}}) \\
& = & |\sum_{i=0}^{j-1}2p_i + h+1 -\sum_{i=0}^{j-1}2p_i - (p_j-1)-1|\\
& = & |h-(p_j-1)| \geq 1 \;\;\;(\textrm{because } h \neq p_j-1) \ \ \textrm{ and}\\
d_x(\hat{B}_{\sigma^{h} e_j},\bar{c}_{\sigma^{g} e_{k}}) & = &
d_x(\hat{B}_{\sigma^{h} e_j},\bar{c}_{\sigma^{p_j-1} e_{j}}) \ = \
d_x(\hat{B}_{\sigma^{h} e_j},\bar{c}_{\sigma^{-1} e_{j}})\\
& = &
|\sum_{i=0}^{j-1}2p_i + h+1 -\sum_{i=0}^{j-1}2p_i -0.3| \geq |h+0.7| \geq 0.7.\\
\end{eqnarray*}

\noindent For $k > j$ we have
\begin{eqnarray*}
d_x(\hat{B}_{\sigma^{h} e_j},\bar{a}_{\sigma^{g} e_{k}}) & = &
|\sum_{i=0}^{j-1}2p_i + h+1 -\sum_{i=0}^{k-1}2p_i - g-1| \\
& = & |\sum_{i=j}^{k-1}2p_i + g - h| \\
& \geq & |2p_j + g -(p_j - 1)| \geq |p_j + 1 + g| \geq 2 \ \ \textrm{ and } \\
d_x(\hat{B}_{\sigma^{h} e_j},\bar{c}_{\sigma^{g-1} e_{k}}) & = &
|\sum_{i=0}^{j-1}2p_i + h+1 -\sum_{i=0}^{k-1}2p_i - g-0.3| \\
& = & |\sum_{i=j}^{k-1}2p_i + g - h-0.7| \\
& \geq & |2p_j + g -(p_j - 1)-0.7| \geq |p_j + 1 + g-0.7| \geq 1.3.
\end{eqnarray*}

\noindent For $k < j$ we have
\begin{eqnarray*}
d_x(\hat{B}_{\sigma^{h} e_j},\bar{a}_{\sigma^{g} e_{k}}) & = &
|\sum_{i=0}^{j-1}2p_i + h+1 -\sum_{i=0}^{k-1}2p_i - g-1| = |\sum_{i=k}^{j-1}2p_i + h - g| \\
& \geq & |2p_k + h -(p_k - 1)| \geq |p_k + 1 + h| \geq 2 \ \ \textrm{ and} \\
d_x(\hat{B}_{\sigma^{h} e_j},\bar{c}_{\sigma^{g-1} e_{k}}) & = &
|\sum_{i=0}^{j-1}2p_i + h+1 -\sum_{i=0}^{k-1}2p_i - g-0.3| = |\sum_{i=k}^{j-1}2p_i + h +0.7 - g| \\
& \geq & |2p_k + h +0.7 -(p_k - 1)| \geq |p_k + 1.7 + h| \geq 2.7.
\end{eqnarray*}
It follows from the comparison that $B_{\sigma^{h} e_j}$ prefers
$a_{\sigma^{h} e_j}$ over $c_{\sigma^{h} e_j}$ and $c_{\sigma^{h}
e_j}$ over any other $a$-woman and the $c$-woman. Hence, the initial
part of the preference list of $B_{\sigma^{h} e_j}$ reads
$$b_{\rho\sigma^{(p_l-1)}e_l}b_{\rho\sigma^{(p_l-2)}e_l}\cdots
b_{\rho e_l} b_{\rho\sigma^{(p_{l-1}-1)}e_{l-1}}\cdots b_{\rho
e_{l-1}}\cdots b_{\rho\sigma^{(p_1-1)}e_1}\cdots b_{\rho e_1}
a_{\sigma^{h} e_j}c_{\sigma^{h} e_j}.$$

Case(ii) $h = p_i-1,\;$ for some $i \in \{1,2,\cdots,l\}$: Now we
compare the distances between $\hat{B}_{\sigma^{h} e_j}$ and the
$a$-women and the $c$-women.
$$
e_j,e_{k} \in Rep(\sigma), \ \ h=p_j-1, \ \ 0 \leq g \leq p_{k}-1,
$$
For $k = j$ and $g = h$, we have
\begin{eqnarray*}
d_x(\hat{B}_{\sigma^{h} e_j},\bar{a}_{\sigma^{g} e_{k}}) & = &
d_x(\hat{B}_{\sigma^{p_j-1} e_j},\bar{a}_{\sigma^{p_j-1} e_{j}}) \\
& = & |\sum_{i=0}^{j-1}2p_i + h+1 -\sum_{i=0}^{j-1}2p_i - h-1| = 0 \ \ \textrm{ and}\\
d_x(\hat{B}_{\sigma^{h} e_j},\bar{c}_{\sigma^{g} e_{k}}) & = &
d_x(\hat{B}_{\sigma^{p_j-1} e_j},\bar{c}_{\sigma^{p_j-1} e_{j}}) =
d_x(\hat{B}_{\sigma^{p_j-1} e_j},\bar{c}_{\sigma^{-1} e_{j}}) \\
& = & |\sum_{i=0}^{j-1}2p_i + p_j -\sum_{i=0}^{j-1}2p_i -0.3| = p_j - 0.3.\\
\end{eqnarray*}
For $k > j$ we have
\begin{eqnarray*}
d_x(\hat{B}_{\sigma^{p_j-1} e_j},\bar{a}_{\sigma^{g} e_{k}}) & = &
|\sum_{i=0}^{j-1}2p_i + p_j -\sum_{i=0}^{k-1}2p_i - g-1| \\
& = & |\sum_{i=j}^{k-1}2p_i + g + 1 - p_j| \ \geq \  |2p_j + g +1 -p_j | \\
& \geq & |p_j + 1 + g| \geq p_j + 1 \ \ \textrm{ and} \\
d_x(\hat{B}_{\sigma^{h} e_j},\bar{c}_{\sigma^{g-1} e_{k}}) & = &
|\sum_{i=0}^{j-1}2p_i + p_j -\sum_{i=0}^{k-1}2p_i - g-0.3| \\
& = & |\sum_{i=j}^{k-1}2p_i + g +0.3 - p_j| \ \geq \  |2p_j + g +0.3 -p_j| \\
& \geq & |p_j + 0.3 + g| \geq p_j +0.3.\\
\end{eqnarray*}
For $k < j$ we have
\begin{eqnarray*}
d_x(\hat{B}_{\sigma^{p_j-1} e_j},\bar{a}_{\sigma^{g} e_{k}}) & = &
|\sum_{i=0}^{j-1}2p_i + p_j -\sum_{i=0}^{k-1}2p_i - g-1| \\
& = & |\sum_{i=k}^{j-1}2p_i + p_j - g-1| \ \geq \ |2p_k + p_j -p_k | \\
& \geq & |p_k + p_j| \geq p_j + 1 \ \ \textrm{ and} \\
d_x(\hat{B}_{\sigma^{h} e_j},\bar{c}_{\sigma^{g-1} e_{k}}) & = &
|\sum_{i=0}^{j-1}2p_i + p_j -\sum_{i=0}^{k-1}2p_i - g-0.3| \\
& = & |\sum_{i=k}^{j-1}2p_i + p_j  - g-0.3| \ \geq \ |2p_k + p_j -(p_k - 1) -0.3| \\
& \geq & |p_k + p_j + 0.7| \geq p_j + 1.7.\\
\end{eqnarray*}
From the above inequalities,  it follows that $B_{\sigma^{h} e_j}$
prefers $a_{\sigma^{p_j-1} e_{j}}$ over $c_{\sigma^{-1} e_{j}}$,
and $c_{\sigma^{-1} e_{j}}$ over $a$-women and $c$-women whose
subscript belongs to $\sigma$ cycles different from that of
$c_{\sigma^{-1} e_{j}}$'s. Now we compute and compare distances from
$\hat{B}_{\sigma^{h} e_j}$ to all the $a$-women and $c$-women whose
subscript is on the same $\sigma$ cycle as $a_{\sigma^{p_j-1}
e_{j}}$'s and $c_{\sigma^{-1} e_{j}}$'s.

\noindent For $0\leq g \leq p_j-2$ we have
\begin{eqnarray*}
& & \hspace*{-.65in} d_x(\hat{B}_{\sigma^{p_j-1} e_j},\bar{c}_{\sigma^{g}
e_{j}}) - d_x(\hat{B}_{\sigma^{p_j-1} e_j},\bar{a}_{\sigma^{g}
e_{j}})\\
& = & |\sum_{i=0}^{j-1}2p_i + p_j -\sum_{i=0}^{j-1}2p_i -
g-1-0.3| -
|\sum_{i=0}^{j-1}2p_i + p_j -\sum_{i=0}^{j-1}2p_i - g-1| \\
& = & p_j-g-1.3 - (p_j-g-1) < 0
\end{eqnarray*}
and for $0\leq g \leq p_j-1$ we have
\begin{eqnarray*}
& & \hspace*{-.80in} d_x(\hat{B}_{\sigma^{p_j-1}
e_j},\bar{a}_{\sigma^{g} e_{j}})-d_x(\hat{B}_{\sigma^{p_j-1}
e_j},\bar{c}_{\sigma^{g-1}
e_{j}})\\
&\hspace*{-.25in} = &|\sum_{i=0}^{j-1}2p_i + p_j
-\sum_{i=0}^{j-1}2p_i - g-1| -
|\sum_{i=0}^{j-1}2p_i + p_j -\sum_{i=0}^{j-1}2p_i - g-0.3| \\
& \hspace*{-.25in} = & p_j-g-1 - (p_j-g-0.3) < 0.
\end{eqnarray*}
From the above comparisons, it follows that $B_{\sigma^{p_j-1} e_j}$
prefers $a_{\sigma^{p_j-1} e_{j}}$ over $c_{\sigma^{p_j-2} e_{j}}$,
$c_{\sigma^{g} e_{j}}$ over $a_{\sigma^{g} e_{j}}$ and
$a_{\sigma^{g} e_{j}}$ over $c_{\sigma^{g-1} e_{j}}$ for $0\leq g
\leq p_j-2$. Stringing these preferences together, we obtain a
portion of $B_{\sigma^{p_j-1} e_j}$'s preference list which appears
as
$$a_{\sigma^{p_j-1} e_{j}}c_{\sigma^{p_j-2} e_{j}}a_{\sigma^{p_j-2} e_{j}}c_{\sigma^{p_j-3}
e_{j}}\cdots c_{\sigma e_{j}}a_{\sigma
e_{j}}c_{e_{j}}a_{e_{j}}c_{\sigma^{-1} e_{j}}(= c_{\sigma^{p_j-1}
e_{j}}).$$

 We remind the reader that
$B_{\sigma^{p_j-1} e_j}$'s preference list has all the $b$-women
appearing at the front appended by the above list of $a$-women and
$c$-women. Hence, the initial part of $B_{\sigma^{p_j-1} e_j}$'s
preference list is
\begin{eqnarray*}
b_{\rho\sigma^{(p_l-1)}e_l}\cdots b_{\rho
e_l}b_{\rho\sigma^{(p_{l-1}-1)}e_{l-1}}\cdots b_{\rho
e_{l-1}}\cdots b_{\rho\sigma^{(p_1-1)}e_1}\cdots b_{\rho e_1} \\
a_{\sigma^{(p_j-1)} e_{j}}c_{\sigma^{(p_j-2)}
e_{j}}a_{\sigma^{(p_j-2)} e_{j}}\cdots
c_{e_{j}}a_{e_{j}}c_{\sigma^{-1} e_{j}}(= c_{\sigma^{(p_j-1)}
e_{j}}).
\end{eqnarray*}

The initial part of the preference lists of men $A_{\sigma^s e_i},
C_{\sigma^{s-1} e_i}$ and $B_{\sigma^s e_i}$ are as follows.
\begin{align*}
e_i&\in Rep(\sigma)\;\;,\\
A_{\sigma^m e_i} \;\;&:\;\;\; a_{\sigma^m e_i} b_{\rho \sigma^m e_i}\;\;,\;\;\;0\leq m \leq p_i-1\\
C_{\sigma^{(m-1)} e_i} \;\;&:\;\;\; c_{\sigma^{(m-1)} e_i} a_{\sigma^{m} e_i}\;\;,\;\;\;0\leq m \leq p_i-1\\
B_{\sigma^m e_i} \;\;&:\;\; b_{\rho\sigma^{(p_l-1)}e_l}\cdots
b_{\rho e_l}b_{\rho\sigma^{(p_{l-1}-1)}e_{l-1}}\cdots b_{\rho
e_{l-1}}\cdots\\
&\;\;\;\;\;b_{\rho\sigma^{(p_1-1)}e_1}\cdots b_{\rho e_1}
a_{\sigma^m e_i} c_{\sigma^m e_i}\;\;,\;\;\;0\leq m \leq p_i-2\\
B_{\sigma^{(p_i-1)} e_i} \;\;&:\;\;
b_{\rho\sigma^{(p_l-1)}e_l}\cdots b_{\rho
e_l}b_{\rho\sigma^{(p_{l-1}-1)}e_{l-1}}\cdots b_{\rho e_{l-1}}\cdots
b_{\rho\sigma^{(p_1-1)}e_1}\cdots b_{\rho e_1} a_{\sigma^{(p_i-1)}}
c_{\sigma^{(p_i-2)} e_i}\\
&\;\;\;\;\;a_{\sigma^{(p_i-2)} e_i}\cdots a_{\sigma e_i}c_{e_i}
a_{e_i}c_{\sigma^{(p_i-1)} e_i}
\end{align*}
Note that these are exactly the same as
those in~(\ref{men}) for any appropriate value of the permutation~$\tau$.
In a similar manner, we can obtain preference lists for the women.
The preference lists for the women are as follows.
\begin{align*}
f_i&\in Rep(\rho)\;\;,\\
b_{\rho^m f_i} \;\;&:\;\;\; A_{\rho^{(m-1)} f_i} B_{\rho^m f_i}\;\;,\;\;\;0\leq m \leq q_i-1\\
c_{\rho^{m} f_i} \;\;&:\;\;\; B_{\rho^{m} f_i} C_{\rho^{m} f_i}\;\;,\;\;\;0\leq m \leq q_i-1\\
a_{\rho^m f_i} \;\;&:\;\; C_{\rho^{(q_k-1)}f_k}\cdots C_{ f_k}
C_{\rho^{(q_{k-1}-1)}f_{k-1}}\cdots C_{f_{k-1}} \cdots\\
&\;\;\;\;\;C_{\rho^{(q_1-1)}f_1}\cdots C_{ f_1}
B_{\rho^m f_i} A_{\rho^m f_i}\;\;,\;\;\;0\leq m \leq q_i-2\\
a_{\rho^{(q_i-1)} f_i} \;\;&:\;\;
C_{\rho^{(q_k-1)}f_k}\cdots C_{ f_k}
C_{\rho^{(q_{k-1}-1)}f_{k-1}}\cdots C_{f_{k-1}}
B_{\rho^{(q_i-1)}f_i}
A_{\rho^{(q_i-2)} f_i}\\
&\;\;\;\;\;B_{\rho^{(q_i-2)} f_i}\cdots B_{\rho f_i}A_{f_i}
B_{f_i}A_{\rho^{(q_i-1)} f_i}
\end{align*}

Now note that, by the construction of the $\rho$ cycles, which go in
order from~$1$ to~$n$,
the list
$$C_{\rho^{(q_k-1)}f_k}\cdots C_{ f_k}
C_{\rho^{(q_{k-1}-1)}f_{k-1}}\cdots C_{f_{k-1}} \cdots
C_{\rho^{(q_1-1)}f_1}\cdots C_{ f_1}$$ is identically $C_n \cdots
C_1$. Thus, the preference lists for the women are identical to
those given in~(\ref{women}). Thus, the rest of the proof is exactly
the same as in the $3$-attribute case, starting from the
introduction of the men's lists (\ref{men}) in
Section~\ref{sect:mfo}.

\end{document}